\documentclass[reqno,11pt,letterpaper]{amsart}
\usepackage{amsmath,amssymb,amsthm,graphicx,mathrsfs,url}
\usepackage[usenames,dvipsnames]{color}
\usepackage[colorlinks=true,linkcolor=red,citecolor=Green]{hyperref}

\def\?[#1]{\textbf{[#1]}\marginpar{\Large{\textbf{??}}}}

\newtheorem{thm}{Theorem}
\newtheorem{prop}{Proposition}[section]
\newtheorem{defi}[prop]{Definition}
\newtheorem{prob}[prop]{Problem}

\newtheorem{ex}{Example}
\newtheorem{lemm}[prop]{Lemma}
\newtheorem{corr}[prop]{Corollary}

\numberwithin{equation}{section}

\newcommand{\nwc}{\newcommand}
\nwc{\ep}{\epsilon}
\nwc{\vareps}{\varepsilon}
\nwc{\Oph}{\operatorname{Op}_\hbar}
\nwc{\ra}{\rangle}

\nwc{\mf}{\mathbf} 
\nwc{\blds}{\boldsymbol} 
\nwc{\ml}{\mathcal} 

\nwc{\defeq}{\stackrel{\rm{def}}{=}}

\nwc{\cE}{\ml{E}}
\nwc{\cN}{\ml{N}}
\nwc{\cO}{\ml{O}}
\nwc{\cP}{\ml{P}}
\nwc{\cU}{\ml{U}}
\nwc{\cV}{\ml{V}}
\nwc{\cW}{\ml{W}}
\nwc{\tU}{\widetilde{U}}
\nwc{\IN}{\mathbb{N}}
\nwc{\IR}{\mathbb{R}}
\nwc{\IZ}{\mathbb{Z}}
\nwc{\IC}{\mathbb{C}}
\nwc{\IT}{\mathbb{T}}
\nwc{\IS}{\mathbb{S}}
\nwc{\tP}{\widetilde{P}}
\nwc{\tPi}{\widetilde{\Pi}}
\nwc{\tV}{\widetilde{V}}
\nwc{\rest}{\restriction}

\title[Renormalization of determinant lines in QFT.]{Renormalization of determinant lines in Quantum Field Theory.}
\author{Nguyen Viet Dang}
\date{}
\begin{document}

\maketitle

\begin{abstract}
On a compact manifold $M$, 
we consider the affine space $\mathcal{A}$ of non-self-adjoint perturbations of some invertible
elliptic operator acting on sections of some Hermitian bundle, 
by some differential operator of lower order. 
 
 We construct and classify all complex analytic functions on the 
Fr\'echet space $\mathcal{A}$ 
vanishing exactly over non-invertible elements, 
having minimal growth at infinity along complex rays in $\mathcal{A}$ and which are
obtained by local renormalization, a concept coming from 
quantum field theory, called \emph{renormalized determinants}.
The additive group of 
local polynomial functionals of finite degrees
acts freely and transitively on the space of renormalized determinants.
We provide  
different representations of the renormalized determinants 
in terms of spectral zeta determinants, Gaussian Free Fields, 
infinite product and renormalized Feynman amplitudes in perturbation theory
in position space \`a la Epstein--Glaser.

 Specializing to the case of Dirac operators 
coupled to vector potentials
and reformulating our results in terms of determinant line bundles, 
we prove our renormalized determinants define some complex analytic trivializations of
some holomorphic line bundle over $\mathcal{A}$.  
This relates our results to a conjectural picture from some unpublished notes by Quillen~\cite{Quillen89}
from April 1989.   
\end{abstract}

\section{Introduction.}

Let $(M,g)$ be a smooth, closed, compact 
Riemannian manifold.
The aim of the present paper is to 
study the analytical properties and
the renormalization of a class of functional determinants defined on some affine space
$\mathcal{A}$ of 
non self-adjoint operators that we divide in two classes: 
\begin{itemize}
\item in the first class, we consider perturbations of the form $\Delta+V$ of some given invertible, 
self-adjoint, generalized Laplacian $\Delta$ acting on some fixed
Hermitian bundle $E\mapsto M$ where $V\in C^\infty(End(E))$ is a smooth potential,
\item in the second class, we look at perturbations of the form 
$D+A$ of some invertible twisted
Dirac operator $D:C^\infty(E_+)\mapsto C^\infty(E_-) $ acting between Hermitian 
bundles $E_{\pm}$ by
some term $A\in C^\infty(Hom(E_+,E_-))$. By invertible, we mean the Fredholm index of $D$ equals $0$ and $\ker(D)=\{0\}$.
\end{itemize}
 We consider lower order perturbations  
since $A$ and $V$ are local operators of order $0$.

\subsubsection{Quantum field theory interacting with some external potential.}

Let us briefly give the physical motivations
underlying our results which are stated in purely mathematical terms. 
The reader uninterested by the physics can safely skip this part.  
Inspired by recent works in mathematical physics~\cite{derezinski2017bosonic, 
derezinski2014quantum, derezinski2002renormalization, derezinski2008quadratic, Fulling2018, Fulling2016} 
and classical works of Schwinger~\cite{schwinger5}~\cite[Chapter 4 p.~163]{Itzykson}, 
our original purpose is to understand
the problem of renormalization 
of some Euclidean quantum field $\phi$ 
defined on $M$ interacting 
with a classical external field which is not quantized~\footnote{sometimes called \emph{background field} in the physical litterature}. 
For instance, consider the Laplace--Beltrami operator $\Delta: C^\infty(M)\mapsto C^\infty(M)$ defined from the metric $g$ on $M$, corresponding to the Dirichlet
action functional:
\begin{equation}
S(\phi)=\int_M \left(\phi\Delta\phi\right) dv
\end{equation}
and their
perturbations by some external potential $V\in C^\infty(M,\mathbb{R}_{\geqslant 0})$
which corresponds to the perturbed Dirichlet
action functional
\begin{equation}
S(\phi)=\int_M\left( \phi\Delta\phi +V\phi^2\right) dv.
\end{equation}

A classical problem in quantum field theory is
to define the \emph{partition function} of a theory, usually represented 
by some ill--defined functional integral.
In the bosonic case, it reads:
\begin{equation}
Z\left(V\right)=\int [d\phi]\exp\left(-\frac{1}{2}\int_M \left( \phi\Delta\phi +V\phi^2 \right)dv\right) 
\end{equation}
where
$V\in C^\infty(M, \mathbb{R}_{>0})$ plays the role of a position dependent mass which is viewed as an external field coupled 
to the 
Gaussian Free Field $\phi$ to be quantized.
The external field can also be the metric $g$~\cite{Bastianelli} in the study of \emph{gravitational anomalies} 
in the 
physics litterature or \emph{gauge fields}, which is the physical terminology for
connection $1$--forms, in the study of \emph{chiral anomalies}.

 In fact, according to Stora~\cite{Storagauge, Perrot-07}, the physics of chiral anomalies~\cite{Fujikawa, Perrot-07, Bastianelli}
can be understood in the case where we have a quantized fermion field interacting with some 
\emph{gauge field} which is treated
as an external field. Consider 
the quadratic Lagrangian 
$\Psi_-D_A\Psi_+$ where $\Psi_{\pm}$ are \emph{chiral} fermions,
$D_A$ is a twisted half--Dirac operator acting from sections of
positive spinors $S_+$, to negative spinors $S_-$, see example~\ref{ex:chiralanomaliessinger}
below for a precise definition of $D_A$.
In this case, the corresponding ill--defined functional 
integral reads
\begin{eqnarray*}
Z\left(A\right)=\int [\mathcal{D}\Psi_-\mathcal{D}\Psi_+] \exp\left(\int_M\left\langle\Psi_-,D_A\Psi_+\right\rangle \right)
\end{eqnarray*}
where we are interested on the dependence in the gauge field $A$.

\subsubsection{Functional determinants in geometric analysis.}

 The above two problems can be formulated 
as the mathematical problem of constructing functional 
determinants $V\mapsto  \det\left(\Delta+V \right)$ and $A\mapsto \det\left(D^*\left(D+A\right) \right)$
with nice functional properties
where we are interested in the dependence in the external potentials $V$ and $A$.
In global analysis, functional determinants also 
appear
in the study of the analytic torsion by Ray--Singer~\cite{ray1971} and more generally as 
metrics of determinant line bundles as initiated by Quillen~\cite{quillen1985determinant, BismutFreed} where 
he considered some affine space $\mathcal{A}$ 
of Cauchy-Riemann operators $D+\omega$ acting on some fixed
vector bundle $E\mapsto X$ over a compact Riemann surface $X$ where $D$ is fixed and 
the
perturbation $\omega$ lives in 
the linear space of $(0,1)$-forms on $X$ with values in the
bundle $End(E)$. The metrics on $X$ and in $E$ induce a metric in the
determinant (holomorphic line) bundle
$ \det\left(\text{Ind}(D) \right)= \Lambda^{top} \ker(D)^*\otimes \Lambda^{top} \text{coker}(D)$
over $\mathcal{A}$. As Quillen showed in~\cite{quillen1985determinant}, 
if this metric in the bundle $\det(\text{Ind}(D))$ is
divided by the function $\det_\zeta\left( D^*D \right)$ 
(here $\det_\zeta$ is the zeta regularized determinant
of the Laplacian $D^*D$),
then the canonical curvature form of this metric, the first Chern form,
coincides with the symplectic form of the natural K\"ahler metric on $\mathcal{A}$.
An important consequence of the above observation, stated as a~\cite[Corollary p.~33]{quillen1985determinant},
is that if one multiplies the Hermitian metric on $ \det\left(\text{Ind}(D) \right)$ by $e^q$ where $q$ is the natural
K\"ahler metric on 
$\mathcal{A}$, then the corresponding 
Chern connection is \textbf{flat}. From the contractibility of $\mathcal{A}$, one deduces the
existence of a \textbf{global holomorphic trivialization} of  $ \det\left(\text{Ind}(D) \right)\mapsto \mathcal{A}$ and 
the image of the canonical section of $ \det\left(\text{Ind}(D) \right)$ by this trivialization 
is an analytic function on $\mathcal{A}$ vanishing over non--invertible elements.

Building on some ideas from the work of Perrot~\cite{Perrot-07,Perrothdr,Perrotrenorm} and some
unpublished notes from Quillen's notebook~\cite{Quillen89}~\footnote{made available by the Clay foundation at 
\url{http://www.claymath.org/library/Quillen/Working_papers/quillen 1989/1989-2.pdf} }, 
we attempt to relate the problem of constructing
renormalized determinants with the construction of holomorphic trivializations of determinant line bundles
over some affine space $\mathcal{A}$ of perturbations of some fixed operator by some differential operator of lower order 
which plays the role of the external potential.

\subsubsection{Quillen's conjectural picture.}
\label{ss:quillenpict}
In some notes on the 30th of April 1989~\cite[p.~282]{Quillen89},
with the motivation to make sense of the technique of adding local counterterms to the Lagrangian 
used in renormalized perturbation theory, 
Quillen proposed to give an
interpretation of QFT partition functions 
in terms of determinant line bundles over the space 
of Dirac operator coupled 
to a gauge potential drawing a direct connection between the two subjects.
The approach he outlined
insists on constructing \emph{complex analytic trivializations} of the
determinant line bundle without mentioning 
any construction of 
Hermitian metrics on the line bundle
which seems different from the original approach 
he pioneered~\cite{quillen1985determinant} and
the Bismut--Freed~\cite{BismutFreed} definition
of determinant line bundle for families of Dirac operators.

To explain this connection, we recall that
for a pair $\mathcal{H}_0,\mathcal{H}_1$ of complex Hilbert spaces,  
there is a \textbf{canonical holomorphic line bundle}
$\mathbf{Det} \longmapsto \text{Fred}_0\left(\mathcal{H}_0,\mathcal{H}_1 \right)$ where $\text{Fred}_0\left(\mathcal{H}_0,\mathcal{H}_1 \right)$ is the space of Fredholm operators of \textbf{index} $0$ with fiber $\mathbf{Det}_B\simeq \Lambda^{top} \ker\left( B\right)^*\otimes \Lambda^{top} \text{coker}\left( B\right)$ and canonical section $\sigma$~\cite[p.~32]{quillen1985determinant}~\cite[p.~137--138]{segal2004definition}. 
Consider the \textbf{complex affine space} $\mathcal{A}=D+C^\infty(Hom(E_+,E_-))$
of perturbations of some fixed invertible Dirac operator $D$ by some 
\textbf{differential operator} $A\in C^\infty(Hom(E_+,E_-))$ of order $0$. We denote by 
$L^2(E_+)$
the space of $L^2$ sections of $E_+$.
Then the map
$\iota: D+A\in \mathcal{A}\mapsto Id+D^{-1}A\in \text{Fred}_0\left(L^2(E_+),L^2(E_+)\right) $
allows to pull--back the holomorphic line bundle $\mathbf{Det} $ as a holomorphic line bundle
$\mathcal{L}=\iota^*\mathbf{Det}\mapsto \mathcal{A}$ over the affine space $\mathcal{A}$ with canonical section $\underline{\det}=\iota^*\sigma$.
We insist that 
we view $C^\infty(Hom(E_+,E_-))$ as a $\mathbb{C}$--vector space, elements in 
$C^\infty(Hom(E_+,E_-))$ need not preserve Hermitian structures. 
According to Quillen~\cite[p.~282]{Quillen89}, the relation with QFT goes as follows, 
\emph{one gives a meaning to the functional integrals}
\begin{eqnarray}
A&\mapsto &\int \mathcal{D}\Psi_+\mathcal{D}\Psi_- e^{\int_M \left\langle\Psi_-,D_A\Psi_+ \right\rangle },
\end{eqnarray}\footnote{$D_A$ is the Dirac operator coupled to the gauge potential $A$ as described in~\cite[section 3 p.~325]{singer1985families}} 
\emph{by trivializing the determinant line }. In other words, denoting by $\mathcal{O}(\mathcal{L})$ (resp $\mathcal{O}\left(\mathcal{A}\right)$) the holomorphic sections (resp functions)
of $\mathcal{L}$ (resp on $\mathcal{A}$),
we aim at constructing a \textbf{holomorphic trivialization} of the line
bundle $\tau:\mathcal{O}\left(\mathcal{L}\right)\longmapsto  \mathcal{O}\left(\mathcal{A}\right)$
so that the \textbf{image $\tau(\underline{\det})$} of the canonical section $\underline{\det}$ by this trivialization is an entire function 
$f(P+\mathcal{V})$ on 
$\mathcal{A}$
vanishing exactly over the set $Z$ of non-invertible elements of $\mathcal{A}$.
In some sense, this should generalize the original
construction of Quillen of the holomorphic trivialization
of the determinant line bundle over the space of Cauchy--Riemann operators~\cite{quillen1985determinant}.
Furthermore, Quillen~\cite[p.~284]{Quillen89} writes:
\\
\\
\fbox{
\begin{minipage}{0.94\textwidth}
\textbf{These considerations lead to the following conjectural picture. Over the space $\mathcal{A}$ of gauge fields there should be a principal bundle
for the additive group of polynomial functions of degree $\leqslant d$ where $d$ bounds the trace which have to be regularized. The idea is that near each $A\in \mathcal{A}$ we should have a well--defined trivialization of $\mathcal{L}$ up to $\exp$ of such a polynomial. Moreover, we should have a flat connection on this bundle. } 
\end{minipage}
}
\\
\\
To address this conjectural picture, we follow a backward 
path compared to~\cite{quillen1985determinant}. Instead of constructing some Hermitian metric
then a flat connection on $\mathcal{L}\mapsto \mathcal{A}$ to trivialize the bundle,   
we prove in Theorem~\ref{t:quillenconjanal} an infinite dimensional analog of the
classical 
Hadamard factorization Theorem~\ref{t:hadamard} in complex analysis.
We classify all determinant-like functions such that:
\begin{itemize}
\item They are entire functions on $\mathcal{A}$ with minimal growth at infinity, a concept with is defined below as the \emph{order} of the entire function,
vanishing over non-invertible elements in $\mathcal{A}$.
\item Their differentials should satisfy some simple identities 
reminiscent of the situation for the usual 
determinant in finite dimension.
\item They are
obtained from a renormalization
by subtraction of local counterterms, a concept coming from quantum field theory 
which will be explained below in paragraph~\ref{s:subtractloc}. 
\end{itemize}
Trivializations of $\mathcal{L}$ are simply obtained 
by dividing the canonical section of $\mathcal{L}$
by the constructed determinant-like functions as showed in Theorem~\ref{t:quillenholotriv}. 
A nice consequence of our investigation
is a new factorization formula 
for zeta regularized determinant~(\ref{e:bosonfacthadamard}),(\ref{e:fermionfacthadamard}) 
in terms of Gohberg--Krein's regularized determinants. 
We show that our renormalized determinants  
are not canonical and there are some ambiguities involved
in their construction of the form $\exp\left(P\right)$ where $P$ is a local polynomial functional of $A$. Then 
we show that the additive 
group of local polynomial functionals of $A$, sometimes called 
the renormalization group of St\"ueckelberg--Petermann in the physics litterature,
acts freely and transitively on the space of renormalized determinants we construct.

{
\small{
\subsection*{Acknowledgements.}

 We warmly thank 
C.~Brouder, Y.~Chaubet, Y.~Colin de Verdi\`ere, C.~Dappiaggi, J.~Derezi\'nski, K.~Gaw\c{e}dzki, C.~Guillarmou, A.~Karlsson, J.~Kellendonk,
M.~Puchol,  G.~Rivi\`ere,  S.~Paycha, K.~Rejzner, 
M.~Wrochna, B.~Zhang  
for many interesting discussions, questions, advices, remarks that helped us improve our work.
Particular thanks are due to I.~Bailleul for his genuine interest 
and careful reading of the manuscript.
Also we thank D.~Perrot for many discussions on anomalies, determinants and Epstein--Glaser renormalization and
the Clay foundation for making available the scanned notes
of Quillen's notebook where we could find an incredible source of inspiration.   
We also would like to thank our wife Tho for excellent working conditions at home and a lot of positive motivation
for us.
}
}

\subsection{Notations.}
\label{ss:notations}

$dv$ is used for a smooth density
$\vert \Lambda^{top}\vert M$ on $M$. 
In the sequel, for every pair $(B_1,B_2)$ of Banach spaces,
$\mathcal{B}(B_1,B_2)$ denotes
the Banach space of bounded operators from $B_1\mapsto B_2$ 
endowed with the norm $\Vert .\Vert_{\mathcal{B}(B_1,B_2)}$.
For any vector bundle $E$ on $M$,
we denote by $\Psi^{\bullet}(M,E)$ the algebra of pseudodifferential operators on
the manifold $M$ acting on 
sections of the bundle $E$
and when there is no ambiguity we
will sometimes
use the short notation $\Psi(M)$. $C^k(E),\Vert .\Vert_{C^k(E)},k\in \mathbb{N}$ denotes
the Banach space of continuous sections of $E$ of regularity $C^k$,
$H^s(E), s\in \mathbb{R}$ denotes Sobolev sections of $E$ endowed with the norm $\Vert.\Vert_{H^s(E)}$ that we shall sometimes write $H^s,\Vert.\Vert_{H ^s}$
for simplicity\footnote{ when $s<0$ these are distributional sections}
and finally $L^p(E), 1\leqslant p\leqslant +\infty$ denotes $L^p$ sections of $E$ endowed with the norm $\Vert .\Vert_{L^p(E)}$.
$C^\infty_c(U,E)$ are smooth sections of $E$ compactly supported on $U$.

For any pair $(E,F)$ of bundles over $M$,
for $C^m$ Schwartz kernels $K$ of operators from $C^m(E)\mapsto C^m(F)$
which are elements of $C^m(M\times M, F\boxtimes E^*)$, 
we denote by $\Vert K \Vert_{C^m(M\times M)}$ their $C^m$ norm which is not to be confused
with the operator norm $\Vert K \Vert_{\mathcal{B}(C^m(E),C^m(F))}$.

For any Hilbert space $H$, we denote by
$\mathcal{I}_p\subset \mathcal{B}(H,H)$ the Schatten ideal 
of compact operators whose $p$-th power is trace class endowed with the norm $\Vert.\Vert_p$ defined as
$\Vert A\Vert_p=\sum_{\lambda\in \sigma(A)} \vert\lambda\vert^p $ where the sum runs over the singular values of $A$.

\section{Preliminary material.}

The goal of this section is to introduce enough material to state precisely 
our main results. We begin with some classical results on entire functions on $\mathbb{C}$ with given zeros.
Since we 
view functional determinants as 
infinite dimensional analogues of entire functions with given zeros, we need to 
recall classical results on holomorphic functions in Fr\'echet spaces. We conclude the introductory part by discussing 
Fredholm determinants and their generalizations by Gohberg--Krein which are also viewed as entire functions with given zeros on some infinite dimensional Banach or Fr\'echet spaces. This also serves as motivation for our main results.

\subsection{Entire functions with given zeros on $\mathbb{C}$.}
\label{ss:entireonevar}
In this paragraph, we recall some 
classical results on entire functions with given zeros.
The \textbf{order} $\rho(f)\geqslant 0$ of an entire function $f$
is the infimum of all the real numbers $\rho$ such that for some $A,K>0$, for all $z\in \mathbb{C}$
$\vert f(z)\vert\leqslant Ae^{K\vert z\vert^\rho}$.
The \textbf{critical exponents} of a sequence $\vert a_n\vert\rightarrow +\infty$,
is the infimum of all $\alpha>0$ such that $\sum_n \frac{1}{\vert a_n\vert^\alpha}<+\infty$.
Finally the \textbf{genus} of $f$ is the order of vanishing of $f$ at $z=0$.
The \textbf{divisor} of an entire function $f$ is 
the set of zeros of $f$ counted with multiplicity.
We recall a classical Theorem due to Hadamard on the structure
of entire functions with given zeros~\cite[p.~78--81]{remmert2012theory}, \cite[Thm 5.1 p.~147]{stein2003complex} (see also~\cite[p.~60]{mullencomplex}):
\begin{thm}[Hadamard's factorization Theorem]\label{t:hadamard}
Let $(a_n)_{n\in \mathbb{N}}$ be some sequence such that $\sum_{n} \vert a_n\vert^{-(p+1)}<+\infty $ but
$\sum_n \vert a_n\vert^{-p}=\infty$. Then any entire
function $f$ whose divisor $Z(f)=\{a_n|n\in \mathbb{N}\}$ has order
$\rho(f)\geqslant p$, and any entire function 
s.t. $Z(f)=\{a_n|n\in \mathbb{N}\}$ and $\rho(f)=p$
has a \textbf{representation} as:
\begin{eqnarray}
\boxed{ f(z)=z^me^{P(z)}\prod_{n=1}^\infty E_p\left(\frac{z}{a_n} \right)  }
\end{eqnarray}
where $P$ is a polynomial of degree $p$, $E_p(z)=(1-z)e^{z+\frac{z^2}{2}+\dots+\frac{z^p}{p}}$ is a Weierstrass factor of order $p$ and $m$ is the genus of $f$.
\end{thm} 
%
The 
lower bound on 
the order of $f$ follows from
Jensen's formula. 
Observe that the entire functions produced
by the Hadamard factorization
Theorem are
not unique due to the polynomial ambiguity which brings a factor $e^P$.
We will meet this ambiguity again in Theorem~\ref{t:quillenconjanal} 
which is responsible for the renormalization group.

\subsection{Entire functions on a complex Fr\'echet space.}

\subsubsection{Smooth functions on Fr\'echet spaces.}

In the present paper, we always work with Fr\'echet spaces
of smooth sections of finite rank vector bundles over some compact manifold $M$.
Smooth functions on Fr\'echet spaces 
will be understood in the sense of Bastiani~\cite[Def II.12]{brouder2018properties}, as popularized by Hamilton~\cite{Hamilton-82} in the context of Fr\'echet spaces and Milnor~\cite{milnor1984}. 
This means  
smooth functions are infinitely 
differentiable in the sense of G\^ateaux
and all the differentials $D^nF:U\times E^n\mapsto \mathbb{C} $ are jointly
continuous on 
$U\times E^n$~\cite[Def II.11]{brouder2018properties}.
We recall the
notion of G\^ateaux differentials and the correspondance between multilinear maps and distributional kernels since these will play a central role in 
our approach:
\begin{defi}[G\^ateaux differentials and Schwartz kernels of multilinear maps]\label{d:multilinearschwartz}
Let $B\mapsto M$ be some Hermitian vector bundle of finite rank on some smooth closed compact manifold 
$M$.
For a smooth function $f:V\in C^\infty(M,B)\mapsto f(V)\in \mathbb{C} $
where $C^\infty(M,B)$ is the Fr\'echet space of smooth sections,
the $n$-th differential
\begin{eqnarray}
D^nf(V,h_1,\dots,h_n)=\prod_{i=1}^n\frac{d}{dt_i} f(V+t_1h_1+\dots+t_nh_n)|_{t_1=\dots=t_n=0}
\end{eqnarray}
is multilinear continuous in $(h_1,\dots,h_n)$, hence it 
can be identified by the multilinear Schwartz kernel Theorem~\cite[lemm III.6]{brouder2018properties} 
with the unique distribution $[\mathbf{D^nf(V)}]$ in $\mathcal{D}^\prime(M^n,B^{\boxtimes n})$, called \textbf{Schwartz kernel of} $D^nf(V)$, s.t.
\begin{equation}
\left\langle [\mathbf{D^nf(V)}] , h_1\boxtimes \dots\boxtimes h_n  \right\rangle=D^nf(V,h_1,\dots,h_n)
\end{equation} 
is jointly continuous in $(V;h_1,\dots,h_n)\in C^\infty(M,B)^{n+1}$~\cite[Thm III.10]{brouder2018properties}.
\end{defi}
In the sequel, to stress the difference beetween the $n$-th differential $D^nf(V)$ of a function $f$ at $V$ 
from its Schwartz kernel, we use the notation $[\mathbf{D^nf(V)}]$ for the Schwartz kernel.  
In the physics litterature, G\^ateaux differentials of smooth functions on spaces of functions
are often called \textbf{functional derivatives}. These functional derivatives play an important role in classical and quantum field theory
and are usually represented (in fact identified) by their Schwartz kernels.

\subsubsection{Holomorphic functions on Fr\'echet spaces.}

First, let us define what we mean by an entire function on a Fr\'echet space. 
\begin{defi}[Holomorphic and entire functions on Fr\'echet spaces]\label{d:analfunfrechet}
Let $\Omega\subset E$ be some open subset in a Fr\'echet space $E$.
A function $F:\Omega \subset E\mapsto \mathbb{C}$ is holomorphic if it is \textbf{smooth}
and
for every 
$V_0\in \Omega$, the Taylor series of $F$ converges in some neighborhood of $V_0$
~\footnote{for the Fr\'echet topology} and $F$ coincides
with its Taylor series:
$$F(V_0+h)=\sum_{n=0}^\infty \frac{1}{n!} D^nF(V_0,h,\dots,h)$$
where the r.h.s. converges absolutely. 
In case $\Omega=E$, $F$ will be called \textbf{entire}.
\end{defi}

\subsection{Fredholm and Gohberg--Krein's determinants as entire functions vanishing over non-invertible elements.}

\subsubsection{Fredholm determinants.}

We briefly recall the definition of Fredholm determinant $\det_F(Id+B)$ for a trace class
operator $B:H\mapsto H$ acting on some separable Hilbert space $H$
and relate them with functional traces of powers of $B$.
These identities will 
imply that $\det_F$ is an example of entire function
on the infinite dimensional space $\left(Id+\text{trace class}\right)$ whose zeros are exactly the
non-invertible operators. 
\begin{defi}[Fredhom determinants]
The Fredholm determinant $\det_F(Id+B)$ is defined
in~\cite[equation (3.2) p.~32]{Simon-traceideals} 
as
\begin{equation}
\text{det}_F(Id+B)=\sum_{k=0}^\infty Tr(\Lambda^kB)
\end{equation}
where 
$\Lambda^kB:\Lambda^kH\mapsto \Lambda^kH$ acting on the $k$-th exterior power $\Lambda^kH$ is trace class.
Using 
the bound $\Vert \Lambda^kB\Vert_1\leqslant \frac{\Vert B\Vert_1}{k!}  $~\cite[Lemma 3.3 p.~33]{Simon-traceideals}, it is immediate
that $\det_F(Id+zB) $ is an \textbf{entire} function in $z\in \mathbb{C}$ (see also~\cite[Thm 2.1 p.~26]{gohberg2012traces}).
\end{defi}

For any compact operator $B$, we will denote by
$(\lambda_k(B))_k$ its eigenvalues 
counted with multiplicity.
By~\cite[Theorem 3.7]{Simon-traceideals}, the Fredholm determinant can be identified with a
Hadamard 
product and is related to the functional traces by the following sequence of identities:
\begin{eqnarray}\label{e:fredholmdet}
\text{det}_F\left(Id+zB \right)=\prod_k\left(1+z\lambda_k(B)\right)
=\underbrace{\exp\left(\sum_{m=1}^\infty (-1)^{m+1} z^m Tr_{L^2}(B^m) \right)}
\end{eqnarray}
where the term underbraced involving traces is well--defined only when $\vert z\vert \Vert B\Vert_1<1$.
Note the important fact that $\exp\left(\sum_{m=1}^\infty (-1)^{m+1} z^m Tr_{L^2}(B^m) \right)$ 
which is defined on the disc $\mathbb{D}=\{\vert z\vert \Vert B\Vert_1<1  \}$ has
analytic continuation as an entire function of $z\in \mathbb{C}$ and $B\mapsto \det_F\left(Id+B \right) $ 
is an entire function vanishing
when $Id+B$ is non-invertible.  

\subsubsection{Gohberg--Krein's determinants.}

Set $p\in \mathbb{N}$ and let $A$ belong to the Schatten ideal
$\mathcal{I}_p\subset \mathcal{B}(H,H)$. 
Following~\cite[chapter 9]{Simon-traceideals}, 
we consider the operator
$$R_p(A)=[(Id+A)\exp(\sum_{n=1}^{p-1}\frac{(-1)^n}{n}A^n )-I ]\in \mathcal{I}_1$$
which is trace class by~\cite[Lemma 9.1 p.~75]{Simon-traceideals}
since $A\in \mathcal{I}_p$. 
Then following~\cite[p.~75]{Simon-traceideals}:
\begin{defi}[Gohberg--Krein's determinants]
For any integer $p\geqslant 2$,
we define
the Gohberg--Krein determinant $\det_p:Id+\mathcal{I}_p\subset \mathcal{B}(H,H)\mapsto \mathbb{C}$
as:
\begin{equation}
\text{det}_p(Id+A)=\text{det}_F\left(Id+R_p(A)\right)
\end{equation}
where $\det_F$ is the
Fredholm determinant.
The quantity $\det_p$ is well defined since
$B=R_p(A)$ is trace class.
\end{defi}
\begin{prop}
Both $\det_F:Id+\mathcal{I}_1\mapsto \mathbb{C}$ and $\det_p:Id+\mathcal{I}_p\mapsto\mathbb{C}$
are entire functions vanishing exactly over non-invertible elements in the following sense:
\begin{eqnarray*}
\text{det}_p\left(Id+B \right)=0 \Leftrightarrow  \text{det}_p\left(Id+zB \right)=(z-1)^{\dim(\ker(Id+B))}\left(C+\mathcal{O}(z-1) \right), C\neq 0.\\
\text{det}_F\left(Id+B \right)=0 \Leftrightarrow \text{det}_F\left(Id+zB \right)=(z-1)^{\dim(\ker(Id+B))}\left(C+\mathcal{O}(z-1) \right), C\neq 0.
\end{eqnarray*}
\end{prop}

\subsection{Geometric setting.}
\label{ss:geomsetting}
 
In the present paragraph, we fix once and for all
the assumptions and  
the general geometric framework of the main Theorems~(\ref{t:quillenconjzeta}),(\ref{t:quillenconjanal}),(\ref{t:quillenholotriv})
and that we shall use in the sequel. 
For $E\mapsto M$ some
smooth Hermitian
vector bundle over the compact manifold 
$M$, we denote by $C^\infty(E)$ smooth sections of $E$.
An operator $\Delta:C^\infty(E)\mapsto C^\infty(E)$ is called generalized Laplacian
if the principal part of
$\Delta$ is positive definite, symmetric (i.e. formally self-adjoint) 
and diagonal with symbol
$g_{\mu\nu}(x)\xi^\mu\xi^\nu\otimes Id_{E_x}$ in local coordinates at
$(x;\xi)\in T^*M$ where $g$ is the Riemannian metric on $M$.
We are interested in the following two 
geometric situations:
\begin{defi}[Bosonic case]\label{d:bosoncase} 
Let $(M,g)$ be a smooth, closed, compact Riemannian manifold and $E$
some Hermitian bundle on $M$. 
We consider the complex affine space $\mathcal{A}$
of perturbations of the form
$\Delta+V$ where $V$ is a smooth endomorphism $V\in C^\infty\left(End(E)\right)$,  
and $\Delta: C^\infty(E) \mapsto C^\infty(E)$ is an invertible generalized Laplacian.
The element $V\in C^\infty(End(E))$ is treated as external potential.
\end{defi}

\begin{defi}[Fermionic case]\label{d:fermioncase} 
Let $(M,g)$ be a smooth, closed, compact Riemannian manifold.
Slightly generalizing the framework described in~\cite[section 3 p.~325--327]{singer1985families}
in the spirit of~\cite[def 3.36 p.~116]{BGV}, 
we are given some pair 
of isomorphic 
Hermitian vector bundles $(E_+,E_-)$ 
of finite rank over $M$ and an invertible, elliptic  
first order differential operator
$D:C^\infty(E_+)\mapsto C^\infty(E_-)$ 
such that both $DD^*:C^\infty(E_-)\mapsto C^\infty(E_-)$ 
and $D^*D:C^\infty(E_+)\mapsto C^\infty(E_+)$ 
are generalized Laplacians 
where $D^*$ is  
the adjoint of $D$
induced by the metric $g$ 
on $M$ and the Hermitian metrics on
the bundles
$(E_+,E_-)$.
We consider the complex affine space $\mathcal{A}$
of perturbations $D+A:C^\infty(E_+)\mapsto C^\infty(E_-)$ where $A\in C^\infty(Hom(E_+,E_-))$.
\end{defi}
Recall that in both cases, we perturb some fixed operator by a
\textbf{local} operator of order $0$.
We next give an important example
from the litterature which 
fits exactly in the fermionic situation:
\begin{ex}[Quantized Spinor fields interacting with gauge fields]
\label{ex:chiralanomaliessinger}
Assume $(M,g)$ is \textbf{spin of even dimension}
whose scalar curvature is nonnegative and positive at some point on $M$. For example
$M=\mathbb{S}^{2n}$ 
with metric $g$ close to the round metric. 
Then it is well--known that the complex spinor bundle $S\mapsto M$ splits as a direct 
sum $S=S_+\oplus S_-$ of isomorphic hermitian vector bundles, the classical Dirac
operator $D:C^\infty(S) \mapsto C^\infty(S)$ is a formally self-adjoint, elliptic operator of Fredholm index $0$
which is invertible by the positivity of the scalar curvature thanks 
to the Lichnerowicz formula~\cite[Cor 8.9 p.~160]{Lawson}.  

 Consider an external hermitian bundle
$\mathcal{F}\mapsto M$ which is coupled to $S$ by tensoring $\left(S_+\oplus S_-\right)\otimes \mathcal{F}= E_+\oplus E_-$. 
For any Hermitian connection $\nabla^{\mathcal{F}}$ on $\mathcal{F}$, we define
the twisted Dirac operator 
$D_{\mathcal{F}}:C^\infty(S\otimes \mathcal{F})\mapsto C^\infty(S\otimes \mathcal{F}) $, which is a first order differential operator of degree $1$ w.r.t. the $\mathbb{Z}_2$ grading,
$D_{\mathcal{F}}=c(e_i)\left(\nabla^S_{e_i}\otimes Id+Id\otimes \nabla^{\mathcal{F}}_{e_i} \right)$ near $x\in M$ where $(e_i)_i$ is a local orthonormal frame 
of $TM$ near $x$, 
$c(e_i)$ is the Clifford action of the local orthonormal frame $(e_i)_i$ of $TM$ on $S$.
In the study of chiral anomalies, one is interested by the half--Dirac operator
$D:C^\infty(S_+\otimes \mathcal{F})\mapsto C^\infty(S_-\otimes \mathcal{F}) $.
If $(M,g)$ has \textbf{positive scalar curvature} and the curvature of $\nabla^{\mathcal{F}}$ is small enough
then $\dim\ker\left(D\right)=0$ and $\text{Ind}\left(D\right)=0$~\cite[prop 6.4 p.~315]{Lawson}. 
Two connections on $\mathcal{F}$ differ by an element $\mathfrak{A}\in \Omega^1(M,End(\mathcal{F}))$. 
So we may define
perturbations $D+A$ of our
half--Dirac operator $D$, 
induced by perturbations of $\nabla^{\mathcal{F}}$, of the form
\begin{equation} 
\boxed{D+A=c(e_i)\left(\nabla^S_{e_i}\otimes Id+Id\otimes \left(\nabla^{\mathcal{F}}_{e_i}+\mathfrak{A}(e_i)\right) \right) .}
\end{equation}  
\end{ex}

%

In the sequel, for a pair  
$(E,F)$ of bundles over $M$, we always identify an element $\mathcal{V}\in C^\infty(Hom(E,F))$, 
which is a $C^\infty$ section of the bundle $Hom(E,F)$ with
the \textbf{corresponding linear operator} $\mathcal{V}:C^\infty(E) \mapsto C^\infty(F)$, in the scalar case this boils down to
identifying a function $V\in C^\infty(M)$ with the multiplication operator $\varphi\in C^\infty(M)\mapsto V\varphi\in C^\infty(M)$.
To avoid repetitions and to stress the similarities between bosons and fermions, 
we will often denote in the sequel (for problem~\ref{d:renormdet1}, Theorems~\ref{t:quillenconjanal} and \ref{t:quillenholotriv}) 
$\mathcal{A}=P+C^\infty(Hom(E,F))$ for the affine space of perturbations 
of $P=\Delta$ of degree $2$, $E=F$ in the bosonic case and of $P=D$ of degree $1$, $E=E_+, F=E_-$ in the fermionic case.
$$\begin{array}{|c|c|c|}
\hline
& \text{Bosons} &\text{Fermions}\\
\hline
\text{Bundles }(E,F) & (E,E) & (E_+,E_-)\\
\hline
\text{Principal part } P & \text{Laplace } \Delta \text{ order }2 & \text{chiral Dirac } D \text{ order } 1\\
\hline
\text{Perturbation }\mathcal{V} & V\in C^\infty(M,End(E)) & A\in C^\infty(M,Hom(E_+;E_-))\\
\hline
\text{Affine space }\mathcal{A} & \Delta+V & D+A\\
\hline 
\end{array} $$
\subsection{Zeta regularization.}

\subsubsection{Defining complex powers by spectral cuts.}

The usual method to construct functional determinants is
the zeta regularization pioneered by Ray--Singer~\cite{ray1971}
in their seminal work on analytic torsion
and relies on spectral or pseudodifferential methods~\cite{Gilkey, Seeley}. 
The reader should see also~\cite{paycha2012regularised,scott2010traces} for some nice recent reviews of various methods to regularize traces and determinants.
Let us review the definition, in our context, of such analytic regularization
(see~\cite[section 3 p.~203]{BK1} for a very nice summary 
of the main results on zeta determinants) keeping in mind the subtle point that we consider non-self-adjoint operators.

Let $M$ be a smooth, closed compact manifold and $E\mapsto M$ some Hermitian bundle.
We denote by $\text{Diff}^1(M,E)$ the space of differential operators with $C^\infty$ coefficients 
of order $1$ acting on $C^\infty\left(E\right)$.
For every perturbation of the form $\Delta+B: C^\infty(E)\mapsto C^\infty(E)$ of an invertible,
symmetric, generalized Laplacian $\Delta
$ by some differential
operator $B\in \text{Diff}^1(M,E)$, the operator $\Delta+B$ has
a canonical closure
from $H^2(M,E)\mapsto L^2(M,E)$ by ellipticity of $\Delta+B\in \Psi^2(M,E)$.
In the notations from subsection~\ref{ss:geomsetting}, $B=V\in C^\infty(M,End(E))$ in the bosonic case has order $0$ or
$B=D^*A\in \text{Diff}^1(M,E_+)$ in the fermionic case in which case $B$ has order $1$.

 By the compactness of the
resolvent $\left(\Delta+B-z\right)^{-1}$ and meromorphic Fredholm theory,
$\Delta+B:H^2(M,E)\mapsto L^2(M,E)$ has discrete spectrum with finite multiplicity which we
denote
by $\sigma\left( \Delta+B\right)\subset \mathbb{C} $.
Since the operator $\Delta+B$ is no longer self-adjoint, we must choose a \textbf{spectral cut} to define
its complex powers.
The operator $\Delta+B$ has principal angle $\pi$ since the value of the principal symbol
$g_{\mu\nu}(x)\xi^\mu\xi^\nu$ of $\Delta+B$ never meets the ray $R_\pi=\{re^{i\pi}, r\geqslant 0\}=\mathbb{R}_{\leqslant 0}$.
Furthermore, for $\Delta^{-1} B\in \mathcal{B}(L^2,L^2)$ small enough, the spectrum $\sigma\left( \Delta+B\right)\subset \mathbb{C} $
will not meet some conical neighborhood $\{ re^{i\theta}, r\geqslant 0, \theta\in [\pi-\varepsilon,\pi+\varepsilon], \varepsilon>0 \}$ of $\mathbb{R}_{\leqslant 0}$, see Proposition~\ref{p:speccut} for more details.
For such $B\in \text{Diff}^1(M)$, $\pi$ is an Agmon angle for $\Delta+B$ and 
$\Delta+B$ is said to be \textbf{admissible with spectral cut} $\pi$.

Since $\Delta+B$ is invertible, we choose some $\rho>0$ s.t. the disc of radius $\rho$ does not meet $\sigma(\Delta+B)$, see Proposition~\ref{p:speccut}.
Then we define the contour~\cite[10.1 p.~87--88]{Shubin}~\cite[p.~12]{kontsevichdeterminants} 
$$\gamma=\{  re^{i\pi}, \infty>r\geqslant \rho \}\cup \{  \rho e^{i\theta}, \theta\in [\pi,-\pi] \}\cup \{re^{-i\pi}, \rho\leqslant r <\infty \}.$$
We define the complex powers as~\footnote{$\lambda^{-s}=e^{-s\log(\lambda)}$ where $\log(\lambda)=\log(\vert \lambda\vert) + i\arg(\lambda)$ for $ -\pi\leqslant \arg(\lambda)\leqslant \pi $}:
$$ \left( \Delta+B\right)_\pi^{-s}=\frac{i}{2\pi}\int_\gamma \lambda^{-s}\left( \Delta+B-\lambda\right)^{-1}d\lambda  . $$

\subsubsection{The spectral zeta function and zeta determinants.}

It is well known that the holomorphic family of operators
$\left( \Delta+B\right)_\pi^{-s}$ is trace class for $Re(s)>\frac{\dim(M)}{2}$ and 
by the works of Seeley~\cite{Gilkey, Seeley}, 
the 
spectral zeta function defined as
\begin{equation}
\zeta_{\Delta+B,\pi}(s)=Tr_{L^2}\left(\left(\Delta+B\right)_\pi^{-s} \right)
\end{equation}
has meromorphic continuation to the complex plane
without poles at $s=0$. 
In fact, much more general operators are considered in the work of Seeley who only requires ellipticity and the existence of an
Agmon angle to define the spectral cut.

To formulate the spectral zeta function entirely
in terms of the spectrum $\sigma(\Delta+B)$, note that $\sigma(\Delta+B)\cap \mathbb{R}_{\leqslant 0}=\emptyset$. 
Then using 
the classical branch
of the logarithm on $\mathbb{C}\setminus \mathbb{R}_{\leqslant 0}$, we can obtain 
an expression for the 
spectral zeta function as~\cite[Eq (2.14) p.~15]{kontsevichdeterminants}
\begin{equation}
\zeta_{\Delta+B,\pi}(s)=\sum_{\lambda\in \sigma(\Delta+B)} \lambda^{-s}
\end{equation}
where the series on the r.h.s. converges absolutely
for $Re(s)>\frac{\dim(M)}{2}$. 
This follows immediately from Lidskii's Theorem and the Weyl law for perturbations of
self-adjoint, positive definite, elliptic operators~\cite[p.~238]{Agranovich-Markus} 
due to Agranovich-Markus.

Let us comment on the 	
above definition now for $B$ arbitrary in $\text{Diff}^1(M,E)$. When $B$ was small in the natural Fr\'echet topology of $\text{Diff}^1(M,E)$, it was unambiguous
to define $\det_\zeta$ with the spectral cut at $\pi$ because we knew $\sigma(\Delta+B)\cap \mathbb{R}_{\leqslant 0}=\emptyset$. However if $B\in\text{Diff}^1(M,E)$ is chosen arbitrarily,
$\sigma(\Delta+B)$ might intersect $\mathbb{R}_{\leqslant 0}$ and we may choose any other spectral cut in $(0,2\pi)$.
In fact, 
the definition of complex powers and spectral zeta function may depend
on the choice of spectral cut but the zeta determinant does not depend on the choice of angle $\theta$ provided $\theta\in (0,2\pi)$ since any such angle $\theta$ is a principal angle.
This is due to the fact that we consider operators of the form $\Delta+B$ where $B$ has order $1$ hence the leading symbol is self-adjoint of Laplace type.
In fact, for any closed conical neighborhood of $\mathbb{R}_{\geqslant 0}$, only a finite number of eigenvalues of $\Delta+B$ lies outside this conical neighborhood as we discuss in Proposition~\ref{p:speccut}. Said differently, for any angle $\theta\in (0,2\pi)$, there exists a conical neighborhood $L_{[\theta-\varepsilon,\theta+\varepsilon]}=\{z | \arg(z)\in [\theta-\varepsilon,\theta+\varepsilon] \}$ s.t. 
$L_{[\theta-\varepsilon,\theta+\varepsilon]}\cap \sigma(\Delta+B)$ is finite.
So moving the cut in $(0,2\pi)$ only crosses 
a finite number of eigenvalues which implies by~\cite[Remark 2.1]{kontsevichdeterminants}~\cite[3.10 p.~206]{BK1} that $\det_{\zeta}(\Delta+B)$
does not depend on $\theta\in (0,2\pi)$. So this justifies why in the sequel we may write unambiguously $\det_\zeta\left(\Delta+B \right)$ where     
we choose any spectral cut $\theta\in (0,2\pi)$ to define $\det_\zeta$. 

\begin{defi}[Spectral zeta determinant]
\label{d:speczetadet}
The zeta determinant
$\det_\zeta$ is defined
as:
\begin{equation}
\text{det}_\zeta(\Delta+B)=\exp\left(-\zeta^\prime_{\Delta+B,\pi}(0) \right).
\end{equation}
\end{defi}
%


We next specialize our definitions of 
zeta determinants in the bosonic and fermionic cases:
\begin{defi}[Zeta determinants for bosons and fermions.]\label{d:bosonsfermionszeta}
We use the geometric setting for bosons and fermions defined in paragraph~\ref{ss:geomsetting}.
For bosons, we define
the corresponding zeta determinant as a map  
\begin{equation}\label{e:zetaboson}
V\in C^\infty(End(E))\mapsto \text{det}_\zeta(\Delta+V).
\end{equation}
For fermions, following~\cite[p.~329]{singer1985families},
we define the corresponding zeta determinant as a map
\begin{equation}\label{e:zetafermion}
A\in C^\infty(Hom(E_+,E_-)) \longmapsto \text{det}_\zeta\left(D^*(D+A) \right).
\end{equation}
\end{defi}

\subsection{Determinants renormalized by subtraction of local counterterms.}
\label{s:subtractloc}
 In order to give a precise definition 
of locality, we recall the definition of smooth local functionals.
\begin{defi}[Local polynomial functionals]
\label{d:smoothlocfun}
A map $P: V\in C^\infty(Hom(E,F)) \mapsto P(V)\in\mathbb{C}$ is called
local polynomial functional if $P$ is smooth in the Fr\'echet sense and there exists $k\in \mathbb{N}$,
$\Lambda: V\in C^\infty(Hom(E,F))  \longmapsto \Lambda\left(V \right)\in  C^\infty(M)\otimes_{C^\infty(M)}\vert\Lambda^{top} \vert M $
s.t. for all $x\in M$, $\Lambda(V)(x)$ depends polynomially on $k$-jets of $V$ at $x$
and
$P(V)=\int_M \Lambda(V)$. 
The vector space of local polynomial functionals of degree $d$
depending on the $k$-jets is denoted by
$\mathcal{O}_{loc,d}\left(J^kHom(E,F) \right)$.
\end{defi}

With the above notion of
local functionals, we can explain the problem of renormalization of 
determinants by subtraction of local counterterms as follows.
We denote by $\mathbb{C}[\varepsilon^{-\frac{1}{2}},\log(\varepsilon)] $ the ring of
polynomials in $\log(\varepsilon)$ and inverse powers $\varepsilon^{-\frac{1}{2}}$. 
If we perturbed some elliptic operator $P$ by any smoothing operator
$\mathcal{V}\in\Psi^{-\infty}$, then 
the Fredholm determinant
$$\mathcal{V}\in \Psi^{-\infty}\mapsto \text{det}_F(Id+P^{-1}\mathcal{V})$$ would be a natural
entire function on $\mathcal{A}=P+\Psi^{-\infty}$ vanishing over non-invertible elements.
Unfortunately, the perturbations $\mathcal{V}\in C^\infty(M,Hom(E,F))$ in both bosonic and fermionic case, are viewed as 
pseudodifferential operators $\mathcal{V}\in \Psi^0(M)$ of degree $0$ hence $\mathcal{V}\in \Psi^0(M)$ is surely not smoothing.
Therefore, the Fredholm determinant $\det_F(Id+P^{-1}\mathcal{V})$ will be ill--defined since
$Id+P^{-1}\mathcal{V}$ does not belong to the determinant class $Id+\mathcal{I}_1$.  
This is why we need to mollify the operator $\mathcal{V}$ by some family $(\mathcal{V}_\varepsilon)_\varepsilon$ of smoothing operators
approximating $\mathcal{V}$ and consider the family $\text{det}_F\left(Id+P^{-1}\mathcal{V}_\varepsilon \right) $ of Fredholm determinants which becomes potentially singular when $\varepsilon\rightarrow 0^+$ and try to absorb the singularities created when $\varepsilon\rightarrow 0^+$
by some multiplicative counterterm. 
This is formalized as follows:
\begin{defi}[Determinants renormalized by subtraction of local counterterms.]\label{d:renormdetssubtraction}
If there is some family $(\mathcal{V}_\varepsilon)_\varepsilon$ of smoothing operators
approximating $\mathcal{V}$, $\mathcal{V}_\varepsilon\underset{\varepsilon\rightarrow 0^+}{\rightarrow} \mathcal{V}$ in $\Psi^{+0}(M)$, some
family of local polynomial functionals 
$P_\varepsilon=\int_M\Lambda_\varepsilon\left(.\right)\in 
\mathcal{O}_{loc,d}\left(J^kHom(E,F) \right)\otimes_{\mathbb{C}} 
\mathbb{C}[\varepsilon^{-\frac{1}{2}},\log(\varepsilon)]$ called \textbf{local counterterms}, such that
the limit 
$$\mathcal{V}\in C^\infty(M,Hom(E,F))\mapsto \lim_{\varepsilon\rightarrow 0^+}
\exp\left(-\int_M \Lambda_\varepsilon(\mathcal{V}(x))\right)  
\text{det}_F\left(Id+P^{-1}\mathcal{V}_\varepsilon \right) $$ 
makes sense as entire function of $\mathcal{V}$~\footnote{$\det_F\left(Id+P^{-1} \mathcal{V}_\varepsilon \right)$ 
is well defined for $\varepsilon>0$ since 
$P^{-1} \mathcal{V}_\varepsilon\in \Psi^{-\infty}$}. 
Then $\lim_{\varepsilon\rightarrow 0^+}
\exp\left(-\int_M \Lambda_\varepsilon(\mathcal{V}(x))\right)  
\text{det}_F\left(Id+P^{-1}\mathcal{V}_\varepsilon \right) $ is the
\textbf{renormalization} of the singular family $\text{det}_F\left(Id+P^{-1}\mathcal{V}_\varepsilon \right) $ \textbf{by subtraction of local counterterms}. 
\end{defi}
%
%

\section{Main Theorems.}

\subsection{Main Theorem on the structure of zeta determinants.}

Our first main result gives 
a factorization formula for zeta determinants in terms of Gohberg--Krein's determinants and renormalized Feynman amplitudes. In fact, the reader can think of this result as some infinite dimensional analog of Hadamard's factorization Theorem~\ref{t:hadamard} in infinite dimension where we think of $\det_\zeta$ as an \textbf{entire function} on the affine space $\mathcal{A}$.
In the sequel, 
we denote by $d_n\subset M^n$, the deepest diagonal $\{(x,\dots,x)\in M^n\text{ s.t. } x\in M\}\subset M^n$ and
by $N^*\left(d_n\subset M^n \right)$ the conormal bundle of $d_n$. 
We use the notion of wave front set $WF(t)$ of a distribution $t$ 
to describe singularities of $t$ in cotangent space 
and refer to~\cite[chapter 8]{HormanderI} for the precise definitions.

For $a\in \mathbb{R}$, we denote by $[a]=\sup_{k\in \mathbb{Z}, k\leqslant a}k$.
The bundle of densities on a manifold $X$ will be denoted by $\vert\Lambda^{top}\vert X$.

\begin{thm}\label{t:quillenconjzeta}
The zeta determinants from definition~\ref{d:bosonsfermionszeta} are entire functions on $\mathcal{A}$ 
satisfying the factorization formula:
\begin{eqnarray}\label{e:bosonfacthadamard}
\text{det}_\zeta\left(\Delta+V \right)&=&e^{Q(V)}\text{det}_{p}\left(Id+\Delta^{-1}V \right), p=[\frac{d}{2}]+1 \text{ in bosonic case}\\
\label{e:fermionfacthadamard}
\text{det}_\zeta(D^*(D+A))&=&e^{Q(A)}\text{det}_{p}\left(Id+D^{-1}A \right), p=d+1 \text{ in fermionic case}
\end{eqnarray}
where $\det_p$ are Gohberg--Krein's determinants,
$Q(V)$ (resp $Q(A)$) has degree $[\frac{d}{2}] $ (resp $d$)~\footnote{Beware that $Q$ is not a local polynomial functional}.   

We furthermore have the following properties
\begin{itemize}
\item an exponential bound on the growth:
\begin{eqnarray*}
\vert \text{det}_\zeta\left(\Delta+V \right)\vert&\leqslant& Ce^{K\Vert V\Vert_{C^{d-3}}^{[\frac{d}{2}]+1}  }\\
\vert \text{det}_\zeta(D^*(D+A))\vert&\leqslant& Ce^{K\Vert A\Vert_{C^{d-1}}^{d+1}  },
\end{eqnarray*}
\item an identity for all the G\^ateaux differentials:
\begin{eqnarray*}
\frac{(-1)^{n-1}}{n-1!}D^n\log\text{det}_\zeta\left(\Delta+V,V_1,\dots,V_n \right)&=&
Tr_{L^2}\left((\Delta+V)^{-1} V_1 \dots (\Delta+V)^{-1} V_n\right),\\\text{ if }\text{supp}(V_1)\cap\dots\cap \text{supp}(V_n)=\emptyset &&  \\
\frac{(-1)^{n-1}}{n-1!}D^n\log\text{det}_\zeta\left(D^*(D+A),A_1,\dots,A_n\right)&=&Tr_{L^2}\left((D+A)^{-1} A_1\dots (D+A)^{-1} A_n\right),\\
 \text{ if }\text{supp}(A_1)\cap\dots\cap \text{supp}(A_n)=\emptyset &&
\end{eqnarray*}
\item a bound on the wave front set of the Schwartz kernels of all the G\^ateaux differentials:
\begin{eqnarray*}
 WF\left( [\mathbf{D\log\text{det}_\zeta\left(\Delta\right)}] \right)=\emptyset,\\
\forall n\geqslant 2, WF\left( [\mathbf{D^n\log\text{det}_\zeta\left(\Delta\right)}] \right)\cap T_{d_n}^*M^n&\subset &N^*(d_n\subset M^n),\\
WF\left( [\mathbf{D\log\text{det}_\zeta\left(D^*D\right)}] \right) =\emptyset,\\ 
\forall n\geqslant 2, WF\left( [\mathbf{D^n\log\text{det}_\zeta\left(D^*D\right)}] \right)\cap T_{d_n}^*M^n&\subset &N^*(d_n\subset M^n), 
\end{eqnarray*}
where $[\mathbf{D^n\log\text{det}_\zeta\left(\Delta\right)}]\in \mathcal{D}^\prime(M^n)$ (resp $[\mathbf{D^n\log\text{det}_\zeta\left(D^*D\right)}] \in \mathcal{D}^\prime(M^n)$) denotes the Schwartz kernel of the $n$-th differential 
$D^n\log\det_\zeta\left(\Delta\right)$ (resp $D^n\log\det_\zeta\left(D^*D\right)$).
\end{itemize}
\end{thm}

The choice of branch of the $\log$ is dictated by the Agmon angle but the results on the differentials of $\log\det_\zeta$ does not depend on the chosen branch of $\log$.

There are several consequences of the above result. The first straightforward consequence
is that
the zeta determinants of Theorem~\ref{t:quillenconjzeta} vanish exactly over non-invertible elements in $\mathcal{A}$ in the following sense:
\begin{eqnarray*}
\text{det}_\zeta(\Delta+V)=0&\Leftrightarrow &\text{det}_\zeta\left(\Delta+zV \right)=(z-1)^{\dim(\ker(\Delta+V))}\left(C+\mathcal{O}(z-1) \right), C\neq 0, \\
\text{det}_\zeta(D^*(D+A))=0&\Leftrightarrow &\text{det}_\zeta\left(D^*(D+zA)\right)=(z-1)^{\dim(\ker(D^*(D+A)))}\left(C+\mathcal{O}(z-1) \right), C\neq 0.
\end{eqnarray*} 
Furthermore:
\begin{corr}[Zeta determinant for
non smooth, non-self-adjoint perturbations]
\label{c:zetafact}
The zeta determinants of Theorem~\ref{t:quillenconjzeta} extend as entire functions
of \textbf{non smooth, non-self-adjoint} perturbations
\begin{itemize}
\item of $\Delta$ of regularity $C^{d-3}(End(E))\cap L^\infty(End(E))$ in the bosonic case,  
\item of $D$ of regularity $C^{d-1}(Hom(E_+,E_-))$ in the fermionic case.
\end{itemize}
\end{corr}


\subsection{An analytic reformulation of Quillen's conjectural picture.}
\label{ss:introanalyticproblem}

In our setting, we attempt to reformulate Quillen's question as a problem 
of constructing an entire function with prescribed zeros in the infinite dimensional space $\mathcal{A}$
generalizing the Fredholm determinant.
Our first Theorem~\ref{t:quillenconjzeta} seems to indicate that the 
zeta determinant $\det_\zeta$ is a good candidate, but is it the only possible construction ? 
A naive approach suggested by Quillen in~\cite{Quillen89}
would be to consider the Fredholm determinant
$\det_F\left(Id+D^{-1}A \right)$ where for small $A$, we expect
that $$\log\text{det}_F\left(Id+D^{-1}A \right)=\sum_{k=1}^\infty \frac{(-1)^{k+1}}{k}Tr\left((D^{-1}A)^k \right).$$
However as remarked by Quillen, the operator $D^{-1}A$ is a pseudodifferential operator of order $-1$, hence for $k>d$ the power $(D^{-1}A)^k$ is trace class hence the 
traces $Tr\left((D^{-1}A)^k  \right)$ are well--defined whereas for $k\leqslant d$ these traces are ill--defined and often \textbf{divergent}
as usual in QFT. We will later see how to deal with these divergent traces in Theorem~\ref{t:quillenconjanal}.

 We next formulate the general problem of finding renormalized determinants
with functional properties closed to zeta determinants:

\begin{prob}[Renormalized determinants]\label{d:renormdet1}
Under the geometric setting from paragraph~\ref{ss:geomsetting}, 
set $\mathcal{A}=P+C^\infty(Hom(E,F)), p=\deg(P)$ 
where $P=\Delta, p=2, E=F$ in the bosonic case and 
$P=D, p=1, E=E_+, F=E_-$ in the fermionic case. 
An entire function $\mathcal{R}\det:\mathcal{A}\mapsto \mathbb{C}$
will be called renormalized determinant
if
\begin{enumerate}
\item $\mathcal{R}\det$ vanishes exactly on the subset of noninvertible elements 
in the following sense
$$\mathcal{R}\text{det}(P+\mathcal{V})=0\Leftrightarrow \mathcal{R}\text{det}\left(P+z\mathcal{V}\right)=(z-1)^{\dim(\ker(P+\mathcal{V}))}\left(C+\mathcal{O}(z-1) \right), C\neq 0,$$
and 
$\mathcal{R}\det$ satisfies the bound:
\begin{eqnarray}\label{e:boundrdet}
\boxed{ \vert\mathcal{R}\det\left( P+\mathcal{V} \right) \vert \leqslant C e^{K\Vert \mathcal{V}\Vert^{[\frac{d}{p}]+1}_{C^m} } }
\end{eqnarray}
for the continuous norm $\Vert .\Vert_{C^m}$  
on $C^\infty(Hom(E,F))$ where $m=d-3$ in the bosonic case, $m=d-1$ in the fermionic case 
and $C,K>0$ independent of $\mathcal{V}$.
\item For $n>[\frac{d}{p}]$, 
\begin{equation}\label{e:constraint} 
\boxed{\frac{(-1)^{n-1}}{n-1!} \left(\frac{d}{dz}\right)^n \log\mathcal{R}\det(P+z\mathcal{V})|_{z=0}=Tr_{L^2}\left(\left(P^{-1}\mathcal{V}\right)^n \right).}
\end{equation}
\item For $\Vert \mathcal{V}\Vert_{C^m}$ small enough, 
we further impose two conditions of microlocal nature 
on the second G\^ateaux differential of $\mathcal{R}\det$. The first one reads:
\begin{equation}\label{e:secderivative}
\boxed{D^2\log\mathcal{R}\det\left(P+\mathcal{V},V_1,V_2 \right)=Tr_{L^2}\left(\left(P+\mathcal{V}\right)^{-1}V_1 \left(P+\mathcal{V}\right)^{-1}V_2\right) }
\end{equation}
if $\text{supp}(V_1)\cap \text{supp}(V_2)=\emptyset$ where the $L^2$ trace is well--defined since $\left(P+\mathcal{V}\right)^{-1}V_1 \left(P+\mathcal{V}\right)^{-1}V_2$ is smoothing.\\
Recall $[\mathbf{D^2\log\mathcal{R}\det(\Delta+\mathcal{V})}]\in \mathcal{D}^\prime(M^2,Hom(E,F)\boxtimes Hom(E,F))$ denotes the Schwartz kernel of 
the bilinear map $D^2\log\mathcal{R}\det\left(P+\mathcal{V},.,. \right)$, then the second condition reads: 
\begin{equation}\label{e:wfsecderivative}
\boxed{WF\left([\mathbf{D^2\log\mathcal{R}\det(\Delta+\mathcal{V})}]\right)\cap T_{d_2}^\bullet M^2 \subset N^*\left(d_2\subset M^2\right).}
\end{equation}
\end{enumerate}
\end{prob}

Note that in the fermionic case, our discussion is non trivial if 
the Fredholm index of $D$ vanishes. But in fact, we require in definition~\ref{d:fermioncase}
that $D$ is invertible which means having Fredholm index $0$ and $\ker(D)=\{0\}$ which is a stronger condition.
Also the $L^2$ trace in the r.h.s. of equation~\ref{e:secderivative} is well--defined since  
$\left(P+\mathcal{V}\right)^{-1}V_1 \left(P+\mathcal{V}\right)^{-1}V_2$ is smoothing because  the condition
$\text{supp}(V_1)\cap \text{supp}(V_2)=\emptyset$
on the supports of $V_1,V_2$ implies the symbol of
$\left(P+\mathcal{V}\right)^{-1}V_1 \left(P+\mathcal{V}\right)^{-1}V_2$ vanishes 
(see~\cite[1.1]{kontsevichdeterminants} 
for similar observations).

Let us motivate the axioms from definition~\ref{d:renormdet1}.
About condition $1)$,
it is natural to require our determinants
to vanish on noninvertible elements since
they generalize the usual Fredholm determinant.
Furthermore, we want to minimize the 
growth at infinity of the entire function
$z\in \mathbb{C}\mapsto \mathcal{R}\det(P+z\mathcal{V}
)$, hence its \emph{order} in the sense of subsection~\ref{ss:entireonevar}. 
We will see in corollary~\ref{c:optimalorder} 
that our condition on the order of $z\in \mathbb{C}\mapsto \mathcal{R}\det(P+z\mathcal{V}
)$ is optimal in the sense this is the smallest growth at infinity we can require. 
This is 
in some sense responsible for the polynomial ambiguity conjectured
by Quillen which prevents us from having a unique solution to problem~\ref{d:renormdet1}. In the same way, 
there is not necessarily a unique solution
to the problem of finding an entire function with prescribed zeros in the Hadamard factorization Theorem~\ref{t:hadamard}.

About condition $2)$ that we impose on the derivatives
of $\log\mathcal{R}\det$, this is reminiscent of
the derivatives for the $\log$ of 
Gohberg--Krein's determinants $\mathcal{V}\mapsto  \log\det_p\left(Id+P^{-1}\mathcal{V} \right)$. 
$\mathcal{V}\mapsto  \det_p\left(Id+P^{-1}\mathcal{V} \right)$ also vanishes 
exactly on non-invertible elements. However, Gohberg--Krein's determinants $\mathcal{V}\mapsto  \det_p\left(Id+P^{-1}\mathcal{V} \right)$
fail to satisfy the conditions on the
second derivative of problem \ref{d:renormdet1}, 
hence by our main Theorem \ref{t:quillenconjanal} they cannot be obtained
from renormalization by subtraction of local counterterms
since our Theorem~\ref{t:quillenconjanal} will show
that these conditions are necessary to 
describe all renormalized determinants which can be obtained
by a renormalization procedure where we subtract only local counterterms.

About condition $3)$,
Equations~(\ref{e:constraint}) and~(\ref{e:secderivative}) are very natural since they are reminiscent of
the usual determinant in
the finite dimensional case.
In the seminal work of Kontsevich--Vishik~\cite[equation (1.4) p.~4]{kontsevichdeterminants}, they 
attribute to Witten
the observation that for the zeta determinant, 
the following identity 
$$D^2\log\text{det}_\zeta\left(A,A_1,A_2 \right)=-Tr_{L^2}\left(A_1 A^{-1}A_2 A^{-1}\right)$$
holds true
where $A_1,A_2$ are pseudodifferential deformations with \textbf{disjoint support}.
This is not surprising provided we want our determinants to give rigorous meaning
to QFT functional integrals.  
~\footnote{In the present paper, we take this as axiom of our renormalized determinants
and the identity (\ref{e:secderivative}) follows
from a formal applications of Feynman rules.}
Finally, we want to subtract only \textbf{smooth local counterterms in} $\mathcal{V}$, this smoothness will be 
imposed by the conditions on the wave front set of the Schwartz kernel of the G\^ateaux differentials.
The bound on $m$ is also optimal, locality forces renormalized determinants to
depend on $m$-jets of the external potential $\mathcal{V}$.

\subsubsection{Solution of problem \ref{d:renormdet1}.}
We now state the main Theorem of the present paper answering Problem~\ref{d:renormdet1}, the assumptions
are from paragraph~\ref{ss:geomsetting} 
:
\begin{thm}[Solution of the analytical problem]
\label{t:quillenconjanal}
A map $\mathcal{R}\det:\mathcal{A}\mapsto \mathbb{C}$ is a solution of problem
\ref{d:renormdet1} if and only if the following equivalent conditions are satisfied:  
\begin{enumerate}
\item there exists $Q\in \mathcal{O}_{loc,[\frac{d}{p}]}\left(J^mHom(E,F) \right)$
such that
{\small
\begin{eqnarray}
V\mapsto\mathcal{R}\det(\Delta+V)&=&\exp\left(Q(V) \right)\text{det}_\zeta\left(\Delta+V \right), p=2, m=d-3\text{ for bosons}\\
A\mapsto \mathcal{R}\det(D+A)&=&\exp\left(Q(A) \right)\text{det}_\zeta\left(D^*(D+A) \right), p=1, m=d-1\text{ for fermions}.
\end{eqnarray}
}

\item $\mathcal{R}\det$ is renormalized by subtraction of local counterterms.
There exists a generalized Laplacian $\Delta$ with heat operator
$e^{-t\Delta}$ and a family $Q_\varepsilon\in \mathcal{O}_{loc,[\frac{d}{p}]}\left(J^mHom(E,F) \right)\otimes_\mathbb{C} \mathbb{C}[\varepsilon^{-\frac{1}{2}},\log(\varepsilon)]$ such that~\footnote{the choice of mollifier 
$e^{-2\varepsilon\Delta}$ is consistent with the GFF interpretation since the covariance of the heat regularized GFF $e^{-\varepsilon\Delta}\phi$
is $e^{-2\varepsilon\Delta}\Delta^{-1}$}:
\begin{eqnarray}
\mathcal{V}\mapsto \mathcal{R}\det\left(P+\mathcal{V} \right)=\lim_{\varepsilon\rightarrow 0^+}\exp\left(Q_\varepsilon(\mathcal{V})\right) \text{det}_F\left(Id+e^{-2\varepsilon\Delta}P^{-1} \mathcal{V} \right).
\end{eqnarray}
\end{enumerate}
\end{thm}
As immediate corollary of the above, we get that
the group $\mathcal{O}_{loc,[\frac{d}{p}]}\left(J^mHom(E,F) \right)$ of local polynomial functionals acts freely and transitively on 
the set of renormalized determinants solutions to~\ref{d:renormdet1}:
\begin{eqnarray}
Q\in \mathcal{O}_{loc,[\frac{d}{p}]}\left(J^mHom(E,F) \right)\mapsto \exp\left(Q(\mathcal{V})\right)\mathcal{R}\det\left(P+\mathcal{V} \right).
\end{eqnarray}
Theorem~\ref{t:quillenconjanal} also shows that zeta determinants are a particular case of some infinite dimensional family  
of renormalized determinants obtained by subtracting singular local counterterms.

\begin{corr}\label{c:renormvshadamard}
In particular under the assumptions of Theorem \ref{t:quillenconjanal} and using the same notations, $p=\deg(P)$ any function $\mathcal{R}\det\left(P+\mathcal{V} \right)$
can be represented as: 
\begin{eqnarray*}
\mathcal{R}\det\left(P+\mathcal{V}\right)&=&\exp\left(Q(\mathcal{V})\right) \text{det}_{[\frac{d}{p}]+1}\left(Id+P^{-1}\mathcal{V} \right) \\
&=&\exp\left(Q(\mathcal{V})\right)\prod_{n=1}^\infty 
E_{[\frac{d}{p}]}\left(\frac{1}{\lambda_n} \right)
\end{eqnarray*}
where $Q$ is a polynomial functional of $\mathcal{V}$ of degree $[\frac{d}{p}]$, $\det_{[\frac{d}{p}]+1}$ is Gohberg--Krein's determinant, $E_{k}(z)=(1-z)e^{z+\frac{z^2}{2}+\dots+\frac{z^{k}}{k}}, k>0$ is a Weierstrass factor 
and the infinite product is over the sequence
$\{\lambda|\dim\ker\left( P+\lambda \mathcal{V}\right)\neq 0\}$.
\end{corr}
 
\subsection{Renormalized determinants and holomorphic trivializations of Quillen's line bundle}

Let us quickly recall the notations from paragraph \ref{ss:quillenpict}.
For a pair $\mathcal{H}_0,\mathcal{H}_1$ of complex Hilbert spaces,  
there is a \textbf{canonical holomorphic line bundle}
$\mathbf{Det} \longmapsto \text{Fred}_0\left(\mathcal{H}_0,\mathcal{H}_1 \right)$ where $\text{Fred}_0\left(\mathcal{H}_0,\mathcal{H}_1 \right)$ is the space of Fredholm operators of \textbf{index} $0$ with fiber $\mathbf{Det}_B\simeq \Lambda^{top} \ker\left( B\right)^*\otimes \Lambda^{top} \text{coker}\left( B\right)$ over each $B\in \text{Fred}_0\left(\mathcal{H}_0,\mathcal{H}_1 \right)$ and canonical section $\sigma$~\cite[p.~32]{quillen1985determinant}~\cite[p.~137--138]{segal2004definition}. 
Recall in the fermionic situation, we considered the \textbf{complex affine space} $\mathcal{A}=D+C^\infty(Hom(E_+,E_-))$
of perturbations of some invertible, elliptic Dirac operator $D$.
Then the map
$\iota: D+A\in \mathcal{A}\mapsto Id+D^{-1}A\in \text{Fred}_0\left(L^2(E_+),L^2(E_+)\right) $
allows to pull--back the holomorphic line bundle $\mathbf{Det}  \longmapsto \text{Fred}_0\left(L^2(E_+),L^2(E_+)\right)$ as a holomorphic line bundle
$\mathcal{L}=\iota^*\mathbf{Det}\mapsto \mathcal{A}$ over the affine space $\mathcal{A}$ with canonical section $\underline{\det}=\iota^*\sigma$.
We denote by
$\mathcal{O}(\mathcal{L})$ the holomorphic sections
from $\mathcal{L}$ and by $\mathcal{O}(\mathcal{A})$
holomorphic functions on $\mathcal{A}$. 
\begin{thm}[Holomorphic trivializations and flat connection]
\label{t:quillenholotriv}
There is a bijection between the set of renormalized $\mathcal{R}\det$ from Theorem
\ref{t:quillenconjanal} and global holomorphic trivialization
$\tau: \mathcal{O}(\mathcal{L}) \mapsto \mathcal{O}\left(\mathcal{A} \right) $
of the line bundle $\mathcal{L}\mapsto \mathcal{A}$ such that
\begin{eqnarray}
 T\in \mathcal{A} \mapsto  \tau\left(\iota^*\underline{\det}(T)\right)=\mathcal{R}\det(T) 
\end{eqnarray}
is a solution of problem~\ref{d:renormdet1}.
The image of the canonical section $\iota^*\underline{\det}(T)$ under this trivialization 
being exactly the entire function $\mathcal{R}\det $ vanishing over non-invertible elements
in $\mathcal{A}$.

For every pair $\left(\tau_1,\tau_2\right)$, there exists an element $\Lambda$ of the additive group $\left(\mathcal{O}_{loc,d},+\right)$ 
s.t.
\begin{eqnarray}
\tau_1(D+A)= \exp\left(\int_M \Lambda(A(x)) \right) \tau_2(D+A),
\end{eqnarray}
where $\Lambda$ depends on the $(d-1)$-jet of $A$.
 For every choice of renormalized determinant $\mathcal{R}\det$, the section 
$\sigma=\mathcal{R}\det^{-1}\iota^*\underline{\det}$ defines a nowhere vanishing global holomorphic section with canonical 
holomorphic
flat connection $\nabla$ s.t. $\nabla \sigma=0$.
\end{thm}

 The ambiguity group that relates all
solutions of problem~\ref{d:renormdet1}
is the renormalization group of 
St\"ueckelberg--Petermann as described by Bogoliubov--Shirkov~\cite{Bogoliubov} and is interpreted here as a gauge group of the line bundle
$\mathcal{L}\mapsto \mathcal{A}$. 
Our result is a variant of the so called
main Theorem of renormalization by Popineau--Stora~\cite{Stora} 
and studied under several aspects by Brunetti--Fredenhagen~\cite{BF} and Hollands--Wald~\cite{HW1,HW2,KM}.
In the aQFT community, there are various recent works exploring the renormalized Wick powers using Euclidean versions of the
Epstein--Glaser renormalization~\cite{DappRicci,DappEucl}.

\subsection*{Relation with other works.}
%
%
 The way we treat the problem of subtraction of local counterterms is strongly inspired
by Costello's work~\cite{Costello} and the point of view of perturbative algebraic quantum field theory
which is explained in Rejzner's book~\cite{rejznerbook}.

 Perrot's notes~\cite{Perrot-07} and Singer's paper~\cite{singer1985families} on quantum anomalies, which played an important role in our understanding of the topic,
are in the real setting. The gauge potential $A$ which is used to perturb 
the half--Dirac operator 
preserves the Hermitian structure whereas 
we do not impose this
requirement and view our perturbations 
as a complex space instead. 
Actually, 
our motivation 
to consider holomorphic determinants 
in some complexified setting
bears strong inspiration from the work of Burghelea--Haller~\cite{BH1,BH2} and Braverman--Kappeler~\cite{BK1}
on finding some complex valued holomorphic version of the Ray--Singer analytic torsion.  

In this short paragraph, we shall adopt the notation of~\cite{kontsevichdeterminants}.  
Our renormalized determinants seem 
related to the multivalued function $f$~\cite[prop 4.10]{kontsevichdeterminants} 
on the space of elliptic operators
$Ell^m_{(-1)}(M;E,F)$ endowed with a natural complex structure~\cite[Remark 4.18]{kontsevichdeterminants} introduced in the seminal work of Kontsevich--Vishik. This multivalued function, defined on certain good classes of
\textbf{non-self-adjoint operators}, naturally \textbf{extends the functional determinant} $\det_{\tilde{\pi}}$ defined on self-adjoint operators~\cite[Prop 4.12]{kontsevichdeterminants}. This is more general than the operators we consider in the present paper since we only work in the restricted class
$\Delta+\Psi^0$ and $\text{Dirac} + \Psi^0$ where only the subprincipal
part is allowed to vary whereas in~\cite{kontsevichdeterminants} the full symbol is also allowed to change. 
We also note that~\cite[4.1 p.~52]{kontsevichdeterminants} also consider determinants of Dirac operators exactly as in 
the present work but on odd dimensional manifolds. 
However,
we obtain a factorization formula for $\det_\zeta$ in terms of Gohberg--Krein determinants and we relate zeta regularization with renormalization
in quantum field theory which seem to be new results. Moreover, 
our factorization formula allows us to consider nonsmooth perturbations of generalized Laplacians or Dirac operators 
which also seems to be a new result. 
The way we define the determinant line bundle is very close to~\cite[section 6]{kontsevichdeterminants} and work of Segal
and differs from the seminal work of Bismut--Freed~\cite{BismutFreed} although all definitions should give the same object
when restricted to self-adjoint families of Dirac operators. In particular, 
we do not focus on Quillen metrics and connections on the determinant line bundle
which are important results of~\cite{BismutFreed} 
but on the holomorphic structure instead and the relation with renormalization ambiguities as conjectured by 
Quillen in his notes~\cite[p.~284]{Quillen89}.

Finally, in a nice recent paper~\cite{friedlandermult}, Friedlander
generalized the classical multiplicative formula 
$\det_\zeta\left(\Delta(Id+T) \right)=\det_\zeta\left(\Delta\right)\det_F\left(Id+T\right)$
when $T$ is smoothing,
in~\cite[Theorem p.~4]{friedlandermult} connecting zeta determinants, 
Gohberg--Krein's determinants and Wodzicki residues.
This bears a strong similarity with 
our Corollary~\ref{c:zetafact}
although our point of view 
stresses the relation with
distributional extensions of products of Green functions\footnote{called \emph{Feynman amplitudes} in physics litterature}
in configuration space.  Another difference with his work is that we bound the 
wave front of the Schwartz kernels of the G\^ateaux differentials of 
zeta determinants which is important from the  
QFT viewpoint and is related to the microlocal spectrum condition used in QFT.

\section{Proof of Theorem~\ref{t:quillenconjzeta}.}

We work under the setting of paragraph
\ref{ss:geomsetting} and the zeta determinants
are defined in definitions \ref{d:speczetadet} and \ref{d:bosonsfermionszeta}.  
We discuss in great detail the bosonic case for $\det_\zeta(\Delta+V)$ where
$V\in C^\infty(M,End(E))$ and 
we indicate precisely the differences when we deal with the fermionic case for
$\det_\zeta(\Delta+D^*A)$ where $\Delta=D^*D$ is a generalized Laplacian, the operator $D:C^\infty(E_+)\mapsto C^\infty(E_-)$ is a generalized 
Dirac operator
and $A\in C^\infty(Hom(E_+,E_-))$. 
\textbf{Both cases consider zeta determinants of a non-self-adjoint perturbation 
of some generalized Laplacian by some differential operator $V$ of order $0$ in the bosonic case and
$V=D^*A$ of order $1$ in the fermionic case}.

\subsection{Reformulation of the Theorem.}

\subsubsection{Polygon Feynman amplitudes.}
\label{sss:feynmanampl}

We will first state a reformulation of our Theorem~\ref{t:quillenconjzeta} 
in terms of certain \emph{Feynman amplitudes}. This emphasizes the 
relation of our result with quantum field theory as in the work of Perrot~\cite{Perrot-07}.
Before, we need to define these Feynman amplitudes as 
formal products of 
Schwartz kernels of the operators involved in 
our problem. 
This will play an important role in our Theorem:
\begin{defi}[Polygon Feynman amplitudes]\label{d:feynmanamplitudes}
Under the geometric setting from paragraph~\ref{ss:geomsetting}, we
set $\mathbf{G}\in \mathcal{D}^\prime\left(M\times M, E\boxtimes E^* \right)$  
to be the Schwartz 
kernel of $\Delta^{-1}$ in the bosonic case and
$\mathbf{G}\in \mathcal{D}^\prime\left(M\times M, E_-\boxtimes E_+^* \right)$  
is the Schwartz 
kernel of $D^{-1}$ in the fermionic case. 
For every $n\geqslant 2$, we formally set
\begin{eqnarray}
\boxed{t_{n}=\mathbf{G}(x_1,x_2)\dots\mathbf{G}(x_{n-1},x_n)\mathbf{G}(x_n,x_1)}
\end{eqnarray}
which is well--defined in
$ C^\infty\left(M^n\setminus \text{Diagonals}\right) $.
\end{defi}

An important remark, we will later see in the proof of our main Theorem
that the formal products $t_n$ are actually well--defined as distributions 
in $\mathcal{D}^\prime(M^n\setminus d_n)$ where we deleted only the deepest diagonal $d_n$. 
This could also be easily 
proved by estimating the wave front set 
of the product 
$t_n=\mathbf{G}(x_1,x_2)\dots\mathbf{G}(x_{n-1},x_n)\mathbf{G}(x_n,x_1)$, 
using the fact that 
the wave front set of the pseudodifferential kernels
are contained in the conormal bundle of the diagonal $d_2\subset M\times M$.

Another remark related to quantum field theory. In QFT, formal products of Schwartz kernels, called \emph{propagators}, are described by 
graphs~\footnote{called Feynman graphs} where every edge of a given graph stands for a Schwartz kernel $\mathbf{G}$ and vertices of the graph represent 
pointwise 
products of propagators. The correspondance $\text{Graph} \mapsto \text{Distributional amplitude}  $  
is described precisely in terms of the Feynman rules~\cite[2.1]{dangzhang}. Using the Feynman rules, one sees that
the amplitude $t_n$ defined above represents some polygon graph with $n$ vertices and edges.

\subsubsection{Alternative statement of Theorem~\ref{t:quillenconjzeta}.}
For any pair $(E,F)$ of Hermitian bundles over $M$, there is a
natural 
fiberwise pairing $\left\langle t,\varphi \right\rangle$ between distributions $t$
in 
$\mathcal{D}^\prime(M^n,Hom(F,E)^{\boxtimes n})$ with elements $\varphi$
in $C^\infty\left(M^n,Hom(E,F)^{\boxtimes n} \right)$ to get an element $\left\langle t,\varphi \right\rangle\in \mathcal{D}^\prime(M^n)$. To 
obtain a
number, we need to integrate this distribution against a density $dv\in \vert\Lambda^{top}\vert M^n $ as in
$\int_{M^n}\left\langle t, \varphi\right\rangle dv$.

\begin{thm}\label{t:mainzetareformulate}
The zeta determinants from definition~\ref{d:bosonsfermionszeta} are entire functions on $\mathcal{A}$ 
satisfying the factorization formula for $\Vert V\Vert_{C^{d-3}(End(E))}$ (resp $\Vert A\Vert_{C^{d-1}(Hom(E,F))}$) small enough:
\begin{eqnarray}\label{e:bosonfacthadamard2}
\text{det}_\zeta\left(\Delta+V \right)&=&e^{Q(V)}\text{det}_{p}\left(Id+\Delta^{-1}V \right), p=[\frac{d}{2}]+1 \text{ in bosonic case}\\
\label{e:fermionfacthadamard2}
\text{det}_\zeta(D^*(D+A))&=&e^{Q(A)}\text{det}_{p}\left(Id+D^{-1}A \right), p=d+1 \text{ in fermionic case}
\end{eqnarray}
where $\det_p$ are Gohberg--Krein's determinants. There exists 
$\ell\in C^\infty(End(E))$ in bosonic case
and $\ell\in C^\infty(Hom(E_+,E_-))$ in fermionic case s.t.
\begin{eqnarray}\label{e:repsformula}
Q(V)&=&\int_M \left\langle \ell,V \right\rangle dv+\sum_{2\leqslant n\leqslant [\frac{d}{2}]} \frac{(-1)^{n+1}}{n}\int_{M^n} \left\langle \mathcal{R}t_n , V^{\boxtimes n}\right\rangle dv_n\\
Q(A)&=&\int_M \left\langle \ell,A \right\rangle dv+ \sum_{2\leqslant n\leqslant d} \frac{(-1)^{n+1}}{n} \int_{M^n} \left\langle \mathcal{R}t_{n} , A^{\boxtimes n}\right\rangle dv_n
\end{eqnarray}
\begin{itemize}
\item $dv\in \vert \Lambda^{top}\vert M$,
$dv_n\in \vert\Lambda^{top} \vert M^n $ are the canonical Riemannian densities,
\item $\mathcal{R}t_{n}$ 
is a 
distribution of order $m$ on $M^n$ extending the distributional product $t_{n}$ from definition~\ref{d:feynmanamplitudes} which is well--defined on 
$M^n\setminus d_n$, $m=d-3$ in the bosonic case and $m=d-1$ in the fermionic case,
\item the wave front set of $\mathcal{R}t_n$ 
satisfies the bound 
$$WF\left(\mathcal{R}t_{n} \right) \cap T^\bullet_{d_n}M^n \subset N^*\left(d_n\subset M^n  \right) . $$ 
\end{itemize}
\end{thm}

Let us compare this statement with the original 
formulation of Theorem~\ref{t:quillenconjzeta}. The distribution $\mathcal{R}t_n$ in the above statement
is nothing but the Schwartz kernel $[\mathbf{D^n\log\det_\zeta\left(\Delta \right)}]$ of the G\^ateaux differentials 
up to some multiplicative constant.

\subsubsection{Plan of the proof.}

The main idea of the proof of the Theorem in both bosonic and fermionic cases
is to calculate G\^ateaux differentials of $V\mapsto \log\det_\zeta(\Delta+V)$ with respect to the perturbation and 
compare with the G\^ateaux differentials of some well chosen Gohberg--Krein determinant $V\mapsto \log\det_p\left(Id+\Delta^{-1}V\right)$.
Physically, this seems to be a natural idea since we look for the response of the \emph{free energy} $\log\det_\zeta$ (the logarithm of the partition function)
under variation of the external field which is $V$ in the bosonic case and $A$ in the fermionic case.
What we will prove is simply that starting from a certain order, all derivatives actually coincide and since we know that both 
sides are analytic, they must coincide up to some polynomial in the perturbation $V$. 
A second important idea is to recognize that the Schwartz kernels of the 
G\^ateaux differentials are \textbf{distributional extensions} 
of products of Green kernels of the elliptic operator that we perturb, paired with external powers of
the perturbation. This relies on explicit representation of the Schwartz kernels of G\^ateaux differentials in terms of heat kernels.

The first step is to discuss the analyticity properties of $\det_\zeta$ in subsubsection~\ref{ss:analyticity}. Our proofs rely on representations
of complex powers and their differentials in terms of heat kernels in Lemma~\ref{l:funder1}. This requires to consider small perturbations of $\Delta$ 
which maintain a spectral gap ensuring heat operators have exponential decay in subsection~\ref{ss:perturbgap}.  
But once the factorization formula is proved for small perturbations, it extends to $\mathcal{A}$ by analyticity
of $\det_\zeta$.  In subsection~\ref{ss:optraces}, we decompose
the integral formula involving integrals of heat operators in two parts: 
a singular part involving short time of the heat operators that we control with the heat calculus of Melrose and a regular part
involving the large time of the heat operators controlled by the exponential decay of the heat semigroup.
This gives equation~\ref{eq:differentialsdetheat} representing differentials of $\det_\zeta$.
Finally, the computation of differentials of $\log\det_\zeta$  
for perturbations with disjoint supports in subsubsection~\ref{ss:disjsupport} 
yields Proposition~\ref{p:functionalder}
which gives a characterization of the Schwartz kernels
of $D^n\log\det_\zeta$ in terms of the Feynman amplitudes $t_n$ introduced in 
subsubsection~\ref{sss:feynmanampl}. Furthermore, these results allow us to conclude
the proof of the factorization formula of Theorem~\ref{t:quillenconjzeta} in subsection~\ref{ss:factozeta}.
The bounds on the wave fronts are discussed in the next section~\ref{s:wfbound}.

\subsubsection{Analyticity.}\label{ss:analyticity}

The results of the present subsection are general enough 
to apply in both bosonic and fermionic cases.  

Differential operators of order $1$ on $M$ have a canonical structure of Fr\'echet space so it makes sense to talk about holomorphic curves in
$\text{Diff}^1(M,E)$. By~\cite[Thm 5.7 p.~215]{BK1}: 
\begin{prop}
For every holomorphic family $z\in \Omega\subset \mathbb{C}\mapsto  V_z\in \text{Diff}^1(M,E)$ s.t. $\pi$ is an Agmon angle of $\Delta+V_z, \forall z\in \Omega$, in particular $\Delta+V_z $ is invertible and all $\theta\in [\pi-\varepsilon,\pi+\varepsilon]$ are principal angles of $\Delta+V_z$,
the composition $$z\in \Omega\mapsto \text{det}_\zeta(\Delta+V_z)\in \mathbb{C}$$ is holomorphic.
\end{prop}
In~\cite{kontsevichdeterminants,BK1}, holomorphicity of the zeta determinant along one parameter holomorphic
families of differential operators is established but where the full symbol is allowed to vary which is stronger than what is needed here. 
We now discuss how parametric holomorphicity implies holomorphicity for $\det_\zeta$ in the sense of definition~\ref{d:analfunfrechet}.
We next state a result proved in Kriegl--Michor~\cite[Thm 7.19 p.~88, Thm 7.24 p.~90]{KM97} which gives a simple criteria which 
implies
holomorphicity in the sense of
definition~\ref{d:analfunfrechet}.
\begin{prop}[Holomorphicity by curve testing]
\label{p:krieglmichor}
Let $E$ be a Fr\'echet space over $\mathbb{C}$. Then
given a map $F:E\mapsto \mathbb{C}$, the following statements are equivalent:
\begin{itemize}
\item For any holomorphic curve
$\gamma:\mathbb{D}\subset \mathbb{C}\mapsto E$, the composite
$F\circ \gamma:\mathbb{D}\mapsto \mathbb{C}$ is holomorphic in the usual sense,
\item $F$ is holomorphic in the sense of definition~\ref{d:analfunfrechet}.
\end{itemize}
\end{prop}
Then by Proposition~\ref{p:krieglmichor}, this implies
that
the maps
$V\mapsto  \det_\zeta\left(\Delta+V \right) $ in the bosonic case, and
$A\mapsto \det_\zeta(\Delta+D^*A)$ in the fermionic case,
are holomorphic for $V\in C^\infty(End(E))$ and $A\in C^\infty(Hom(E_+;E_-))$
close enough to $0$.

\subsection{Perturbations with a spectral gap.}
\label{ss:perturbgap}
 We need to consider small perturbations $\Delta+B$ of some positive definite generalized Laplacian $\Delta$  by differential operators $B$ of order $1$  s.t.
the corresponding heat semigroup $e^{-t(\Delta+B)}$ has exponential decay. 
Let us recall briefly some classical properties of such perturbations.
The operator $\Delta+B:C^\infty(M)\mapsto H^s(M)$ admits a closed extension $H^{s+2}(M)\mapsto H^s(M)$
for every real $s\in \mathbb{R}$ 
by ellipticity of $\Delta+B\in \Psi^2(M)$. 
Since $(\Delta+B-z)^{-1}$ is compact and the 
resolvent set is non empty (see Lemma~\ref{l:gapspec1}),
by holomorphic Fredholm theory,
the family $(\Delta+B-z)^{-1}$ is a meromorphic family of Fredholm operators with poles of finite multiplicity corresponding to 
the discrete spectrum of $\Delta+B:H^{s+2}(M)\mapsto H^s(M)$.
In fact, it has been proved by Agranovich--Markus that 
such perturbations $\Delta + B$
have their spectrum contained in a
parabolic region containing the real axis. 

Our goal here is to bound the spectrum for small perturbations to show that the spectral gap
of $\Delta$ is maintained. This will imply exponential decay of the semigroup $e^{-t(\Delta+B)}$. 
We assume that $\Delta$ is a positive, self-adjoint
generalized Laplacian 
hence there is  $\delta>0$ such that 
$\sigma\left(\Delta\right)\geqslant\delta$.
We have the following Lemma proved in appendix:
\begin{lemm}\label{l:gapspec1}
There exists some neighborhood $\mathcal{U}$ of $0\in \text{Diff}^1(M,E)$
such that for every $B\in \mathcal{U}$, $\sigma\left(\Delta+B \right)\subset\{ Re(z)\geqslant \frac{\delta}{2}, Im(z)\leqslant \frac{\delta}{2} \} $.
\end{lemm}
So we control the spectrum of the perturbed operator in some $\frac{\delta}{2}$ neighborhood of the half--line
$[\delta,+\infty)$. 

Therefore, we can specialize the above Lemma to our particular situations as:
\begin{lemm}\label{l:specgap2}
In the bosonic case, there exists some open 
neighborhood $\mathcal{U}\subset C^\infty(End(E))$ of $0$ such that for
all small perturbations  
$V\in \mathcal{U}$, $\Delta+V$ is invertible and 
$\sigma\left(\Delta+V \right)\subset \{Re(z)\geqslant \frac{\delta}{2}\}$.

 In the fermionic case, the discussion is similar. $D$ invertible implies that $\Delta=D^*D$ is positive self-adjoint, $\sigma\left( \Delta\right)\geqslant\delta>0$. It follows that there is some open subset $\mathcal{U}\subset C^\infty(Hom(E_+,E_-))$ s.t. $\forall A\in \mathcal{U}$, 
$\sigma\left(D^*\left(D+A \right) \right)\subset \{ Re(z)\geqslant \frac{\delta}{2} \}$.
\end{lemm}

In the sequel, until subsection~\ref{ss:factozeta}, we shall take 
small enough perturbations in the neighborhood $\mathcal{U}$ given by Lemma~\ref{l:specgap2} 
so that the semigroups $e^{-t(\Delta+V)}$ and $e^{-t(D^*(D+A))}$ are analytic semigroups with exponential decay and 
the spectras $\sigma(\Delta+V)$ and $\sigma(D^*(D+A))$ are contained in the half--plane $Re(z)>0$.

\subsubsection{Taking the $\log$ of $\det_\zeta$.}

In this situation, we can take the $\log\det_\zeta$. However our results on the functional
derivatives of $\log\det_\zeta$ do not depend on the chosen branch of the logarithm. Since $\det_\zeta$ is well--defined for all
perturbations and is holomorphic, all our identities on G\^ateaux differentials that are proved for small perturbations will become automatically valid for any perturbations by analytic continuation.

\subsection{Resolvent bounds and exponential decay.}
We also prove
for small perturbations that we have some nice sectorial estimates on the resolvent which allow to apply the Hille--Yosida Theorem for
analytic semigroups.
\begin{lemm}\label{l:resolventbound}
Let $B\in \mathcal{U}\subset\text{Diff}^1(M,E) $ where $\mathcal{U}$ is the open set from Lemma~\ref{l:gapspec1}. Then
there exists a convex angular sector $\mathcal{R}=\{\arg(z)\in [-\theta,\theta]\}, \theta\in (0,\frac{\pi}{2})$ 
containing the half--line $[0,+\infty)$, $R>0$ 
s.t. the resolvent $(\Delta+B-z)^{-1}$ exists when $\vert z\vert \geqslant R, z\notin \mathcal{R}$ and satisfies a bound of the form:
\begin{equation}\label{e:resolventbound} 
\Vert \left(\Delta+B-z\right)^{-1}\Vert_{\mathcal{B}(L^2,L^2)} \leqslant K\text{dist}\left(z,\mathcal{R} \right)
\end{equation} 
for some $K>0$.
\end{lemm}
These sectorial bounds are straightforward consequences of the more general~\cite[Theorem 9.2 p.~85 and Theorem 9.3 p.~86]{Shubin}
since the operator $\Delta+B-z$ is elliptic with parameter $z$ of order $2$ in the sense of Shubin~\cite[p.~79]{Shubin} on $M\times \Lambda$
where $\Lambda\subset \mathbb{C}$ is any closed cone containing $(-\infty,0]$ avoiding the half-line $[\frac{\delta}{2},+\infty)$. 

%

%
A consequence of the above estimate also reads:
\begin{prop}\label{p:speccut}
Under the assumptions of the previous Lemma,
only a finite number of eigenvalues of $\Delta+B$ lies outside any conic neighborhood of the half--line $[0,+\infty)$.
\end{prop}

The resolvent bound on $ (\Delta+B-z)^{-1}:L^2(M) \mapsto L^2(M)  $
from Lemma~\ref{l:resolventbound} immediately implies
by the Hille--Yosida theorem for analytic semigroups~\cite[Theorem 4.22 p.~36]{Hairer}:
\begin{prop}
Let $B\in \mathcal{U}\subset\text{Diff}^1(M,E) $ where $\mathcal{U}$ is the open set from Lemma~\ref{l:gapspec1}.
Then there exists a unique strongly continuous heat 
semigroup $e^{-t(\Delta
+B)}:L^2(M) \mapsto L^2(M)$ generated by $\Delta+B$ which satisfies exponential decay estimates of the form
$$\Vert e^{-t(\Delta+B)}\Vert_{\mathcal{B}(L^2,L^2)}\leqslant Ce^{-\frac{\delta}{2} t} .$$
\end{prop}

\subsection{Relating the heat kernel and complex powers.}

In the sequel, we will strongly rely on the 
formulation of complex powers in terms of the heat kernel.
To make this correspondance precise, we shall
need the following Proposition proved in appendix: 
\begin{prop}\label{p:heatcomplex} 
Let $B\in \mathcal{U}\subset\text{Diff}^1(M,E) $ where $\mathcal{U}$ is the open set from Lemma~\ref{l:gapspec1},
for every $Q\in \text{Diff}^1(M,E)$, we have the following identity: 
\begin{eqnarray}
\boxed{Tr_{L^2}\left(Q(\Delta+B)^{-s} \right)=\frac{1}{\Gamma(s)}\int_0^\infty Tr_{L^2}\left(Qe^{-t(\Delta+B)} \right) t^{s-1}dt.}
\end{eqnarray}
\end{prop}

We shall use the above formula to represent differentials of $\log\det_\zeta$ 
in terms of heat operators in the next subsection.

\subsection{G\^ateaux differentials.}

Inspired by the 
nice exposition 
in Chaumard's thesis~\cite[p.~31-32]{chaumard2003},
we calculate the derivatives in $z$ of $\log\det_\zeta\left(\Delta+zV \right)$ 
near $z=0$ 
and we find in the bosonic case that 
for $n>\frac{d}{2}$, the derivative
of order $n$ of $z\mapsto\log\det_\zeta(\Delta+zV)$ at $z=0$ equals 
$(-1)^{n-1}(n-1)!Tr_{L^2}\left((\Delta^{-1}V)^n \right) $ 
where the $L^2$--trace is well--defined, in the fermionic case a similar result holds 
true for $n>d$. 

We introduce a method which allows to simultaneously 
calculate the functional 
derivatives of $\log\det_\zeta$ and bound the wave front set of
their Schwartz kernels. 
We start by using the following:
\begin{prop}
For any analytic family $(V_t)_{t\in \mathbb{R}^n}$ of perturbations, setting
$A_t=\Delta+V_t$
we know that
$Tr(\Delta+V_t)^{-s}$ is \textbf{holomorphic} near $s=0$ and depends smoothly on $t=(t_1,\dots,t_n)\in \mathbb{R}^n$,
and satisfies the variation formula:
\begin{eqnarray}\label{e:eqfirstderivative}
\boxed{\frac{d}{dt_i}Tr\left(A_t^{-s} \right)=-sTr\left( \frac{dV}{dt_i} A_t^{-s-1} \right), i\in \{1,\dots,n\}}
\end{eqnarray}
which is valid away from the poles of the analytic continuation in $s$ of $Tr_{L^2}\left( A_t^{-s} \right)$
hence the above equation holds true near $s=0$.
\end{prop}
\begin{proof}
In fact, the claim of our proposition is 
identical to \cite[Theorem d) (1.12.2) p.~108]{Gilkey}
except Gilkey states his results \textbf{only for positive definite, self-adjoint} operators $\Delta+V$ whereas we need it for small non-self-adjoint perturbations $V\in C^\infty(End(E))$.
We need to choose $V\in \mathcal{U}\subset C^\infty(End(E))$ where $\mathcal{U}$ is some sufficiently small neighborhood of $0$ such that
the semigroup
$e^{-t(\Delta+V)}$ has exponential decay, $\mathcal{U}$ is given by Lemma~\ref{l:gapspec1}.
The proof in our non-self-adjoint case 
follows almost verbatim from Gilkey's proof since his proof relies on:
\begin{itemize}
\item the asymptotic expansion of the heat kernel for a smooth family of non-self-adjoint generalized Laplacians~\cite[Lemma 1.9.1 p.~75]{Gilkey}, the smoothness of the 
terms in the asymptotic expansion in the parameters and the variation formula for the resolvent and heat kernel which are proved in~\cite[Lemma 1.9.3 p.~77]{Gilkey} in the non-self-adjoint case, we also refer to
~\cite[section 2.7]{BGV} for similar results, 
\item the result of \cite[Lemma 1.12.1 p.~106]{Gilkey} which still applies to $Tr_{L^2}(e^{-t(\Delta+V)})$,
\item the identities
\begin{eqnarray}
Tr_{L^2}\left(Q(\Delta+V)^{-s} \right)=\frac{1}{\Gamma(s)} \int_0^\infty t^{s-1}Tr_{L^2}\left(Qe^{-t(\Delta+V)} \right)dt , \forall Q\in \text{Diff}^1
\end{eqnarray}
which we established in Proposition~\ref{p:heatcomplex},
and 
\begin{eqnarray}
Tr_{L^2}\left(Qe^{-t(\Delta+V)} \right)\sim \sum_{n}a_n(Q,\Delta+V)t^{\frac{n-\dim(M)-1}{2}}
\end{eqnarray}
which is established in~\cite[Lemma 1.9.1 p.~75]{Gilkey}.
\end{itemize}
\end{proof}

The holomorphicity of $Tr(\Delta+V_t)^{-s}$ implies the Laurent series
expansion $Tr(\Delta+V_t)^{-s}=\sum_{k=0}^\infty a_k(V_t)s^k$ near $s=0$. 
By definition $\log\det_\zeta(\Delta+V_t)=-\frac{d}{ds}|_{s=0}Tr(\Delta+V_t)^{-s} $
which implies that
\begin{eqnarray*}
&&\frac{d}{dt_i}\log\text{det}_\zeta(\Delta+V_t)=-\frac{d}{dt_i}\frac{d}{ds}|_{s=0}Tr(\Delta+V_t)^{-s}=-\frac{da_1(V_t)}{dt_i}\\
&=&-\frac{d}{ds}|_{s=0}\sum_{k=0}^\infty \frac{da_k(V_t)}{dt_i}s^k=-\frac{d}{ds}|_{s=0}\frac{d}{dt_i}Tr(\Delta+V_t)^{-s}=\frac{d}{ds}|_{s=0}sTr\left(  \frac{dV}{dt_i} A_t^{-s-1}\right) . 
\end{eqnarray*}
Thus for higher derivatives, using equation~(\ref{e:eqfirstderivative}), we immediately 
deduce that
\begin{eqnarray*}
\frac{d}{dt_1}\dots \frac{d}{dt_{n+1}}\log\text{det}_\zeta(\Delta+V_t)&=&\frac{d}{ds}|_{s=0}s\frac{d}{dt_1}\dots \frac{d}{dt_{n}}Tr\left( \frac{dV_t}{dt_{n+1}} A_t^{-s-1} \right)\\
&=&FP|_{s=0}\frac{d}{dt_1}\dots \frac{d}{dt_{n}}Tr\left( \frac{dV_t}{dt_{n+1}} A_t^{-s-1} \right) 
\end{eqnarray*}
where the finite part $FP$ of a meromorphic germ at $s=0$ is defined to be the constant term in the Laurent series expansion about $s=0$.

So specializing the above identity to the family
$t\in \mathbb{R}^n\mapsto V+t_1V_1+\dots+t_{n+1}V_{n+1}$,
we find a preliminary formula for the G\^ateaux differentials
of $\log\det_\zeta$ for bosons:
{\small
\begin{eqnarray}\label{e:eqhigherderivatives}
\boxed{D^{n+1}\log\text{det}_\zeta(\Delta+V,V_1,\dots,V_{n+1})=FP|_{s=0}D^{n}Tr\left(\left(\Delta + V\right)^{-s-1} V_{n+1}\right)\left(V_1,\dots,V_n\right)}
\end{eqnarray}
}
and for fermions
{\small
\begin{eqnarray*}
\boxed{D^{n+1}\log\text{det}_\zeta(\Delta+D^*A_0,A_1,\dots,A_{n+1})=FP|_{s=0}D^nTr \left(\left(\Delta + D^*A_0\right)^{-s-1} D^*A_{n+1}\right)\left(A_1,\dots,A_n\right).}
\end{eqnarray*}
}
The heavy notation $D^{n}Tr\left(\left(\Delta + V\right)^{-s-1} V_{n+1}\right)\left(V_1,\dots,V_n\right)$ means the 
$n$-th G\^ateaux differential of the function $V\mapsto Tr\left(\left(\Delta + V\right)^{-s-1} V_{n+1}\right)$ in the directions $(V_1,\dots,V_n)$.

At this level of generality the formulas in both bosonic and fermionic cases are very similar just replacing
$V$ by $D^*A$ gives the fermionic formulas.

\subsubsection{Inverting traces and differentials.}

There is a subtlety which motivates our next discussion. We would like to invert 
the G\^ateaux differentials and the trace. In the next part, we shall prove estimates which allow to
carefully justify the inversion of $Tr$ and G\^ateaux differentials.

To calculate more explicitely the G\^ateaux differentials on the r.h.s of equation \ref{e:eqhigherderivatives}, we shall study in more details
the analytic map
$V\mapsto (\Delta+V)^{-s-1}\in \mathcal{B}(L^2,L^2)$ in both bosonic and fermionic situations and try to represent this analytic map in terms of heat kernels.  

\subsubsection{From the heat operator $e^{-t(\Delta+V)}$ to $(\Delta+V)^{-s-1}$ as analytic functions of $V$.}
Assume $\Delta$ is a generalized Laplacian, not necessarily symmetric,
s.t. $\sigma\left(\Delta \right)\subset \{Re(z)\geqslant \delta>0\}$.
By Duhamel formula, the heat operator $e^{-t(\Delta+V)}$ can be expressed
in terms of $e^{-t\Delta}$ as the Volterra series: 
\begin{eqnarray}\label{e:dysonvolterra}
 e^{-t(\Delta+V)}=\sum_{k=0}^\infty (-1)^k \int_{t\Delta_k} e^{-(t-t_k)\Delta}V\dots Ve^{-t_1\Delta}  
\end{eqnarray} 
where the series converges absolutely in $\mathcal{B}\left(L^2,L^2 \right)$ since in the bosonic case, we have the bound
$$ \Vert \int_{t\Delta_k} e^{-(t-t_k)\Delta}V\dots Ve^{-t_1\Delta}   \Vert_{\mathcal{B}\left(L^2,L^2 \right)} \leqslant e^{-t\delta} \frac{ t^k\Vert V\Vert_{\mathcal{B}\left(L^2,L^2 \right)}^k}{k!} . $$
Since $t\mapsto e^{-t\Delta}$ is a strongly continuous semigroup, it is easy to see that 
the series on the r.h.s of~\ref{e:dysonvolterra} is strongly continuous and defines a solution to the operator equation
\begin{eqnarray}
\frac{dU}{dt}=-(\Delta+V)U(t):C^\infty(M)\mapsto L^2 \text{ with } U(0)=Id
\end{eqnarray}
which defines uniquely the semigroup $e^{-t(\Delta+V)}$ hence this justifies that both sides are equal.

In the fermionic case, the convergence is slightly more subtle.
We start from the bound

\begin{lemm}\label{l:decayheatdirac}
Assume that $\Delta=D^*\left(D+A_0 \right)\in \mathcal{A}$ is a
generalized Laplacian s.t.
$\sigma\left(\Delta \right)\subset
\{Re(z)\geqslant \delta >0\}$.
For any differential operator $P$ of order $1$,
\begin{eqnarray}
 \Vert e^{-t\Delta} P \Vert_{ \mathcal{B}\left(L^2,L^2 \right)   } \leqslant Ct^{-\frac{1}{2}}e^{-\frac{t}{2}\delta}. 
\end{eqnarray}
\end{lemm}
\begin{proof}
Assume that $\Delta=D^*\left(D+A_0 \right)\in \mathcal{A}$ is a
generalized Laplacian s.t.
$\sigma\left(\Delta \right)\subset
\{Re(z)\geqslant \delta >0\}$ hence 
$\Delta$ is not necessarily self-adjoint.
We shall use the following ingredient. 
In Proposition~\ref{p:heatcomplex}, we proved that fractional powers of $\Delta=D^*\left(D+A_0\right)$ defined as 
$$\Delta^{-s}=\frac{1}{\Gamma(s)} \int_0^\infty e^{-t\Delta} t^{s-1} dt $$ coincide 
with complex powers of $\Delta$ defined by the contour integrals of the resolvent as in the work of Seeley.
Note that 
by results of Seeley, the powers $\Delta^s, s\in \mathbb{R}$ are
well--defined elliptic pseudodifferential operators of degree $2s$. Therefore 
$$ \Vert \Delta^{\frac{1}{2}}u \Vert_{L^2}=\Vert \Delta^{\frac{1}{2}}A^{-\frac{1}{2}} A^{\frac{1}{2}} u \Vert_{L^2}\leqslant \Vert \Delta^{\frac{1}{2}}A^{-\frac{1}{2}}\Vert_{ \mathcal{B}(H^{-1},H^{-1})}   \Vert A^{\frac{1}{2}} u \Vert_{L^2}=\Vert \Delta^{\frac{1}{2}}A^{-\frac{1}{2}}\Vert_{ \mathcal{B}(H^{-1},H^{-1})}  \Vert u \Vert_{H^{-1}},$$
where $\Vert \Delta^{\frac{1}{2}}A^{-\frac{1}{2}}\Vert_{ \mathcal{B}(H^{-1},H^{-1})} <+\infty $ since $\Delta^{\frac{1}{2}}A^{-\frac{1}{2}}$ is a pseudodifferential operator of order $0$ by composition.
Recall that under our assumption of taking small perturbations
of a positive definite self--adjoint $D^*D$, $e^{-t\Delta}$ generates an analytic semigroup.
\begin{eqnarray*}
\Vert Pe^{-t\Delta} \Vert_{\mathcal{B}(L^2,L^2)}\leqslant \underset{ \leqslant e^{-\frac{t}{2}\delta} }{\underbrace{\Vert e^{-\frac{t}{2}\Delta} \Vert_{\mathcal{B}(L^2,L^2)}}} 
\Vert e^{-\frac{t}{2}\Delta}\Delta^{\frac{1}{2}} \Vert_{\mathcal{B}(L^2,L^2)}
\underbrace{\Vert \Delta^{-\frac{1}{2}} P \Vert_{\mathcal{B}(L^2,L^2)}}
\end{eqnarray*}
where the term underbraced is bounded by a constant since $\Delta^{-\frac{1}{2}} P\in \Psi^0 $ by pseudodifferential calculus is bounded on every Sobolev space. We are now reduced to estimate the term 
$\Vert e^{-t\Delta}\Delta^{\frac{1}{2}} \Vert_{\mathcal{B}(L^2,L^2)}$. Then
a straightforward application of~\cite[Proposition 4.36 p.~40]{Hairer}, which is valid for generators of analytic semigroups whose spectrum is contained in the right half-plane, yields the estimate $\Vert e^{-t\Delta}\Delta^{\frac{1}{2}} \Vert_{\mathcal{B}(L^2,L^2)}\leqslant Ct^{-\frac{1}{2}} $ hence this implies the final result. 
\end{proof}

Therefore setting $V=D^*A$, we find that the series on the r.h.s of identity
(\ref{e:dysonvolterra}) 
converges absolutely in $\mathcal{B}\left(L^2,L^2 \right)$
by the bound 
\begin{eqnarray*}
&& \Vert \int_{t\Delta_k} e^{-(t-t_k)\Delta}D^*A\dots D^*Ae^{-t_1\Delta}   \Vert_{\mathcal{B}\left(L^2,L^2 \right)}\\
& \leqslant &\Vert A\Vert_{\mathcal{B}\left(L^2(E_+),L^2(E_-) \right)}^ke^{-\frac{t}{2}\delta} C^k \int_0^t \vert t-t_k \vert^{-\frac{1}{2}}
\dots \int_0^{t_{2}}  \vert t_1 \vert^{-\frac{1}{2}}  dt_1\dots dt_{k}    \\
&= & \Vert A\Vert_{\mathcal{B}\left(L^2(E_+),L^2(E_-) \right)}^ke^{-\frac{t}{2}\delta} \frac{(k+1)C^k t^k\Gamma(\frac{1}{2})^{k+1}}{\Gamma(\frac{k+3}{2})},
\end{eqnarray*}
where $C$ is the constant of Lemma \ref{l:decayheatdirac} and the r.h.s. is the general term of some convergent series by the asymptotic behaviour of the Euler 
$\Gamma$ function
\footnote{Rewrite
$\int_0^t \vert t-t_k \vert^{-\frac{1}{2}}
\dots \int_0^{t_{2}}  \vert t_1 \vert^{-\frac{1}{2}}  dt_1\dots dt_{k}=\int_{\{u_1+\dots+u_{k+1}=t\}} \prod_{i=1}^{k+1} u_i^{-\frac{1}{2}} d\sigma $,
$d\sigma$ measure on the simplex $\{u_1+\dots+u_{k+1}=t\}$. Note that
$\int_{u_1+\dots+u_{k+1}=t} \prod_{i=1}^{k+1} u_i^{-\frac{1}{2}} d\sigma=\frac{d}{dt}\int_{\{u_1+\dots+u_{k+1}\leqslant t\}} \prod_{i=1}^{k+1} u_i^{-\frac{1}{2}}du_i=\frac{d}{dt} \frac{ t^{k+1}\Gamma(\frac{1}{2})^{k+1}}{\Gamma(\frac{k+3}{2})} $
following a beautiful identity due to Dirichlet }.

Furthermore using the Hadamard--Fock--Schwinger formula proved in Proposition~\ref{p:heatcomplex},
for $Re(s)>0$, we find that
\begin{eqnarray*}
(\Delta+V)^{-s-1}V_{n+1}&=&\sum_{k=0}^\infty \frac{(-1)^k}{\Gamma(s+1)}\int_0^\infty t^s  \int_{t\Delta_k} e^{-(t-t_k)\Delta}V\dots Ve^{-t_1\Delta}V_{n+1}   dt\text{ for bosons},\\
(\Delta+D^*A)^{-s-1}D^*A_{n+1}&=&\sum_{k=0}^\infty \frac{(-1)^k}{\Gamma(s+1)}\int_0^\infty t^s  \int_{t\Delta_k} e^{-(t-t_k)\Delta}D^*A\dots D^*Ae^{-t_1\Delta}D^*A_{n+1}   dt\text{ for fermions},\\
\end{eqnarray*}
where both series converge absolutely in $V\in \mathcal{B}\left(L^2(E),L^2(E) \right)$ (resp $A\in \mathcal{B}(L^2(E_+),L^2(E_-))$) by the above bounds
since we have exponential decay in $t$. 

From the above identities, we find that
$$ D^{n+1}\log\text{det}_\zeta\left(\Delta+V \right)(V_1,\dots,V_{n+1})=FP|_{s=0}\frac{d}{dt_1}\dots\frac{d}{dt_{n}}
Tr(\Delta+V_t)^{-s-1}V_{n+1} ) |_{t=0}$$ for the family $V_t=t_1V_1+\dots+t_nV_n$ and by definition of the G\^ateaux differential.
\begin{eqnarray*}
&&\frac{d}{dt_1}\dots\frac{d}{dt_{n}}
Tr((\Delta+V_t)^{-s-1}V_{n+1} ) |_{t=0}\\
&=&\frac{d}{dt_1}\dots\frac{d}{dt_{n}}
Tr(\sum_{k=0}^\infty \frac{(-1)^k}{\Gamma(s+1)}\int_0^\infty t^s  \int_{t\Delta_k} e^{-(t-t_k)\Delta}V_t\dots V_te^{-t_1\Delta}V_{n+1}   dt ) |_{t=0} \\
&=& \frac{(-1)^nn!}{\Gamma(s+1)}Tr (\int_0^\infty t^s  \int_{t\Delta_{n}} e^{-(t-t_n)\Delta}V_1\dots V_ne^{-t_1\Delta}V_{n+1}   dt ). 
\end{eqnarray*}
~\footnote{The combinatorial factor $n!$ comes from 
the fact that for every symmetric polynomial $S_n:E^n\mapsto \mathbb{C}$ homogeneous of degree $n$ and 
$v_t=t_1v_1+\dots+t_nv_n; t=(t_1,\dots,t_n)\in \mathbb{R}^n$, $\frac{d}{dt_1}\dots\frac{d}{dt_{n}}S_n(v_t,\dots,v_t)|_{t=0}=n!S_n(v_1,\dots,v_n)$.} 
A similar identity holds true for Fermions.

\begin{lemm}[G\^ateaux differentials of $\log\det_\zeta$]
\label{l:funder1}
Following the notations from definitions (\ref{d:bosonsfermionszeta}) and (\ref{d:speczetadet}).
Let $M$ be a smooth, closed, compact Riemannian manifold of dimension $d$, 
and $\Delta$ (resp $\Delta=D^*D$) some generalized Laplacian acting on $E$ (resp $E_+$) s.t. $\sigma\left(\Delta \right)\subset \{ Re(z)\geqslant \delta>0 \}$ for bosons (resp fermions).
The G\^ateaux differentials
of $\log\det_\zeta$ satisfy the following identities. For bosons,   
for every $(V_1,\dots, V_{k+1})\in L^\infty(M,End(E))^{k+1}$, 
{\small
\begin{eqnarray*}
&&\frac{1}{k!}
D^{k+1}\log\text{det}_\zeta(\Delta+V,V_1,\dots,V_{k+1})|_{V=0}\\
&=&FP|_{s=0} \frac{(-1)^k}{\Gamma(s+1)}
Tr\left(
\int_{[0,\infty)^{k+1}} (\sum_{e=1}^{k+1} u_{e})^s  e^{-u_{k+1}\Delta}V_1\dots e^{-u_1\Delta}V_{k+1} \prod_{e=1}^{k+1} du_e\right).
\end{eqnarray*}
}
For fermions, for
every $(A_1,\dots,A_{k+1})\in L^\infty(M,Hom(E_+,E_-))$
{\small
\begin{eqnarray*}
&&\frac{1}{k!} D^{k+1}\log\text{det}_\zeta(\Delta+D^*A,A_1,\dots,A_{k+1})|_{A=0}\\
&=&FP|_{s=0} \frac{(-1)^k}{\Gamma(s+1)}
Tr\left(  \int_{[0,\infty)^{k+1}} (\sum_{e=1}^{k+1} u_{e})^s e^{-u_{k+1}\Delta}D^*A_1\dots e^{-u_1\Delta}D^*A_{k+1} \prod_{e=1}^{k+1} du_e\right).
\end{eqnarray*}
}
\end{lemm}
We want to determine more explicitely 
the above G\^ateaux differentials and invert the $L^2$-trace and the integral.
To justify this inversion, it suffices
to prove that
for $Re(s)>\frac{d+1}{2}$,  for every integer $k$, the 
operator valued integral 
$\int_0^\infty t^s  \int_{t\Delta_k} e^{-(t-t_k)\Delta}V\dots Ve^{-t_1\Delta}V   dt$ 
represents a trace class operator with continuous 
kernel. In this case, the $L^2$ trace of this operator 
coincides with the flat trace of this operator.
Then it is enough to prove that the integral
$\int_{[0,\infty)^{k+1}} (\sum_{e=1}^{k+1} u_{e})^s 
Tr_{L^2}\left( e^{-u_{k+1}\Delta}V_1\dots e^{-u_1\Delta}V_{k+1}\right) \prod_{e=1}^{k+1} du_e$
converges absolutely by carefully estimating the operator traces 
$ Tr_{L^2}\left( e^{-u_{k+1}\Delta}V_1\dots e^{-u_1\Delta}V_{k+1}\right)$.
We use the fact that outside some subset of measure $0$, the operator 
$\left( e^{-u_{k+1}\Delta}V_1\dots e^{-u_1\Delta}V_{k+1}\right)$ is smoothing and depends continuously on $(u_1,\dots,u_{k+1})$
therefore both $L^2$ and flat trace coincide.

\subsection{Bounding operator traces.}
\label{ss:optraces}
Our next task is to make sense and study the analyticity properties of the integrals
of the form
$$\int_{[0,\infty)^{k+1}} (\sum_{e=1}^{k+1} u_e)^s  Tr\left(e^{-u_{k+1}\Delta}V_1\dots e^{-u_1\Delta}V_{k+1} \right) \prod_{e=1}^{k+1} du_e =I(s,V_1,\dots , V_{k+1})$$
for $s$ near $0$ and also determine the holomorphicity domain in $s$ as a function of $k\in \mathbb{N}$.

\subsubsection{A decomposition.}

As in~\cite{dangzhang}, the strategy
relies on methods from quantum field theory: 
using the symmetries of the integrand
by permutation of variables, we integrate on a simplex $\{u_{k+1}\geqslant\dots\geqslant u_1\geqslant 0\}$ called Hepp's sector:
\begin{eqnarray*}
&&\int_{[0,\infty)^{k+1}} (\sum_{e=1}^{k+1} u_e)^s   e^{-u_{k+1}\Delta}V\dots e^{-u_1\Delta}V  \prod_{e=1}^{k+1} du_e \\
&=&(k+1)!\int_{\{u_{k+1}\geqslant\dots\geqslant u_{1}\geqslant 0\} } (u_1+\dots+u_{k+1})^s  e^{-u_{k+1}\Delta}V\dots e^{-u_1\Delta}V  du_1\dots du_{k+1} 
\end{eqnarray*}
We will show that the only divergence is in the variable $u_{k+1}$.
In the next definition, we cut the integral in two parts, $u_{k+1}\geqslant 1$ and $u_{k+1}\leqslant 1$.

\begin{defi}[Decomposition]\label{d:decomp}
Under the assumptions of Lemma \ref{l:funder1}.
We set
\begin{eqnarray}\label{e:keyIntegraltrace}
\boxed{I(s,V_1,\dots , V_{k+1})=\int_{[0,\infty)^{k+1}} (\sum_{e=1}^{k+1} u_e)^s  
Tr\left( e^{-u_{k+1}\Delta}V_1\dots e^{-u_1\Delta}V_{k+1}\right)  \prod_{e=1}^{k+1} du_e}
\end{eqnarray}
that we shall decompose in two pieces
\begin{eqnarray}
I(s;V_1,\dots,V_{k+1})=S(s;V_1,\dots,V_{k+1})+R(s;V_1,\dots,V_{k+1})
\end{eqnarray}
where 
{\small
\begin{equation}
R(s;V_1,\dots,V_{k+1})=\sum_{\sigma\in S_{k+1}}\int_{\{u_{k+1}\geqslant\dots\geqslant u_{1}\geqslant 0,u_{k+1}\geqslant 1 \} } (\sum_{e=1}^{k+1} u_e)^s  
Tr\left( e^{-u_{k+1}\Delta}V_{\sigma(1)}\dots e^{-u_1\Delta}V_{\sigma(k+1)} \right)  \prod_{e=1}^{k+1} du_e,
\end{equation}
}
and
{\small
\begin{equation}
S(s;V_1,\dots,V_{k+1})=\sum_{\sigma\in S_{k+1}}\int_{  \{1\geqslant u_{k+1}\geqslant\dots\geqslant u_{1}\geqslant 0\} } (\sum_{e=1}^{k+1} u_e)^s
Tr\left( e^{-u_{k+1}\Delta}V_{\sigma(1)}\dots e^{-u_1\Delta}V_{\sigma(k+1)} \right) \prod_{e=1}^{k+1} du_e. 
\end{equation}
}
\end{defi}

We use the above decomposition
for both bosons and fermions 
where 
$(V_i=D^*A_i, A_i\in C^\infty(Hom(E_+,E_-)))_{i=1}^{k+1}$ for fermions. 
The function $S$ (resp $R$) is the singular (resp regular) part of $I$.
We will later deal with 
the singular part $S$ using the heat calculus of Melrose~\cite[Chapter 7]{melrose1993atiyah}~\cite{grieserheat,dai2012adiabatic}.

\subsubsection{Controlling the regular part.}

We shall first show that the regular part $R$ has analytic continuation as holomorphic function in $s$ 
on the whole complex plane.
\begin{lemm}\label{l:boundsoperatortraces}
Following the notations from definitions (\ref{d:bosonsfermionszeta}) and (\ref{d:speczetadet}).
Let $M$ be a smooth, closed, compact Riemannian manifold of dimension $d$, 
and $\Delta$ (resp $\Delta=D^*D$) some generalized Laplacian acting on $E$ (resp $E_+$) s.t. $\sigma\left(\Delta \right)\subset \{ Re(z)\geqslant \delta>0 \}$ for bosons (resp fermions).
For every $(V_1,\dots, V_{k+1})\in C^\infty(End(E))^{k+1}$ in the bosonic case and for 
every $(A_1,\dots,A_{k+1})\in C^\infty(Hom(E_+,E_-))$ where $V_1=D^*A_1,\dots,V_{k+1}=D^*A_{k+1}$ in the fermionic case, 
the regular part $R(s;V_1,\dots,V_{k+1})$ has analytic continuation as a holomorphic function of $s\in \mathbb{C}$.
\end{lemm}
\begin{proof}
For $p> d$ and $B\in \mathcal{B}(L^2,H^{p})$, $B$ is trace class and 
satisfies the simple bound
$ \vert Tr_{L^2}(B) \vert \leqslant C\Vert B \Vert_{\mathcal{B}(L^2,H^{p})} $~\cite[Prop B 21 p.~502]{DyZwscatt}.
Hence in the bosonic case,
\begin{eqnarray*}
\vert Tr_{L^2}\left(e^{-u_{k+1}\Delta}V_1\dots e^{-u_1\Delta}V_{k+1} \right)\vert 
&\leqslant &\Vert e^{-\frac{1}{2}\Delta}\Vert_{\mathcal{B}(L^2,H^{p})}  \Vert e^{-(u_{k+1}-\frac{1}{2})\Delta}\Vert_{\mathcal{B}(L^2,L^2)} \prod_{i=1}^{k+1} \Vert V_i \Vert_{\mathcal{B}(L^2,L^2)}\\
&\leqslant & e^{-(u_{k+1}-\frac{1}{2})\delta}\Vert e^{-\frac{1}{2}\Delta}\Vert_{\mathcal{B}(L^2,H^{p})} \prod_{i=1}^{k+1} \Vert V_i \Vert_{\mathcal{B}(L^2,L^2)}
\end{eqnarray*}
the integrand has exponential decay which ensures the holomorphicity.

In the fermionic case where $(V_i=D^*A_i, A_i\in C^\infty(Hom(E_+,E_-)))_{i=1}^{k+1}$, 
the bound reads:
\begin{eqnarray*}
&&\vert Tr_{L^2}\left(e^{-u_{k+1}\Delta}V_1\dots e^{-u_1\Delta}V_{k+1} \right)\vert \\ 
&\leqslant &
\Vert e^{-\frac{1}{2}\Delta}\Vert_{\mathcal{B}(L^2,H^{p})} 
 \Vert e^{-(u_{k+1}-\frac{1}{2})\Delta}D^*\Vert_{\mathcal{B}(L^2,L^2)}
\prod_{i=1}^k \Vert e^{-u_i\Delta}D^* \Vert_{\mathcal{B}(L^2,L^2)}
\prod_{i=1}^{k+1}  \Vert A_i \Vert_{\mathcal{B}(L^2,L^2)}\\
&\leqslant & \sqrt{2}C^{k+1}\Vert e^{-\frac{1}{2}\Delta}\Vert_{\mathcal{B}(L^2,H^{p})}  e^{-(\frac{u_{k+1}}{2}-\frac{1}{4})\delta} \prod_{i=1}^k u_i^{-\frac{1}{2}}\prod_{i=1}^{k+1}  \Vert A_i \Vert_{\mathcal{B}(L^2,L^2)}
\end{eqnarray*}
where $C$ is the constant from Lemma \ref{l:decayheatdirac}, 
$\Vert e^{-\frac{1}{2}\Delta}\Vert_{\mathcal{B}(L^2,H^{p})} <+\infty$ since the heat kernel is smoothing and the r.h.s has exponential decay in $u_{k+1}$ which ensures holomorphicity.
\end{proof}

\subsubsection{Bounding the singular part.}\label{sss:boundsing}

It remains to deal with the term $S(s;V_1,\dots,V_{k+1})$ 
involving the integral for $u_{k+1}\in [0,1]$. Without loss of generality, we will discuss the case
$V=V_1=dots=V_{k+1}$ in the next two subsections, the general case follows by polarization.  
\subsubsection*{The case when $k=0$.}
In this simple case, for the bosonic case, we directly use the diagonal 
asymptotic expansion of the heat kernel~\cite[Lemma 1.8.2]{Gilkey}
$e^{-t\Delta}(x,x)\sim \sum_{k=0}^\infty \frac{a_k(x,x)t^{k-\frac{d}{2}}}{(4\pi)^{\frac{d}{2}}}$
which yields that
$\int_{0}^{+\infty} u^s Tr\left(e^{-u\Delta}V\right)du $ is holomorphic when $Re(s)>\frac{d}{2}-1$ since the integrand converges absolutely, it has analytic continuation as a meromorphic function in $s\in \mathbb{C}$ and also that
 \begin{eqnarray*}
&&FP|_{s=0}\frac{1}{\Gamma(s+1)} \int_{0}^{+\infty} u^s Tr\left(e^{-u\Delta}V\right)du
=Tr\left(e^{-\Delta }\Delta^{-1}V\right)
+ \int_0^1 dt \int_M r_{N+1}(t,x,x)V(x)dv \\
&+&
\sum_{k=0, k\neq \frac{d}{2}-1}^N \frac{\left(\int_{M}a_k(x,x)V(x)dv\right)}{(4\pi)^{\frac{d}{2}}\left(s+1+k-\frac{d}{2}\right)} - 
\frac{\Gamma^\prime(1)}{(4\pi)^{\frac{d}{2}}} \left(\int_{M}a_{\frac{d}{2}-1}(x,x)V(x)dv\right).
\end{eqnarray*}
This means there exists  $\ell\in C^\infty\left(End(E) \right)$ and $dv\in \vert \Lambda^{top}\vert M$ a smooth density such that $FP|_{s=0}\frac{1}{\Gamma(s+1)} \int_{0}^{+\infty} u^s Tr\left(e^{-u\Delta}V\right)du=\int_M\left\langle\ell,V \right\rangle dv$.
A similar result holds true in the fermionic case where
we find that
$FP|_{s=0}\frac{1}{\Gamma(s+1)} \int_{0}^{+\infty} u^s Tr\left(e^{-u\Delta}D^*A\right)du=\int_M \left\langle\ell,A\right\rangle dv$
where $\ell\in C^\infty(Hom(E_-,E_+))$ and $dv\in \vert \Lambda^{top}\vert M$.

\subsubsection*{When $k>0$.}

We use the formalism of the heat calculus 
of Melrose as exposed in
the work of Grieser~\cite{grieserheat} (see also~\cite[p.~62]{taylor2013partial} for related construction) 
whose 
notations are adopted. 
We start from the fact that in local coordinates,
$e^{-t\Delta}(x,y)=t^{-\frac{d}{2}}\tilde{A}(t,\frac{x-y}{\sqrt{t}},y)$ where $\tilde{A}\in C^\infty\left([0,\infty)_{\frac{1}{2}}\times \mathbb{R}^d\times U, E\boxtimes E^* \right)$ since the heat kernel 
is an element in $\Psi^{-1}_H(M,E)$~\cite[definition 2.1 p.~6]{grieserheat}.
Then,
we note that the $k+1$-fold composition $K\star\dots\star K$ belongs to
$\Psi_H^{-k-1}\left(M,E\right)$ by the composition Theorem in the heat calculus~\cite[Proposition 2.6 p.~8 ]{grieserheat}. 
Hence this means for every $p\in M$, there are local coordinates $U\ni p$ s.t.:
\begin{eqnarray*}
K^{\star (k+1)}(t,x,y)=t^{-\frac{d+2}{2}+(k+1)}\tilde{A}(t,\frac{x-y}{\sqrt{t}},y)
\end{eqnarray*}
where $\tilde{A}\in C^\infty\left([0,\infty)_{\frac{1}{2}}\times \mathbb{R}^d\times U,  E\boxtimes E^*\right)$ by definition of the
elements in the heat calculus $\Psi_H^\bullet\left(M,E\right)$.
Therefore by definition of $\star$, we find:
\begin{eqnarray*}
S(s;V,\dots,V)=\int_0^1 t^s  \int_{t\Delta_k} Tr\left(e^{-(t-t_k)\Delta}V\dots Ve^{-t_1\Delta}V \right)  dt=\int_0^1 t^s\left(\int_{M}\left( K^{\star k+1}\right)(t,x,x)dv \right) dt
\end{eqnarray*}
where $\int_0^1 t^s\left( K^{\star k+1}\right)(t,x,x)dt=\int_0^1 t^{s-\frac{d+2}{2}+(k+1)}\tilde{A}(t,0,x)dt$ in local coordinates
on $M$ and the r.h.s is Riemann integrable in $t$ and holomorphic in the domain $Re(s) > \frac{d+2}{2}-k-2 $ by Fubini. 
In particular, the term $S$ is holomorphic near $s=0$ as soon as $k+1>\frac{d}{2}$. 

In the fermionic case, we start from the fact that in local coordinates,
$e^{-t\Delta}(x,y)=t^{-\frac{d}{2}}\tilde{A}(t,\frac{x-y}{\sqrt{t}},y)$ where $\tilde{A}\in C^\infty\left([0,\infty)_{\frac{1}{2}}\times \mathbb{R}^d\times U, E_+\boxtimes E_+^* \right)$. 
From the observation that
$$ D_{y^i}t^{-\frac{d}{2}}\tilde{A}(t,\frac{x-y}{\sqrt{t}},y)=t^{-\frac{d+1}{2}}(y^i-x^i)\left(D_{X^i}\tilde{A}\right)(t,\frac{x-y}{\sqrt{t}},y)+t^{-\frac{d}{2}}\left(D_{y^i}\tilde{A}\right)(t,\frac{x-y}{\sqrt{t}},y) ,$$
we deduce that 
$K=e^{-t\Delta}D^*A\in \Psi_H^{-\frac{1}{2}}(M,E_+)$. We loose $\frac{1}{2}$ of regularity 
compared to the bosonic case since 
we do not consider the heat kernel alone but composed with a differential operator of order $1$.   
Then by composition in the heat calculus, 
we find that:
\begin{eqnarray*}
S(s;D^*A,\dots,D^*A)&=&\int_0^1 t^s  \int_{t\Delta_k} Tr\left(e^{-(t-t_k)\Delta}D^*A\dots e^{-t_1\Delta}D^*A \right)  dt\\
&=&
\int_0^1 t^s
\left(\int_{M}\left( K^{\star k+1}\right)(t,x,x)dv\right)dt
\end{eqnarray*}
where $\int_0^1 t^s\left( K^{\star k+1}\right)(t,x,x)dt=\int_0^1 t^{s-\frac{d+2}{2}+\frac{k+1}{2}}\tilde{A}(t,0,x)dt$
in local coordinates and 
the r.h.s is Riemann integrable in $t$, holomorphic in $s$ in the half--plane $Re(s)>\frac{d}{2}-\frac{k+1}{2}$ by Fubini and has analytic continuation as a meromorphic function in $s\in \mathbb{C}$.
In particular, the term $S$ is holomorphic near $s=0$ as soon as $k+1>d$.

In both cases, when $k+1>\frac{d}{2}$ in the bosonic case and $k+1>d$ in the fermionic case, 
we find that
$$\lim_{s\rightarrow 0}I\left(s;V_1,\dots,V_{k+1} \right)=\int_{[0,\infty)^{k+1}} 
Tr\left( e^{-u_{k+1}\Delta}V_1\dots e^{-u_1\Delta}V_{k+1}\right)  \prod_{e=1}^{k+1} du_e$$ 
where the integral on the right hand side is convergent. 
%
%

\subsubsection{Inverting integrals and traces in Lemma~\ref{l:funder1}.}
 From the estimates on the operator trace
$ Tr\left(e^{-u_{k+1}\Delta}V_1\dots e^{-u_1\Delta}V_{k+1} \right) $  of Lemma~\ref{l:boundsoperatortraces} controlling the exponential decay for large times
$u_{k+1}\in [1,+\infty)$ and of subsubsection~\ref{sss:boundsing} which bound the operator trace for small times $u_{k+1}\in [0,1]$, we conclude
that the integrals
$$\int_{[0,\infty)^{k+1}} (\sum_{e=1}^{k+1} u_e)^s  Tr\left(e^{-u_{k+1}\Delta}V_1\dots e^{-u_1\Delta}V_{k+1} \right) \prod_{e=1}^{k+1} du_e =I(s,V_1,\dots , V_{k+1})$$
converge for $Re(s)>\frac{d}{2}+1$. 
Therefore, this implies that we have the identity:
\begin{eqnarray}\label{eq:differentialsdetheat}
D^{n}\left(Tr\left(\left(\Delta + V\right)^{-s-1} V_{n+1}\right),V_1,\dots,V_n\right)=\frac{n!(-1)^n}{\Gamma(s+1)}I(s,V_1,\dots,V_{n+1}).
\end{eqnarray}
Combined with the analytic continuation near $s=0$ of both sides,  
this justifies the inversion 
of traces and integrals
in the formula for the G\^ateaux differentials of $\log\det_\zeta$ from Lemma~\ref{l:funder1}.

\subsection{Schwartz kernels of the G\^ateaux differentials of $\log\det_\zeta$.}

Up to now, we worked with operators and their compositions. Now, we will work with integral kernels
and products of operator kernels instead.
The goal of the present subsection is to perform an in depth study
of the Schwartz kernels $[\mathbf{D^n\log\det_\zeta}]$ of the $n$-th G\^ateaux differentials $D^n\log\det_\zeta$ and to show they are distributional extensions of some products of Green functions (either the Schwartz kernel of $\Delta^{-1}$ for bosons or the Schwartz kernel of $D^{-1}$ for fermions), which is a priori well--defined only on $M^n\setminus d_n$, to the whole configuration space $M^n$.

\subsubsection{G\^ateaux differentials with disjoint supports.}
\label{ss:disjsupport}

Let us start by discussing 
$[\mathbf{D^n\log\det_\zeta}]$ as a distribution on $M^n\setminus d_n$, so outside the deepest diagonal.
In this subsubsection,
we
assume that $(V_1,\dots,V_{k+1})$
are such that 
$\text{supp}\left( V_1\right)\cap\dots\cap \text{supp}\left( V_{k+1}\right)=\emptyset$. So the mutual support is empty.
We need the following Lemma whose proof is given in the appendix:  
\begin{lemm}[Microlocal convergence of heat operator]\label{l:convheat}
Let $e^{-t\Delta}$ be the heat operator. Then we have the convergence
$ e^{-t\Delta}\underset{t\rightarrow 0^+}{\rightarrow} Id\text{ in }\Psi^{+0}_{1,0}(M)$. 
\end{lemm}

%
This implies that
the family $(e^{-t\Delta}V_i)_{t\in [0,1]}$ defines a 
\textbf{bounded family of pseudodifferential operators} in $\Psi^{+0}_{1,0}(M,E)$ 
whose 
wave front set is uniformly controlled in $T^*_{\text{supp}(V_i)}M$.
This means that for 
every pair of  
cut--off functions
$(\chi_1,\chi_2)\in C^\infty(M)^2$ s.t. $\text{supp}(V_i)\cap \text{supp}(\chi_1)\cap \text{supp}(\chi_2)=\emptyset$, the family  
$ \left(\chi_2e^{-t\Delta}V_i\chi_1\right)_{t\in [0,+\infty)} $ is bounded in 
$\Psi^{-\infty}(M,E)$~\footnote{It means the
corresponding 
family of Schwartz kernels are bounded in $C^\infty(M\times M,E\boxtimes E^*)$ for the usual Fr\'echet topology}.
Otherwise,
$(\chi_2e^{-t\Delta}V_i\chi_1)_{t\in [0,+\infty)}$ is bounded in $\Psi^{+0}_{1,0}(M,E)$.
This implies that the family 
$$e^{-(t-t_k)\Delta}V_{k+1}\dots V_2e^{-t_1\Delta}V_1, 
\text{ for } \{ 0\leqslant t_1\leqslant\dots t_k\leqslant t\leqslant 1 \}$$
is bounded in $\Psi^{-\infty}(M,E)$ by the condition on the support of $(V_i)_{i=1}^{k+1}$.
Finally,
$$\int_{t\Delta_k} Tr\left(e^{-(t-t_k)\Delta}V_{k+1}\dots V_2e^{-t_1\Delta}V_1 \right)=\mathcal{O}(1)$$
and 
$\lim_{s\rightarrow 0}I(s,V_1,\dots,V_{k+1})=Tr_{L^2}\left(\Delta^{-1}V_1\dots\Delta^{-1}V_{k+1} \right)$
where the $L^2$-trace on the r.h.s is well--defined since $\Delta^{-1}V_1\dots\Delta^{-1}V_{k+1}\in \Psi^{-\infty}\left(M,E\right)$.
Hence, for every $(k+1)$-uple of open subsets
$\left(U_1,\dots,U_{k+1}\right)$ s.t.
$\overline{U_1}\cap\dots\cap \overline{U_{k+1}}=\emptyset$,
the multilinear map 
{\small
\begin{eqnarray*}
\left(V_1,\dots,V_{k+1}\right)\in C^\infty_c(U_1,End(E))\times \dots\times C^\infty_c(U_{k+1},End(E))\mapsto \lim_{s\rightarrow 0} \frac{I(s;V_1,\dots,V_{k+1})}{\Gamma(s+1)} =
Tr\left(\Delta^{-1}V_1\dots\Delta^{-1}V_{k+1} \right) 
\end{eqnarray*}
}
is multilinear continuous and 
$$\left(V_1,\dots,V_{k+1}\right)\in C^\infty(End(E))^{k+1}\mapsto FP|_{s=0}\frac{I(s;V_1,\dots,V_{k+1})}{\Gamma(s+1)} $$
\emph{coincides with the G\^ateaux differential} 
$$\frac{(-1)^k}{k!}D^{k+1}\log\text{det}_\zeta(\Delta+V,V_1,\dots,V_{k+1})$$ of the analytic function 
$\log\text{det}_\zeta$ on $C^\infty(End(E))$.
Observe that $M^{k+1}\setminus d_{k+1}$ is covered by open subsets of the
form
 $U_1\times\dots\times U_{k+1}$ s.t.
$\overline{U_1}\cap\dots\cap \overline{U_{k+1}}=\emptyset$.
By the multilinear Schwartz kernel (see definition~\ref{d:multilinearschwartz}), the above multilinear map is represented by a distribution 
$\mathcal{R}t_{k+1}\in \mathcal{D}^\prime(M^{k+1},End(E)^{\boxtimes k+1})$ which \textbf{coincides} with the product
$t_{k+1}=\mathbf{G}(x_{1},x_{2})\dots\mathbf{G}\left(x_{k+1},x_1\right)\in \mathcal{D}^\prime(M^{k+1}\setminus d_{k+1},End(E)^{\boxtimes k+1})$
since 
$$Tr\left(\Delta^{-1}V_1\dots\Delta^{-1}V_{k+1} \right)= \left\langle t_{k+1} , V_1\boxtimes \dots\boxtimes V_{k+1} \right\rangle$$ for
$\text{supp}(V_1)\cap \dots\cap \text{supp}(V_{k+1})=\emptyset$ and $\left\langle .,. \right\rangle $ is a distributional pairing on $M^{k+1}$.
In the fermionic case, the discussion is almost identical.

From the above observation, we deduce the following claim which holds true in \textbf{both
bosonic and fermionic settings} which summarizes all above results:
\begin{prop}\label{p:functionalder}
Following the notations from definitions (\ref{d:bosonsfermionszeta}) and (\ref{d:speczetadet}).
Let $M$ be a smooth, closed, compact Riemannian manifold of dimension $d$, 
and $\Delta$ (resp $\Delta=D^*D$) some generalized Laplacian acting on $E$ (resp $E_+$) s.t. $\sigma\left(\Delta \right)\subset \{ Re(z)\geqslant \delta>0 \}$ for bosons (resp fermions).

In the bosonic case,
for every invertible $\Delta+V\in \mathcal{A}=\Delta+C^\infty(End(E))$,
for every $(V_1,\dots,V_{n})\in C^\infty(End(E))^n$,
if $n>\frac{d}{2}$ or $\text{supp}\left(V_1\right)\cap \dots 
\cap\text{supp}\left(V_{n}\right)=\emptyset$ then:
\begin{equation}
D^n\log\text{det}_\zeta\left(\Delta+V\right)(V_1,\dots,V_n)
=(-1)^{n-1} (n-1)!Tr_{L^{2}}\left(\left(\Delta+V\right)^{-1}V_1\dots  \left(\Delta+V\right)^{-1}V_n\right).
\end{equation}
For 
general $(V_1,\dots,V_n)\in C^\infty(End(E))^n$:
\begin{equation}
\frac{(-1)^{n-1}}{n-1!}D^n\log\text{det}_\zeta\left(\Delta+V\right)(V_1,\dots,V_n)=\left\langle \mathcal{R}t_n,V_1\boxtimes \dots \boxtimes V_n\right\rangle
\end{equation}
where $\mathcal{R}t_n$ is a \textbf{distributional extension} of
$t_n=\mathbf{G}(x_{1},x_{2})\dots\mathbf{G}\left(x_{n},x_1\right) \in \mathcal{D}^\prime(M^{n}\setminus d_{n},End(E)^{\boxtimes n})$
where $\mathbf{G}\in \mathcal{D}^\prime(M\times M, E\boxtimes E^*)$ is the Schwartz kernel of $\left(\Delta+V\right)^{-1}$.

In the fermionic case, for every invertible $D+A:C^\infty(E_+)\mapsto C^\infty(E_-)$,
for every $(A_1,\dots,A_n)\in C^\infty(Hom(E_+,E_-))^n$,
if $n>d$ or $\text{supp}\left(A_1\right)\cap \dots 
\cap\text{supp}\left(A_{n}\right)=\emptyset$ then:
\begin{equation}
D^n\log\text{det}_\zeta\left(\Delta+D^*A\right)(A_1,\dots,A_n)=(-1)^{n-1} (n-1)!Tr_{L^{2}}\left((D+A)^{-1}A_1\dots  (D+A)^{-1}A_n\right).
\end{equation}
For 
general $(A_1,\dots,A_n)\in C^\infty(Hom(E_+,E_-))^n$
\begin{equation}
\frac{(-1)^{n-1}}{n-1!}D^n\log\text{det}_\zeta\left(\Delta+D^*A\right)(A_1,\dots,A_n)=
\left\langle \mathcal{R}t_n,A_1\boxtimes \dots \boxtimes A_n\right\rangle
\end{equation}
where $\mathcal{R}t_n$ is a \textbf{distributional extension} of
$t_n=\mathbf{G}(x_{1},x_{2})\dots\mathbf{G}\left(x_{n},x_1\right) \in \mathcal{D}^\prime(M^{n}\setminus d_{n},Hom(E_-,E_+)^{\boxtimes n})$
where $\mathbf{G}\in \mathcal{D}^\prime(M\times M,E_+\boxtimes E_-^*)$ is the Schwartz kernel of $(D+A)^{-1}$.
\end{prop}
\begin{proof}
In the bosonic case, we proved the claim for all $\Delta+V$  
s.t. $\sigma\left(\Delta+V\right)\subset \{ Re(z)\geqslant \delta>0 \}$
since we need the exponential decay of the heat semi--group $e^{-t(\Delta+V)}$ to make the 
regular part $R$ from definition~\ref{d:decomp} convergent. However, by \textbf{analyticity of both sides}
of the identity $D^n\log\text{det}_\zeta\left(\Delta+V\right)(V_1,\dots,V_n)=(-1)^{n-1} (n-1)!Tr_{L^{2}}\left(\left(\Delta+V\right)^{-1}V_1\dots  \left(\Delta+V\right)^{-1}V_n\right)$ in $V\in C^\infty(M,End(E))$, 
the claim holds true everywhere on $C^\infty(M,End(E))$ by 
analytic continuation using the fact that the subset of invertible elements 
in $\mathcal{A}$ is \textbf{connected}~\footnote{ $\Delta+V$ invertible iff $Id+\Delta^{-1}V$ invertible which is true for $V$ in a small neighborhood of $V=0$. Then consider complex rays $z\in \mathbb{C}\mapsto Id+z\Delta^{-1}V$ 
which are non-invertible at isolated values of $z$ since $\Delta^{-1}V$ compact}. 
The discussion is identical for the fermionic case. 
\end{proof}
%

From the results of the above proposition, we can conclude 
the proof of Theorem~\ref{t:mainzetareformulate} and prove all claims 
except the bound on the wave front set of the Schwartz kernels of the G\^ateaux differentials for which we shall devote the whole section
\ref{s:wfbound}.
\subsection{Factorization formula relating $\det_\zeta$ and Gohberg--Krein's determinants $\det_p$.}
\label{ss:factozeta}
We give the proof for 
bosons and
write 
the factorization formula for
fermions for simplicity since the 
discussion is almost similar in both cases.
The proof relies crucially on the following well-known Lemma
on the Gohberg--Krein determinants~\cite[Thm 9.2 p.~75]{Simon-traceideals}:
\begin{lemm}[Gohberg--Krein's determinants and functional traces]
\label{l:detregtraces}
For all $A\in \mathcal{I}_p$, Gohberg--Krein's determinant $\det_p(Id+zA)$ is an \textbf{entire
function} in $z\in \mathbb{C}$ and is related to traces $Tr(A^n)$ for $n\geqslant
 p$ by the following
formulas:
\begin{eqnarray*}
\boxed{\text{det}_p(Id+zA)= \exp\left(\sum_{n=p}^\infty  \frac{(-1)^{n+1}z^n}{n}Tr(A^n)  \right)=\prod_k\left[(1+z\lambda_k(A))\exp\left( \sum_{n=1}^{p-1}\frac{(-1)^n}{n}\lambda_k(A)^n  \right)\right]}
\end{eqnarray*}
where the infinite product vanishes exactly when 
$z\lambda_k(A)=-1$ with multiplicity.
\end{lemm}
The product
$\prod_k\left[(1+z\lambda_k(A))\exp\left( \sum_{n=1}^{p-1}\frac{(-1)^nz^n}{n}\lambda_k(A)^n  \right)\right]$
reads
$\prod_kE_{p-1}(-z\lambda_k(A)) $ where $E_{p-1}$ is the Weierstrass factor.
From the above, we deduce:
\begin{prop}\label{p:convdethadamard}
Let $(M,g)$ be a closed compact Riemannian 
manifold of dimension $d$, 
$(E,F)$ a pair of isomorphic Hermitian bundles
over $M$ and $P:C^\infty\left(E\right)\mapsto C^\infty\left(F\right)$ 
an invertible elliptic operator of degree $k$. For any $\mathcal{V}\in C^\infty(Hom(E,F))$,
the series 
$$\sum_{n>\frac{d}{k}}\frac{(-1)^{n+1}}{n} Tr\left(\left(P^{-1}\mathcal{V}\right)^n \right)$$
converges absolutely for $\Vert \mathcal{V}\Vert_{L^\infty(Hom(E,F))}$ small enough and 
$$\mathcal{V}\mapsto \exp\left(\sum_{n>\frac{d}{k}}\frac{(-1)^{n+1}}{n} Tr\left(\left(P^{-1}\mathcal{V}\right)^n \right) \right)=\text{det}_{[\frac{d}{k}]+1}\left(Id+P^{-1}\mathcal{V} \right)$$ extends uniquely as an entire
function on $C^\infty(Hom(E,F))$.
\end{prop}
\begin{proof}
Choose some auxiliary bundle isomorphism $E\mapsto F$
which induces an elliptic invertible operator $U\in \Psi^0(M,E,F):L^2(E) \mapsto L^2(F)$ and
$UP^{-1}\in \Psi^{-k}(M,E)$  belongs to
the Schatten ideal $\mathcal{I}_{[\frac{d}{k}]+1}$ hence 
$\Vert UP^{-1}\Vert_{[\frac{d}{k}]+1}<+\infty$. 
The claim then 
follows from Lemma~\ref{l:detregtraces} applied to $A=P^{-1}\mathcal{V}\in \Psi^{-k}(M,E)$ which belongs to
the Schatten ideal $\mathcal{I}_{[\frac{d}{k}]+1}$
and the series converges since the Schatten norm satisfies the estimate:
$$ \Vert P^{-1}\mathcal{V} \Vert_{[\frac{d}{k}]+1}\leqslant \Vert U^{-1} \Vert_{\mathcal{B}(L^2(E),L^2(F))} \Vert \mathcal{V}\Vert_{\mathcal{B}(L^2(E),L^2(F))} \Vert U P^{-1}\Vert_{[\frac{d}{k}]+1}  $$ which
can be made $<1$ if $\Vert \mathcal{V}\Vert_{L^\infty(Hom(E,F))} < \Vert UP^{-1}\Vert_{[\frac{d}{k}]+1} \Vert U^{-1} \Vert_{\mathcal{B}(L^2(E),L^2(F))}^{-1}$. 
\end{proof}

Proposition~\ref{p:convdethadamard} and Lemma~\ref{l:detregtraces} imply that for $z$ small enough, 
the series 
$$\sum_{n\geqslant \frac{d}{2}+1}\frac{(-1)^{n+1}z^n}{n}
Tr_{L^2}\left( \left(\Delta^{-1}V\right)^n \right)$$ 
converges and 
equals $\log\det_p\left(Id+z\Delta^{-1}V\right)$ for $p=[\frac{d}{2}]+1$. 
In particular, 
$\det_p\left(Id+\Delta^{-1}V\right)$ is analytic 
in $V\in C^\infty(End(E))$ and vanishes iff $\Delta+V$ is non-invertible
and $z\mapsto \det_p\left(Id+z\Delta^{-1}V\right)$ 
is an entire function.

It follows from Proposition \ref{p:functionalder} that 
G\^ateaux differentials of $\log\det_p(Id+\Delta^{-1}V)$ and $\log\det_\zeta(\Delta+V)$ at $V=0$
coincide when $k>\frac{d}{2}$ 
. 
Now we shall use the following general
Lemma whose proof is given in subsection~\ref{ss:appendixholo} in appendix:
\begin{lemm}\label{l:raytofrechet}
Let $E$ be a complex Fr\'echet space and $F_1,F_2$ 
a pair of holomorphic functions on some open subset $\Omega\subset E$.
Assume there is some integer $k>0$, s.t. $F_1,F_2$ have the same G\^ateaux differentials of order $n$ for every $n\geqslant k$. 
Then there exists a smooth polynomial function $P:E\mapsto \mathbb{C}$ of degree $k-1$ s.t.
$F_1=P+F_2$.
\end{lemm}
The Lemma
implies that we have the equality
\begin{equation}\label{e:zetavshadamard}
\log\text{det}_\zeta(\Delta+V)=P(V)+\log\text{det}_p(Id+\Delta^{-1}V)
\end{equation}
for $V$ close enough to $0$ where $P$ is a \textbf{continuous polynomial function} of $V$.
Equation~(\ref{e:zetavshadamard}) together with the fact that $H\mapsto\det_p(Id+H)$ is an entire function
on the Schatten ideal $\mathcal{I}_p$ vanishing exactly over noninvertible $Id+H$, proves that
$V\mapsto \det_\zeta(\Delta+V)=e^{P(V)}\text{det}_p(Id+\Delta^{-1}V) $ extends uniquely 
as an entire function 
on $\mathcal{A}$ vanishing exactly over non-invertible elements.
Then by 
Proposition~\ref{p:functionalder}, 
the Schwartz kernels $\frac{(-1)^{n-1}}{n-1!}[\mathbf{D^n\log\det_\zeta(\Delta)}]$ of the G\^ateaux differentials 
are \textbf{distributional extensions} of the distributions $$t_n=\mathbf{G}(x_1,x_{2})\dots \mathbf{G}(x_n,x_1)\in \mathcal{D}^\prime(M^n\setminus d_n,End(E)^{\boxtimes n}).$$
It follows that
$P(V)=\sum_{1\leqslant n\leqslant \frac{d}{2}} \frac{(-1)^{n+1}}{n} \left\langle \mathcal{R}t_n,V^{\boxtimes n} \right\rangle$ where 
$\mathcal{R}t_n\in \mathcal{D}^\prime(M^n,End(E)^{\boxtimes n})$ is a distributional
extension of the product $t_n\in \mathcal{D}^\prime(M^n\setminus d_n,End(E)^{\boxtimes n})$.

It remains to estimate the wave front set of $\mathcal{R}t_n$ over the deepest diagonal $d_n$ which is
the purpose of the 
next section where we will show that
$WF\left(\mathcal{R}\left(t_n\right)\right)$ satisfies the bound 
$WF\left( \mathcal{R}\left(t_n\right) \right) \cap T_{d_n}^*M^n\subset N^*(d_n\subset M^n)$ 
by Proposition~\ref{p:wfbound} in the next section.
This will conclude the proof that
$\det_\zeta$ admits the representation~\ref{e:repsformula}.
The proof for fermions is similar
and yields the factorization formula
$
\det_\zeta\left(\Delta+D^*A\right)=\exp\left(P(A) \right)\det_p\left(Id+\Delta^{-1}D^*A\right)
$
for $p=d+1$ and $P$ a continuous polynomial of degree $d$ on $C^\infty(Hom(E_+,E_-))$.

Our goal for the next part is to study the microlocal properties of the Schwartz kernel 
$\mathcal{R}t_n=[\mathbf{D^n\log\det_\zeta}]$ near the deepest diagonal $d_n$.
As explained in the introduction, the bounds on the wave fronts are needed to ensure the renormalized determinants are obtained by subtraction of smooth local counterterms.

\section{Wave front set of Schwartz kernels of $D^n\log\det_\zeta$. }
\label{s:wfbound}
\subsubsection{Traces and integrals on configuration space.}
\label{ss:integparagraph}
%
First, we need to reformulate all composition of operators appearing in $D^n\log\det_\zeta$ as integrals
of products of operator kernels on configuration space.
In the bosonic case, for $(u_1,\dots,u_{k+1})\in (0,1]^{k+1}$, we reformulate the trace term
$Tr\left(e^{-u_{1}\Delta}V_1\dots  e^{-u_{k+1}\Delta}V_{k+1} \right)$ as an 
integral over configuration space
\begin{eqnarray*}
\small{ \int_{M^{k+1}} \left\langle e^{-u_1\Delta}(x_1,x_2)\dots e^{-u_{k+1}\Delta}(x_{k+1},x_1),\chi(x_1,\dots,x_{k+1})\right\rangle dv_{k+1} }
\end{eqnarray*}
where $dv_{k+1}\in \vert \Lambda^{top}\vert M^{k+1}$,
the product $e^{-u_1\Delta}(x_1,x_2)\dots e^{-u_{k+1}\Delta}(x_{k+1},x_1)$ on the l.h.s
is an element in $C^\infty(M^{k+1},End(E^*)^{\boxtimes k+1})$, 
$\chi=V\boxtimes\dots\boxtimes V\in C^\infty(M^{k+1},End(E)^{\boxtimes k+1})$ 
is the test section and the $\left\langle.,.\right\rangle$
denotes the natural fiberwise
pairing between elements of $End(E)^{\boxtimes k+1}$ and $End(E^*)^{\boxtimes k+1}$.
Starting from now on in the bosonic case, the test function part $\chi$ will
be chosen arbitrarily in $C^\infty(M^{k+1},End(E)^{\boxtimes k+1})$.

 In the fermionic case, we will consider 
the operator $e^{-t\Delta}D^*:C^\infty(E_-)\mapsto C^\infty(E_+)$
which has smoothing kernel when $t>0$ (since $\Psi^{-\infty}$ is an ideal)
hence the trace term
$Tr\left(e^{-u_{1}\Delta}D^*A_1\dots  e^{-u_{k+1}\Delta}D^*A_{k+1} \right)$
is reformulated as the integral over configuration space:
\begin{eqnarray}
\small{ \int_{M^{k+1}} \left\langle e^{-u_1\Delta}D^*(x_1,x_2)\dots e^{-u_{k+1}\Delta}D^*(x_{k+1},x_1),\chi(x_1,\dots,x_{k+1})\right\rangle dv_{k+1}}
\end{eqnarray}
where $dv_{k+1}\in \vert \Lambda^{top}\vert M^{k+1}$, the product $e^{-u_1\Delta}D^*(x_1,x_2)\dots e^{-u_{k+1}\Delta}D^*(x_{k+1},x_1)$ 
on the l.h.s
is an element in $C^\infty(M^{k+1},Hom(E_-,E_+)^{\boxtimes k+1})$, $\chi=A\boxtimes\dots\boxtimes A\in C^\infty(M^{k+1},Hom(E_+,E_-)^{\boxtimes k+1})$
and $\left\langle.,.\right\rangle$
denotes the natural fiberwise
pairing between elements of $Hom(E_+,E_-)^{\boxtimes k+1}$ and $Hom(E_-,E_+)^{\boxtimes k+1}$.

In what follows, we will localize the study in
some open subset of the form $U^{k+1}$ near an element of the form
$(x,\dots,x)\in d_{k+1}\subset M^{k+1}$ where $U\subset M, x\in U$
is an open chart that we choose to identify with some bounded open subset $U$ of $\mathbb{R}^d$ making some
abuse of notations.
Recall that a consequence of the heat calculus is that in local coordinates
$e^{-t\Delta}(x,y)=t^{-\frac{d}{2}}\tilde{A}(t,\frac{x-y}{\sqrt{t}},y), \tilde{A}\in C^\infty([0,+\infty)_{\frac{1}{2}}\times
\mathbb{R}^d\times U, E \boxtimes E^*)$ for bosons
and 
$e^{-t\Delta}D^*(x,y)=t^{-\frac{d+1}{2}}\tilde{A}(t,\frac{x-y}{\sqrt{t}},y), \tilde{A}\in C^\infty([0,+\infty)_{\frac{1}{2}}\times\mathbb{R}^d\times U, E_+ \boxtimes E_-^*)$ for fermions.
From this observation on the asymptotics of the kernel
$e^{-t\Delta}D^*$ the proofs in both bosonic and fermionic cases are uniform. The only changes
occur in the numerology since there is a loss of $t^{-\frac{1}{2}}$ in powers of $t$ in the expansion of 
$e^{-t\Delta}D^*$.
%
%
%
%
%
%
%
\begin{defi}
We define for $ (u,x)=((u_e)_{e=1}^{k+1},(x_i)_{i=1}^{k+1})\in (0,1]^{k+1}\times U^{k+1} $:
{\small
\begin{eqnarray*}
J(u,x;\chi)=\left\langle\prod_{1\leqslant e\leqslant k+1}\tilde{A}(u_e,\frac{x_{i(e)}-x_{j(e)}}{\sqrt{u_e}},x_{j(e)}),\chi\right\rangle
\end{eqnarray*}}
where $i(e)=e,j(e)=e+1$ when $e\in \{1,\dots,k\}$ and $i(k+1)=k+1,j(k+1)=1$, 
the bracket $\left\langle.,.\right\rangle $ 
denotes the appropriate fiberwise pairing defined above, $J(.,.,\chi)\in C^\infty \left((0,1]^{k+1}\times U^{k+1} \right)$ and 
$J$ depends linearly on $\chi$.
\end{defi}
Then we can express 
$S$ from definition~\ref{d:decomp} in terms 
of $J$:
\begin{eqnarray*}
S(s;\chi)=  \int_{\mathbf{\Delta}_{k+1}} (\sum_{e=1}^{k+1}u_e)^s 
\left(\int_{M^{k+1}} J(u,x;\chi) \right) d^{k+1}u, \chi=V_1\boxtimes \dots\boxtimes V_{k+1}.
\end{eqnarray*}
We will next prove that $S$ is a distribution valued in meromorphic functions of the variable $s$~\cite[def 3.5 p.~11]{dangzhang}
based on the methods of~\cite{dangzhang} using blow--ups.
In fact, we will also recover the distributional order from our proof.  

\subsubsection{Resolving products of heat kernels.}
\label{ss:blowupparagraph}

The problem in the definition of $S$ is that the 
product $\prod_{1\leqslant e\leqslant k+1} \tilde{A}(u_e,\frac{x_{i(e)}-x_{j(e)}}{\sqrt{u_e}},x_{j(e)})$ 
is not smooth
on
$\mathbf{\Delta}_{k+1} \times U^{k+1}$.
We choose the test section $\chi$ supported in $U^{k+1}\subset M^{k+1}$. 
We integrate w.r.t the Riemannian volume $dv=\rho d^dx$ with smooth density $\rho$ w.r.t. the Lebesgue measure $d^dx$ hence without loss of generality, we choose to absorb $\rho$ in the test section $\chi$.
 
\begin{defi}[Blow--up]
\label{d:blowup}
Consider the following change of variables:
\begin{eqnarray*}
\beta: (x,h_1,\dots,h_{k},t_1,\dots,t_{k+1})&\longmapsto &((x_1=x,x_i=x+\sum_{j=1}^{i-1}(t_{j}\dots t_{k+1}) h_j)_{i=2}^{k+1},(u_l=\prod_{j=l}^{k+1} t_j^2)_{l=1}^{k+1})\\
U\times \mathbb{R}^{kd}\times [0,1]^{k+1}&\longmapsto &U^{k+1}\times \mathbf{\Delta}_{k+1}  
\end{eqnarray*}
~\cite[def 5.3]{dangzhang} which resolves
the singular product in $J$.
We use the short notation $(x,h,t)=(x,(h_1,\dots,h_{k}),(t_1,\dots,t_{k+1}))\in U\times \mathbb{R}^{kd}\times [0,1]^{k+1}$.
\end{defi}
Replacing in the integral expression of $S$ yields,
$$S(s;\chi)=
\int_{[0,1]^{k+1}\times U\times \mathbb{R}^{dk}} ((t_1\dots t_{k})^2+\dots+1)^s   \beta^*J(t,x,h;\chi)  t_1\dots t_{k+1}^{2k+2s+1} d^{k+1}td^dxd^{kd}h,
$$
where the factor $t_{k+1}^{2k+1}$ comes from the Jacobian determinant $\beta^*\left(d^{k+1}u\right)=2^{k+1}t_1\dots t_{k+1}^{2k+1} d^{k+1}t $
in the change of variables.
One of the key results from~\cite[Thm 5.2]{dangzhang} is that for every $e\in \{1,\dots,k\}$,
the pull--back 
$\beta^*\left(\frac{x_{i(e)}-x_{j(e)}}{\sqrt{u_e}}\right)$ by the blow--down map $\beta$
is a \textbf{smooth function} on the resolved space
$U\times \mathbb{R}^{dk} $ hence $\beta^*J$ is also smooth on the resolved space.
In the bosonic (resp fermionic) case, 
the change of variables
in $\beta^*\left(\prod_{e=1}^{k+1}u_e^{-\frac{d}{2}}\right)$ 
(resp $\beta^*\left( \prod_{e=1}^{k+1}u_e^{-\frac{d+1}{2}}\right)$) brings a factor
of the form $\prod_{1\leqslant l \leqslant k+1} (t_{l}\dots t_{k+1})^{-d}$ (resp $ \prod_{1\leqslant l \leqslant k+1} (t_{l}\dots t_{k+1})^{-d-1}$) and
the Jacobian determinant of the variable change 
yields a factor $(t_1\dots t_{k+1})^{d}\dots (t_kt_{k+1})^{d} $
from which we can extract the power of $t_{k+1}$ to be equal to
$t_{k+1}^{-d} $ (resp $t_{k+1}^{-d-k-1} $).
Replacing in the integral formula yields the following identity~\cite[Proposition 5.2]{dangzhang},
\begin{eqnarray}\label{e:integS}
S(s;\chi)=\int_{[0,1]^{k+1}}((t_1\dots t_{k})^2+\dots+1)^{s}P(t_1,\dots,t_k)\\
\nonumber\times \left( \int_{\mathbb{R}^{d(k+1)}} A(t,x,h;\chi)d^dxd^{kd}h \right) t_{k+1}^{2s+2k+1-d}  dt_1\dots dt_{k+1} 
\end{eqnarray}
where $P$ is a polynomial function whose explicit expression is irrelevant 
and 
$A(.;\chi)=\beta^*J(.;\chi)\in C^\infty([0,1]^{k+1}\times U\times \mathbb{R}^{dk})$. For fermions, we would get $t_{k+1}^{2s+k-d}$ in factor under the integral sign instead of $t_{k+1}^{2s+2k+1-d}$.
As in~\cite[Lemma 5.4]{dangzhang}, $A$ depends \textbf{linearly} on
$\beta^*\chi(t,x,h)=\chi(x,x+\sum_{j=1}^{i-1}(t_{j}\dots t_{k+1}) h_j)_{i=2}^{k+1} $. We find
\begin{lemm}
Under the previous notations, in both bosonic and fermionic case,
the quantity $\partial_{t_{k+1}}^p A(t,x,h)|_{t_{k+1}=0}$ 
depends linearly on $p$-jets of the coefficients of 
$\chi$ along the 
diagonal
$d_{k+1}\subset M^{k+1}$. 
\end{lemm}

In the following paragraph, 
we shall prove that 
both $ \chi\mapsto S(s;\chi)$ 
and $\chi\mapsto I(s;\chi)$
are distributions valued in 
meromorphic germs
at $s=0$.

\subsubsection{The case when $1\leqslant k\leqslant \frac{d}{2}$ (resp  $1\leqslant k\leqslant d$) for bosons (resp fermions) case and integration by parts.}
\label{ss:badcase}

In the bosonic case,
if $k\leqslant \frac{d}{2}$, the factor $t_{k+1}^{2(s+k)+1-d} $ appearing in factor
of 
$A$ is potentially divergent
since near $s=0$, 
$2(Re(s)+k)+1-d $ 
is no longer necessary $>-1$.
Then as usual in Riesz regularization, we need to Taylor expand $A$ 
w.r.t the variable $t_{k+1}$ up to order $p$
in such a way that
$(p+1)+2k+1-d>-1\implies p+1>d-2(k+1)$.
This yields
\begin{eqnarray*}
&&\int_{[0,1]^{k+1}}((t_1\dots t_{k})^2+\dots+1)^{s}P(t_1,\dots,t_k)\\
&\times & \left( \int_{\mathbb{R}^{d(k+1)}} A(t,x,h;\chi)d^dxd^{kd}h \right) t_{k+1}^{2(s+k)+1-d}  dt_1\dots dt_{k+1} \\
&=&\sum_{p=0}^{\sup(d-2(k+1),0)} \int_{[0,1]^{k}}((t_1\dots t_{k})^2+\dots+1)^{s}P(t_1,\dots,t_k)\\ &\times &\frac{\int_{\mathbb{R}^{d(k+1)}} \partial^p_{t_{k+1}}A(t,x,h;\chi)|_{t_{k+1}=0}d^dxd^{kd}h }{s+k+p+1-\frac{d}{2}}  dt_1\dots dt_{k}
+\text{holomorphic at }s=0 
\end{eqnarray*}
where the holomorphic part depends linearly on the $(d-3)$-jet of $\chi$.
This implies that in the general case 
$FP|_{s=0} \frac{1}{\Gamma(s+1)} I(s,\chi)$ depends linearly on
the $(d-3)$-jets of $\chi$ and when
$\text{supp}(\chi)$ does not meet the deepest diagonal $d_{k+1}\subset M^{k+1}$,
we already know that 
$$FP|_{s=0}\frac{1}{\Gamma(s+1)}
I(s,\chi)=\int_{M^{k+1}} \left\langle \prod_{e=1}^{k+1}\mathbf{G}(x_{i(e)},x_{j(e)}),\chi\right\rangle \prod_{i=1}^{k+1} dv(x_i) .$$ 
Altogether, this proves
that the Schwartz kernel
$\frac{(-1)^{k}}{k!}[\mathbf{D^{k+1}\log\det_\zeta(\Delta)}]$ is a \textbf{distributional
extension} of the product
$\mathbf{G}(x_1,x_2)\dots\mathbf{G}(x_{k+1},x_1)$. This distributional extension has 
\textbf{order at most} $(d-3)$.
The fermionic case is similar only the numerology differs,
we need to expand coefficients of $\chi$ at order $p$ in
$t_{k+1}$ so that $p+1>d-2$ 
therefore
$FP|_{s=0} \frac{1}{\Gamma(s+1)} I(s,\chi) $ depends linearly 
on $(d-1)$-jets of the coefficients of $\chi$.

\subsubsection{Bounds on the Fourier transform of the singular term $S$.}
\label{ss:wfbounds}
In this part, the 
bosonic and fermionic cases are similar and therefore
we restrict to the former case for simplicity. To bound the wave front set of the meromorphic family of distributions $S(s;.)$,
using the above notations, 
we should study $S(s;\chi)$ where the test function part $\chi$ is chosen of the
form $\chi=\psi(x_1) e^{ix_1\xi_1}\dots \psi(x_{k+1}) e^{ix_{k+1}\xi_{k+1}} $
where $\psi$ has small support in the 
coordinate 
chart $U$ of $M$ and  
for large momenta $(\xi_1,\dots,\xi_{k+1})$ in some closed conic set $V\subset \mathbb{R}^{d(k+1)*}$ that we will later determine.
As usual for wave front bounds,
we are just localizing with the function $\psi$ and taking 
Fourier transforms.

After the change of variables of definition~\ref{d:blowup}, 
the exponential factor becomes 
\begin{eqnarray*}
\exp\left(ix(\xi_1+\dots+\xi_{k+1}   ) + i\sum_{e=2}^{k+1} \sum_{j=1}^{e-1} (t_j\dots t_{k+1})h_j \xi_e \right)\\
=\exp\left(ix(\xi_1+\dots+\xi_{k+1}   ) + i\sum_{j=1}^{k} (t_j\dots t_{k+1})h_j(\sum_{j+1 \leqslant e \leqslant k+1}  \xi_e) \right).
\end{eqnarray*}
Therefore, the term $A(t,x,h;\chi)$ has the explicit from $$A(t,x,h;\chi)=A(t,x,h;\psi^{\boxtimes k+1}) e^{\left(ix\sum_{j=1}^{k+1}\xi_j + i\sum_{e=2}^{k+1} \sum_{j=1}^{e-1} (t_j\dots t_{k+1})h_j \xi_e \right)}$$
so the interesting term in factor of $S$ that we should integrate by parts
w.r.t. $t_{k+1}$
reads
\begin{eqnarray*}
&&\vert\partial^p_{t_{k+1}}\int_{\mathbb{R}^{d(k+1)}} e^{\left(ix\sum_{j=1}^{k+1}\xi_j+ i\sum_{e=2}^{k+1} \sum_{j=1}^{e-1} (t_j\dots t_{k+1})h_j \xi_e \right)}A(t,x,h;\psi^{\boxtimes k+1})d^dxd^{kd}h\vert\\
&\leqslant & C (1+K\sum_{e=1}^{k+1}\vert \xi_e\vert)^p \sup_{1\leqslant j\leqslant p}\sup_{t\in [0,1]^{k+1}}\vert \widehat{\partial_t^jA}(t,\sum_{e=1}^{k+1}\xi_e,(t_j\dots t_{k+1})\sum_{e=j+1}^{k+1}\xi_e;\psi^{\boxtimes k+1}) \vert\\
&\leqslant & C_N(1+\vert \sum_{e=1}^{k+1} \xi_e\vert)^{-N}
\end{eqnarray*}
uniformly in $t\in [0,1]^{k+1}$ for all $N$ by \textbf{smoothness, in the $x$ variable, of $A$ and its derivatives $\partial^j_tA$}.
Assume that $(\xi_1,\dots,\xi_{k+1})$ belongs to some
closed conic set 
$V\subset \mathbb{R}^{d(k+1)*}$ which does not meet
the hyperplane
$\{ \sum_{i=1}^{k+1} \xi_i=0 \}$. 
Then there exists a constant $C$
such that for every 
$(\xi_1,\dots,\xi_{k+1})\in V$ satisfying $\sum_{i=1}^{k+1}\vert\xi_i\vert^2\geqslant R $, we have
$ \vert \sum_{i=1}^{k+1} \xi_i \vert\geqslant \varepsilon \left(1+\sum_{i=1}^{k+1}\vert\xi_i\vert\right)$.
Therefore, we obtain the estimate:
\begin{eqnarray*}
&&\vert\partial^p_{t_{k+1}}\int_{\mathbb{R}^{d(k+1)}} e^{\left(ix\sum_{j=1}^{k+1}\xi_j+ i\sum_{e=2}^{k+1} \sum_{j=1}^{e-1} (t_j\dots t_{k+1})h_j \xi_e \right)}A(t,x,h;\psi^{\boxtimes k+1})d^dxd^{kd}h\vert\\
&\leqslant & C_N\varepsilon^{-N}(1+ \sum_{e=1}^{k+1}\vert \xi_e\vert)^{-N}
\end{eqnarray*}
uniformly in $t\in [0,1]^{k+1}$. 
Finally, this means that
for any element $(x,\dots,x;\xi_1,\dots,\xi_{k+1})$ in the wave front set $ 
WF\left(FP|_{s=0}S(s,.) \right)$, we must have 
$\xi_1+\dots +\xi_{k+1}=0$. 
This implies the wave front set bound
$$\left(WF\left(FP|_{s=0}S(s,.)\right) \cap T_{d_{k+1}}^\bullet M^{k+1}\right)\subset N^*\left(d_{k+1}\subset M^{k+1} \right) $$
over the diagonal $d_{k+1}$. In the next paragraph, 
we will use these bounds on the Fourier transform of $S$ to estimate the wave front set
of the Schwartz kernel of the G\^ateaux differentials over the diagonal.

\subsubsection{Wave front bounds.}

The next step is to use the above methods to bound
the wave front set of the Schwartz kernels of the G\^ateaux differentials. 
An important application of the blow--up techniques is 
to estimate the wave front set of the extensions
$\mathcal{R}t_n\in \mathcal{D}^\prime(M^n)$ 
over the deep diagonal $d_n\subset M^n$ which is proved to be contained in the conormal bundle $
N^*\left(d_n\subset M^n  \right)$.

\begin{prop}\label{p:wfbound}
Following the notations from definitions (\ref{d:bosonsfermionszeta}) and (\ref{d:speczetadet}).
Let $M$ be a smooth, closed, compact Riemannian manifold of dimension $d$, 
$E\mapsto M$ some 
Hermitian bundle over $M$.

For every invertible generalized Laplacian  
$\Delta+V\in \mathcal{A}=\Delta+C^\infty(End(E))$ acting on $E$
s.t. $\sigma\left(\Delta \right)\subset \{ Re(z)\geqslant \delta>0 \}$,
set $\mathbf{G}\in \mathcal{D}^\prime(M\times M,E\boxtimes E^*)$ to be the Schwartz kernel of
$\left(\Delta+V\right)^{-1}$. The Schwartz kernel 
of the G\^ateaux differential of $\log\det_\zeta$ defined as 
$\mathcal{R}t_n=\frac{(-1)^{n-1}}{(n-1)!}[\mathbf{ D^n\log\det_\zeta(\Delta+V)}]$ is a distributional 
extension of
the product $$t_n=  \prod_{e\in E(G)} \mathbf{G}(x_{i(e)},x_{j(e)})  \in \mathcal{D}^\prime(M^n\setminus d_n,End(E^*)^{\boxtimes n}) $$
and
satisfies the wave front set bound 
\begin{equation}
\boxed{\left(WF\left(\mathcal{R}t_n\right) \cap T_{d_n}^\bullet M^n\right)\subset N^*\left(d_n\subset M^n \right).}
\end{equation}
\end{prop}

\begin{proof}
We need to show that the distribution $\mathcal{R}t_n$ defined as
$$\left\langle \mathcal{R}t_n, V_1\boxtimes \dots\boxtimes V_n \right\rangle=\frac{(-1)^{n}}{(n-1)!}
\frac{d}{ds}|_{s=0}D^nTr\left(\left(\Delta+V\right)^{-s}\right)(V_1,\dots,V_n)   $$
satisfies the wave front bound
$
WF(\mathcal{R}t_n)\cap T_{d_n}^\bullet M^n\subset N^*\left( d_n\subset M^n \right)$.
In fact the proof reduces to the scalar case and we assume without loss
of generality that we work with a scalar generalized Laplacian $\Delta$ acting on functions and the potential
$V\in C^\infty(M,\mathbb{C})$.
We start from the expression
{\small
\begin{eqnarray*}
\frac{1}{\Gamma(s+1)} \int_{[0,\infty)^n} \prod_{e=1}^{n} du_e(u_1+\dots+u_n)^s\left( 
\int_{M^n}e^{-u_1\Delta}(x_1,x_2) \dots e^{-u_n\Delta}(x_n,x_1)\chi(x_1,\dots,x_n) dv_n\right) , 
\end{eqnarray*}
}
where $dv_n$ is the volume form on $M^n$.
We work on a local chart $U^{n}$
where we choose the test section $\chi$ to be equal 
to $\chi=\psi(x_1) e^{ix_1\xi_1}\dots \psi(x_n) e^{ix_n\xi_n} $ 
where $\psi\in C^\infty_c(U)$ is supported on some chart $U$.
There is a competition between:
\begin{enumerate} 
\item integration of heat kernels on $[1,+\infty)$ 
which yields smoothing operators in the sense the family
$\left(e^{-u\Delta}(x,y)\right)_{u\in [1,+\infty)}$
is bounded in $C^\infty(M\times M)$ since $e^{-u\Delta}=e^{-\frac{1}{4}\Delta}\underset{\text{bounded}}{\underbrace{e^{-(u-\frac{1}{2})\Delta}}}e^{-\frac{1}{4}\Delta}$
where the term in the middle is uniformly bounded 
in $\mathcal{B}(L^2,L^2)$ by spectral assumption and 
both factors $e^{-\frac{1}{4}\Delta}$ on the left and right
are smoothing operators
in $(x,y)$ variable,
\item integration on $[0,1]$ which yields singular distributions
whose wave front set is conormal in the sense
that the family $
\left(e^{-u\Delta}(x,y)\right)_{u\in (0,1]}$
is a bounded family of distributions 
in $\mathcal{D}^\prime_{N^*\left(d_2\subset M^2 \right)}(M\times M)$~\footnote{in the sense of the seminorms in~\cite[p.~204]{DangBrouderHelein}}
which is the space of distributions whose wave front set is 
contained in the conormal bundle
$N^*\left(d_2\subset M^2 \right)$.
\end{enumerate}
Introduce a first decomposition where we sum over permutations
$S_n$ of $\{1,\dots,n\}$ in the second sum:
\begin{eqnarray*}
&&\int_{[0,\infty)^n} (u_1+\dots+u_n)^s 
\int_{M^n}\left\langle e^{-u_1\Delta} \dots e^{-u_n\Delta},\chi\right\rangle  \prod_{e=1}^{n} du_e\\
&=&\sum_{k=0}^n \frac{1}{n!}\sum_{\sigma\in S_n} \int_{[0,1]^k\times [1,+\infty)^{n-k} } (u_1+\dots+u_n)^s 
\int_{M^n}\left\langle e^{-u_{\sigma(1)}\Delta} \dots e^{-u_{\sigma(n)}\Delta},\chi\right\rangle  \prod_{e=1}^{n} du_e
\end{eqnarray*}
Without loss of generality, we only treat the terms corresponding to the
identity permutation of $S_n$.
When $k<n$, 
we use the hypocontinuity of the product of distributions whose wave front set
is fixed~\cite[Thm 6.1 p.~219]{DangBrouderHelein}. From the fact that the family $e^{-u_i\Delta}(x_i,x_{i+1})$, viewed as distribution on $M^n$, 
is bounded in $\mathcal{D}^\prime_{N^*\left(d_{\{i,i+1\}} \subset M^n \right) }(M^n)
$ where $N^*\left(d_{\{i,i+1\}} \subset M^n \right)$ 
is the conormal of the diagonal $d_{\{i,i+1\}}=\{x_i=x_{i+1}\}\subset M^n $,
we note that the distributional product
$\left(e^{-u_1\Delta}(x_1,x_2)\dots e^{-u_k\Delta}(x_k,x_{k+1})\right)_{(u_1,\dots,u_k)\in [0,1]^k}$ 
is bounded in $\mathcal{D}^\prime_\Gamma(M^n)$ 
for
$\Gamma=\bigcup_{I} N^*\left(d_I\subset M^n \right) \subset T^\bullet M^n$, 
where the union runs over the sets $I=\{i,\dots,j\}$, where $\{i,\dots,j\}$ contains the arithmetic progression from $i$ to $j$,
for $1\leqslant i<j\leqslant k $.
Then it follows immediately that
for $k<n$,
$$WF\left( \int_{[0,1]^k\times [1,+\infty)^{n-k} } (u_1+\dots+u_n)^s e^{-u_{1}\Delta} \dots e^{-u_{n}\Delta}  \prod_{e=1}^{n} du_e\right)\cap T^*_{d_n}M^n \subset N^*\left( d_n\subset M^n\right).$$

 For the term
where $k=n$, the result follows simply from the bounds on the Fourier transform of the singular
term $\left(WF(FP|_{s=0}S(s,.))\cap T_{d_n}^*M^n \right)\subset N^*\left(d_n\subset M^n \right)$
from paragraph~\ref{ss:wfbounds}. Gathering both cases yields the claim from the Proposition.
\end{proof}

\section{Proof of Theorem~\ref{t:quillenconjanal}.}

Equation~\ref{e:repsformula} shows that zeta regularized determinants,
defined by purely spectral conditions,
admit a position space representation in terms of Feynman amplitudes and 
that zeta determinants 
are just a particular case of some infinite dimensional family of renormalized determinants 
obtained by subtraction of local counterterms.

As above, we give the proof for bosons since the fermion case is similar and presents no extra difficulties. 

\subsubsection{Any element of the form $\mathcal{R}\det=e^{\text{Polynomial}}\det_\zeta$ solves Problem~\ref{d:renormdet1}.}
\label{ss:renormgroupzeta}

Assume $ \mathcal{R}\det(\Delta+V)=e^{Q(V)}\det_\zeta(\Delta+V)$ for some $Q\in \mathcal{O}_{loc,[\frac{d}{2}]}\left(J^mEnd(E) \right)$.
The zeta determinants from definition \ref{d:speczetadet} 
are solutions of problem~\ref{d:renormdet1} by Theorem~\ref{t:quillenconjzeta} and
Proposition~\ref{p:functionalder} where we found the second 
G\^ateaux differentials of $\det_\zeta$ 
to be equal 
to
\begin{eqnarray*}
D^2\log\text{det}_\zeta(\Delta+V,V_1,V_2)=Tr_{L^2}((\Delta+V)^{-1}V_1(\Delta+V)^{-1}V_2) 
\end{eqnarray*}
when $(V_1,V_2)\in C^\infty(End(E))^2$
have disjoint supports and $\sigma(\Delta+V)\subset\{Re(z)\geqslant \delta>0\} $. Therefore,
since $Q\in \mathcal{O}_{loc,[\frac{d}{2}]}\left(J^mEnd(E) \right)$ is a local polynomial functional
of degree $[\frac{d}{2}]$, the map 
$V\mapsto \mathcal{R}\det(\Delta+V)=\exp\left(Q(V) \right) \det_\zeta(\Delta+V)$ 
satisfies
$D^2\log\mathcal{R}\det(\Delta+V,V_1,V_2)= D^2\log\det_\zeta(\Delta+V,V_1,V_2)  $
where $D^2Q(V,V_1,V_2)=0$ since $(V_1,V_2)$ have disjoint supports and $Q$ is local~\cite[Prop V.5 p.~16]{brouder2018properties}.
This means 
\begin{eqnarray*}
D^2\log\mathcal{R}\det(\Delta+V,V_1,V_2)&=& D^2\log\text{det}_\zeta(\Delta+V,V_1,V_2)\\
=Tr_{L^2}((\Delta+V)^{-1}V_1(\Delta+V)^{-1}V_2) &=& \int_{M\times M} t_{2}(x_1,x_2)V(x_1)V(x_2)dv(x_1)dv(x_2).  
\end{eqnarray*}
The wave front bound from Proposition~\ref{p:wfbound} 
shows that the Schwartz kernel 
$$[\mathbf{D^2\log\mathcal{R}\det(\Delta+V)}]=[\mathbf{D^2\log\text{det}_\zeta(\Delta+V)}]+[\mathbf{D^2Q(V)}]$$
is a distribution $\mathcal{R}t_{2}\in \mathcal{D}^\prime(M\times M)$ satisfying 
$WF\left(\mathcal{R}t_{2}\right)\cap T_{d_2}^\bullet M^2\subset N^*\left(d_2\subset M^2\right) $
where we used the fact that $WF\left([\mathbf{D^2Q(V)}] \right)\subset N^*\left(d_2\subset M^2\right)$ by~\cite[Lemma VI.9 p.19~]{brouder2018properties}.
For the moment, we found
$\mathcal{R}\det$   
solves the equations~(\ref{e:secderivative}) and~(\ref{e:wfsecderivative}) and equation~(\ref{e:constraint}) is easily satisfied
by the factorization formula $\mathcal{R}\det(\Delta+V)=\det_\zeta\left(\Delta+V \right)e^{Q(V)}=e^{(P+Q)(V)}\det_{[\frac{d}{2}]+1}\left(Id+\Delta^{-1}V \right)$, $\deg(P+Q)\leqslant [\frac{d}{2}] $ and the properties of $\det_p$.
The last step is to use the factorization 
formula $\det_\zeta\left(\Delta+V \right)=e^{Q(V)}\det_{[\frac{d}{2}]+1}\left(Id+\Delta^{-1}V \right) $ from 
the previous section and the bound
$$\vert\text{det}_{[\frac{d}{2}]+1}\left(Id+\Delta^{-1}V \right)\vert\leqslant e^{K_1\Vert \Delta^{-1}V \Vert_{[\frac{d}{2}]+1}^{[\frac{d}{2}]+1}} \leqslant 
e^{K_1\left(\Vert \Delta^{-1}\Vert_{[\frac{d}{2}]+1}  \Vert V\Vert_{C^0}\right)^{[\frac{d}{2}]+1}} $$
which results from~\cite[b) Thm 9.2 p.~75]{Simon-traceideals} for the norm $\Vert .\Vert_{[\frac{d}{2}]+1}$ in the 
Schatten ideal $\mathcal{I}_{[\frac{d}{2}]+1}$, the fact that $\Delta^{-1}\in \Psi^{-2}(M,E)$ belongs to 
$\mathcal{I}_{[\frac{d}{2}]+1}$ since $\Delta^{-[\frac{d}{2}]-1}\in \mathcal{I}_1$~\cite[Prop B.21]{DyZwscatt} and H\"older's inequality
$\Vert \Delta^{-1}V \Vert_{[\frac{d}{2}]+1}\leqslant \Vert \Delta^{-1}\Vert_{[\frac{d}{2}]+1}  \Vert V\Vert_{C^0}$.
From the above facts, we deduce the 
bound:
$$ \vert \mathcal{R}\det\left(\Delta+V \right) \vert\leqslant \vert \text{det}_\zeta\left(\Delta+V \right)  \vert\vert e^{P(V)}\vert\leqslant 
 Ce^{K\Vert V\Vert_{C^m}^{[\frac{d}{2}]+1} } $$
for some $C,K>0$ independent of $V$ which proves the bound~(\ref{e:boundrdet}). 
Finally, 
$\mathcal{R}\det$ solves 
problem~\ref{d:renormdet1}.    
\subsubsection{Any renormalized determinant is of the form $e^{\text{Polynomial}}\det_\zeta$.}
%
Let $\mathcal{R}\det$ be any other solution of problem~\ref{d:renormdet1}, then for every $V$,
the entire functions 
$z\mapsto \mathcal{R}\det\left( \Delta+zV \right)$ and $z\mapsto \det_\zeta\left(\Delta+zV \right)$ have the same  
divisor (which means same zeros with multiplicities). It follows that the ratio
$ z\mapsto \frac{\mathcal{R}\det\left( \Delta+zV \right)}{\det_\zeta\left(\Delta+zV \right)} $ is an entire function
\textbf{without zeros} on $\mathbb{C}$ which satisfies the bound
\begin{eqnarray*}
\vert \frac{\mathcal{R}\det\left( \Delta+zV \right)}{\det_\zeta\left(\Delta+zV \right)} \vert\leqslant Ce^{K\vert z\vert^{[\frac{d}{2}]+1}\Vert V\Vert_{C^m}^{[\frac{d}{2}]+1} }, m=d-3.
\end{eqnarray*}
By the uniqueness part of Hadamard's Theorem~\ref{t:hadamard}, this implies that for every fixed $V$, 
$ z\mapsto \frac{\mathcal{R}\det\left( \Delta+zV \right)}{\det_\zeta\left(\Delta+zV \right)} =e^{P(z;V)} $ 
where $P$ is a polynomial of degree $[\frac{d}{2}]+1$ in $z$. 
We already know the map $V\mapsto \log\mathcal{R}\det\left( \Delta+V \right)-\log\det_\zeta\left(\Delta+V \right)$ is analytic near $V=0$
hence locally bounded near $V=0$ and also the above shows that
for every fixed 
$V$, $z\mapsto \log\mathcal{R}\det\left( \Delta+zV \right)-\log\det_\zeta\left(\Delta+zV \right)$
is a polynomial of degree $[\frac{d}{2}]+1$ in $z$. 
By proposition~\ref{p:douadyweakstrongholo}, this implies that 
the difference $\log\mathcal{R}\det\left( \Delta+V \right)-\log\det_\zeta\left(\Delta+V \right)=P(V)$ where $P$ is actually a
\textbf{continuous polynomial function} 
in $V$ of degree $[\frac{d}{2}]+1$. 
But condition~\ref{e:constraint} imposes the derivatives $(\frac{d}{dz})^{[\frac{d}{2}]+1} \log\mathcal{R}\det(\Delta+zV)$ 
and $(\frac{d}{dz})^{[\frac{d}{2}]+1}\log\det_\zeta\left(\Delta+zV \right)$
to coincide at $z=0$ hence $P$ has degree $[\frac{d}{2}]$.
It remains to show that
$P$ is local. 
The fact that both $\log\mathcal{R}\det\left( \Delta+V \right)$ and $\log\det_\zeta\left(\Delta+V \right)$ are solutions
of functional equation~\ref{e:secderivative} implies that
$D^2P(V,V_1,V_2)=0 $ if $\text{supp}(V_1)\cap \text{supp}(V_2)=\emptyset$.
Observe that
$$V\mapsto [\mathbf{ D^2\log\left( \frac{\mathcal{R}\det\left(\Delta+V \right)}{\det_\zeta\left(\Delta+V \right)  }\right)}]=[\mathbf{D^2P(V)}]
\in \mathcal{D}^\prime(M\times M)$$
is polynomial in $V$ valued in distributions on
$M\times M$ with wave front set in $N^*\left(d_2\subset M^2 \right)$.
To extract the homogeneous part , we use 
the finite difference operator $\Delta_V$ 
defined in the proof of~\ref{p:douadyweakstrongholo}, 
the element
$\frac{\Delta_{V}^{n-2}[\mathbf{D^2P(V)}]}{(n-2)!}=[\mathbf{P_n(V,\dots,V,.,.)}]\in \mathcal{D}^\prime(M\times M)$
has also wave front set in $N^*\left(d_2\subset M^2 \right)$ thus
$V\mapsto  [\mathbf{P_n(V,\dots,V,.,.)}]$
satisfies the assumptions of lemma~\ref{l:localityappendix} 
proved in appendix.
Therefore 
$V\mapsto \left\langle \tilde{P}_n , V^{\boxtimes n}\right\rangle$ is a local functional which 
equals $ \int_M \Lambda_n(V(x))dv(x) $
where $\Lambda_n(V(x)) $ depends on the $m$--jets of $V$ at $x$ for $m=d-3$
and is homogeneous of degree $n$ in $V$. 
It is important to stress that the function $\Lambda_n$ is not uniquely defined~\footnote{Only up to boundary terms} 
but the \textbf{functional} $V\mapsto \int_M \Lambda_n(V(x))dv(x) $ is uniquely defined. 
Then locality of $P$ together with the representation formula for $\det_\zeta$ 
from Theorem~\ref{t:quillenconjzeta}
implies that any solution of problem~\ref{d:renormdet1} has the form given by 
equation~\ref{e:repsformula}.
The infinite product representation is an easy consequence of
the representation of Gohberg--Krein's's determinants $\det_p$ as infinite products.  

To complete the proof of our Theorem, it 
remains to show that any $\mathcal{R}\det$ solution of Problem~\ref{d:renormdet1}
is obtained by a renormalization with subtraction 
of local counterterms in the sense of the third property in~\ref{t:quillenconjanal} which is the goal of the next section.

\section{Local renormalization and Theorem~\ref{t:mainthm6GFFreps} on Gaussian Free Field representation.}

We follow the notations from subsection~\ref{s:subtractloc}
where we explained the notion of subtraction of local counterterms.
The aim of this section is to show the third claim of Theorem \ref{t:quillenconjanal}. Namely 
that
all $\mathcal{R}\det$ solutions from
problem~\ref{d:renormdet1} are obtained from renormalization by subtraction
of \textbf{local counterterms} which concludes the proof of Theorem~\ref{t:quillenconjanal}: 
there exists a generalized Laplacian $\Delta$ with heat operator
$e^{-t\Delta}$ and a family $Q_\varepsilon\in \mathcal{O}_{loc,[\frac{d}{p}]}\left(J^mHom(E_+,E_-) \right)\otimes_\mathbb{C} \mathbb{C}[\varepsilon^{-\frac{1}{2}},\log(\varepsilon)]$ such that:
\begin{eqnarray}
\mathcal{V}\mapsto \mathcal{R}\det\left(P+\mathcal{V} \right)=\lim_{\varepsilon\rightarrow 0^+}\exp\left(Q_\varepsilon(\mathcal{V})\right) \text{det}_F\left(Id+e^{-2\varepsilon\Delta}P^{-1} \mathcal{V} \right).
\end{eqnarray}

\subsection{Extracting singular parts.}

In this subsection, we shall use the methods of~\cite{dangzhang}
based on blow--ups
to extract the singular parts of regularized traces $Tr_{L^2}\left(\left( e^{-2\varepsilon\Delta}\Delta^{-1}
V\right)^n \right)$ to show:
\begin{lemm}\label{l:singulartermsfeynman}
In the bosonic case, for 
every
$V\in C^\infty(M,End(E))$, we have an asymptotic expansion
\begin{eqnarray*}
Tr_{L^2}\left(\left( e^{-2\varepsilon\Delta}\Delta^{-1}
V\right)^{k+1} \right)  =
P_\varepsilon(V)+\mathcal{O}(1) 
\end{eqnarray*}
where $P_\varepsilon(V)=\int_M\Lambda_\varepsilon(V)dv\in \mathcal{O}_{loc,[\frac{d}{2}]}\left(J^mEnd(E) \right)\otimes_{\mathbb{C}}\mathbb{C}[\varepsilon^{-\frac{1}{2}},\log(\varepsilon)]$ and $m=d-3$
and 
\begin{eqnarray*}
(V_1,\dots,V_{k+1})\in C^\infty(End(E))^{k+1} &\mapsto &FP|_{\varepsilon=0}
Tr_{L^2}\left( e^{-2\varepsilon\Delta}\Delta^{-1}
V_1\dots e^{-2\varepsilon\Delta}\Delta^{-1}
V_{k+1} \right)\\
&=&\lim_{\varepsilon\rightarrow 0^+} Tr_{L^2}\left(\left( e^{-2\varepsilon\Delta}\Delta^{-1}
V\right)^{k+1} \right)  -
P_\varepsilon(V)
\end{eqnarray*}
is well--defined and multilinear 
continuous.

For fermions, for every $A\in C^\infty(Hom(E_+,E_-))$
, we have an asymptotic expansion
\begin{eqnarray*}
Tr_{L^2}\left(\left( e^{-2\varepsilon\Delta}\Delta^{-1}
D^*A\right)^{k+1} \right)  =
P_\varepsilon(A)+\mathcal{O}(1) 
\end{eqnarray*}
where $P_\varepsilon(A)=\int_M\Lambda_\varepsilon(A)dv\in \mathcal{O}_{loc,d}\left(J^mEnd(E) \right)\otimes_{\mathbb{C}}\mathbb{C}[\varepsilon^{-\frac{1}{2}},\log(\varepsilon)]$ and $m=d-1$
and 
\begin{eqnarray*}
(A_1,\dots,A_{k+1})\in C^\infty(Hom(E_+,E_-))^{k+1} &\mapsto &FP|_{\varepsilon=0}
Tr_{L^2}\left( e^{-2\varepsilon\Delta}\Delta^{-1}
D^*A_1\dots e^{-2\varepsilon\Delta}\Delta^{-1}
D^*A_{k+1} \right)\\
&=&\lim_{\varepsilon\rightarrow 0^+} Tr_{L^2}\left(\left( e^{-2\varepsilon\Delta}\Delta^{-1}
D^*A\right)^{k+1} \right)  -
P_\varepsilon(A)
\end{eqnarray*}
is well--defined and multilinear 
continuous.
\end{lemm}
Note that
in the previous Lemma, the functionals $P_\varepsilon$ depend only 
on the $m$-jets of their argument.
\begin{proof}
We prove the claim only for 
bosons, the fermionic case
is similar. 
In this lemma, we shall use the following notation, for two functions $a(\varepsilon),b(\varepsilon)$, we shall note
$a\simeq b$ if $b-a=\mathcal{O}(1)$ when $\varepsilon\rightarrow 0^+$. This means that $a,b$ have the same singular parts as $\varepsilon
$ approaches $0$.
We start from the identity:
\begin{eqnarray*}
Tr_{L^2}\left( e^{-2\varepsilon\Delta}\Delta^{-1}
V_1\dots e^{-2\varepsilon\Delta}\Delta^{-1}
V_{k+1} \right)&\simeq &\int_{[\varepsilon,1]^{k+1}} Tr_{L^2}\left( e^{-u_1\Delta}V_1
\dots e^{-u_{k+1}\Delta}V_{k+1}
 \right)du_1\dots du_{k+1} 
\end{eqnarray*} 
as a direct consequence 
of $e^{-2\varepsilon\Delta}\Delta^{-1}=\int_{2\varepsilon}^\infty e^{-t\Delta}dt$
and 
since the operator valued integral $\int_1^\infty e^{-t\Delta} dt \in \Psi^{-\infty} $ is smoothing.
Now without loss of generality and using the symmetry of the integral, we may assume that we 
work in the Hepp sector 
$\{\varepsilon\leqslant u_1<\dots <u_{k+1}\leqslant 1\}$ which is a semialgebraic subset
of the unit simplex $\mathbf{\Delta}_{k+1}=\{0\leqslant u_1\leqslant \dots\leqslant u_{k+1}\leqslant 1\}$.
So we need to study the asymptotics when $\varepsilon\rightarrow 0^+$ of 
$$  (k+1)! \int_{ \{\varepsilon\leqslant u_1<\dots <u_{k+1}\leqslant 1\} }
Tr_{L^2}\left( e^{-u_1\Delta}V_1\dots e^{-u_{k+1}\Delta}V_{k+1}
 \right)du_1\dots du_{k+1} .$$
Setting $\chi=V_1\boxtimes \dots \boxtimes V_{k+1}\in C^\infty(M^{k+1},End(E)^{k+1})$ 
and using the notations and conventions 
from paragraphs~\ref{ss:integparagraph} 
and~\ref{ss:blowupparagraph},
the blow--up from definition \ref{d:blowup}
yields a blow--down map
\begin{eqnarray*}
\beta: (x,h_1,\dots,h_{k},t_1,\dots,t_{k+1})&\longmapsto &((x_1=x,x_i=x+\sum_{j=1}^{i-1}(t_{j}\dots t_{k+1}) h_j)_{i=2}^{k+1},(u_l=\prod_{j=l}^{k+1} t_j^2)_{l=1}^{k+1})\\
U\times \mathbb{R}^{kd}\times \Omega_\varepsilon  &\longmapsto &U^{k+1}\times \{\varepsilon\leqslant u_1<\dots <u_{k+1}\leqslant 1\} 
\end{eqnarray*}
where $\Omega_\varepsilon$ is
the \textbf{semialgebraic set} defined by
$\Omega_\varepsilon=\{ \varepsilon\leqslant (t_1\dots t_{k+1})^2\}\cap [0,1]^{k+1}$.
Now, following the calculations 
of paragraph~\ref{ss:badcase} we set:
\begin{eqnarray*}
\omega(\chi)=\left(\int_{\mathbb{R}^{d(k+1)}}t_{k+1}^{2k+1-d}P(t_1,\dots,t_k)\beta^*J(t,x,h;\chi)d^dxd^{kd}h\right)dt_1\wedge\dots\wedge dt_{k+1} 
\end{eqnarray*}
where $t_{k+1}^{d-2k-1}\omega$ is a smooth differential form of top degree on the cube $[0,1]^{k+1}$.
To extract the 
singular part of $\int_{\Omega_\varepsilon}\omega$, we need to Taylor expand  
$\int_{\mathbb{R}^{d(k+1)}} A(t,x,h;\chi) d^dxd^{kd}h $ in the variable $t_{k+1}$:
\begin{eqnarray*}
\int_{\Omega_\varepsilon}\omega\simeq \sum_{j=0}^{}\int_{\Omega_\varepsilon} \frac{t_{k+1}^{2k+1-d+j}}{j!}P(t_1,\dots,t_k) 
\underbrace{\left(\partial_{t_{k+1}}^j\int_{\mathbb{R}^{d(k+1)}} A(t,x,h;\chi) d^dxd^{kd}h\right)|_{t_{k+1}=0}}
dt_1\wedge\dots\wedge dt_{k+1}
\end{eqnarray*}
where the term underbraced is a conormal distribution of $\chi\in C_c^\infty(U^{k+1},End(E)^{\boxtimes k+1})$
supported by $d_{k+1}$ by the results of paragraph~\ref{ss:wfbounds}.
So setting $\chi=V^{\boxtimes k+1}$, we can view the term 
underbraced as a functional of $V$: the map
$$V\in C^\infty_c(U,End(E))\mapsto \left(\partial_{t_{k+1}}^j\int_{\mathbb{R}^{d(k+1)}} A(t,x,h;V^{\boxtimes k+1}) d^dxd^{kd}h\right)|_{t_{k+1}=0}$$ is an element of $\mathcal{O}_{loc,k+1}\left(J^jE\right)$. 

Thus to extract precise asymptotics, we set
\begin{eqnarray*}
\omega_j=\frac{t_{k+1}^{2k+1-d+j}}{j!} \left(\partial_{t_{k+1}}^jP\int_{\mathbb{R}^{d(k+1)}} A(t,x,h;\chi) d^dxd^{kd}h\right)|_{t_{k+1}=0}
dt_1\wedge\dots\wedge dt_{k+1}.
\end{eqnarray*}
 Then
we may slice the \textbf{semialgebraic set} 
$\Omega_\varepsilon $
by the fibers $(t_1\dots t_{k+1})^2=\text{constant}$ of the map
$F(t_1,\dots,t_{k+1})=(t_1\dots t_{k+1})^2$. Practically, 
this means we will pushforward the differential form $\omega_j$ along the fibers
of the $b$-map $F:(t_1,\dots,t_{k+1}) \in [0,1]^{k+1}\mapsto F(t_1,\dots,t_{k+1})\in \mathbb{R} $~\cite[def 2.11 p.~16]{grieserbcalc}~\cite[p.~51--52]{melrose1992calculus} where the cube $[0,1]^{k+1}$
is viewed as a $b$-manifold in the sense of
Melrose~\cite[def 2.2 p.~8]{grieserbcalc}~\cite[p.~51]{melrose1992calculus}.
Then we will conclude by using
the pushforward Theorem of Melrose~\cite[Thm 4 p.~58]{melrose1992calculus} in the form discussed in the nice survey of Grieser~\cite[Thm 3.6 p.~25]{grieserbcalc}.
We define
the form $\frac{\omega_j}{dF}$ which is called Gelfand--Leray form~\cite[Lemma 5.11 p.~123]{zoladek2006} and
the function $J_j(t)=\int_{F=t}\frac{\omega_j}{dF}$.
By Fubini's Theorem, we find that
$ \int_{\Omega_\varepsilon} \omega_j= \int_{\varepsilon}^1 J_j(t)dt $. Finally, the pushforward Lemma~\ref{t:jeanquartier}, which is stated in the next subsubsection below, implies that for each $j\in \mathbb{N}$, the map
$\varepsilon\mapsto\int_{\Omega_\varepsilon} \omega_j$ admits an asymptotic expansion as $\varepsilon\rightarrow 0^+$
of the required form which concludes the proof.
\end{proof}

\subsubsection{Pushforward Lemma.}

Here we state the Lemma on asymptotic 
integrals used in
the previous proposition.
\begin{lemm}[Pushforward by Jeanquartier, Melrose]\label{t:jeanquartier}
Let $\omega\in \Omega^{n}_c([0,1]^{n})$ be a smooth differential form of top degree on $[0,1]^n$ 
and $F:(t_1,\dots,t_n)\in\mathbb{R}_+^n\mapsto t_1^2\dots t_n^2\in \mathbb{R}$.
Then for every $m\in \mathbb{N}$, the map 
$$t\mapsto J(t)=\int_{F^{-1}(t)} \frac{t_n^{-m}\omega}{dF}=\left\langle \delta(t-F),t_n^{-m}\omega \right\rangle $$
has an asymptotic expansion:
$J(t)\sim \sum_{p,q} t^p\log(t)^q a_{p,q}(\omega)$
where $p\in \frac{\mathbb{Z}}{2}$ runs over a finite set of growing arithmetic sequences of rational numbers  
and $a_{p,q}$ is a \textbf{distribution} supported by the
algebraic set $\{F=0\}$.
 
 This implies  
that the map
$\varepsilon\longmapsto \int_\varepsilon^1 J(t)dt$ also has an asymptotic expansion:
$ \int_\varepsilon^1 J(t)dt \sim \sum_{p,q} \varepsilon^p\log(\varepsilon)^qb_{p,q}(\omega) $
where $p\in \frac{\mathbb{Z}}{2}$ runs over a finite set of growing arithmetic sequences of rational numbers
and $b_{p,q}$ are distributions supported by $F=\{0\}$.
\end{lemm}
\begin{proof}
The result for smooth forms and real analytic $F$ is due to Jeanquartier~\cite{jeanquartier}~\cite[Theorem 5.54 p.~155]{zoladek2006}. 
Here we need the same result for a polyhomogeneous top form
$t_n^{-m}\omega$ and $F=t_1^2\dots t_n^2$ which is a particular case of the pushforward Theorem of Melrose~\cite[Thm 3.6 p.~25]{grieserbcalc}~\cite{melrose1992calculus} by the $b$-map $F$
which yields an index set contained in $\frac{\mathbb{Z}}{2}$ since the $b$-map $F$ vanishes at order $2$
on each boundary face of $[0,1]^n$.
%
%
%
%
Let us give a proof based on remarks from Jeanquartier on the Mellin transform~\cite{jeanquartiermellin}.
The index set of the asymptotics of $ t\mapsto \left\langle\delta(t-F),t_n^{-m}\varphi\right\rangle$ is exactly 
given by
the poles 
with multiplicity of the Mellin transform 
$\int_0^\infty t^{s} J(t) \frac{dt}{t}=\int_{[0,1]^n} F^{s-1}t_n^{-m}\varphi d^nt $ by~\cite[Prop 4.3 p.~304 and Prop 4.4 p.~306]{jeanquartiermellin}.
By successive Taylor expansion with remainder as follows,
start from $\varphi$ then Taylor expand with remainder at order $N$ in $t_1$ keeping other variables
$(t_2,\dots,t_n)$ as parameters, then Taylor expanding successively in $t_2,\dots,t_n$ with remainder
at order $N$ yields:
$\varphi(t_1,\dots,t_n)=\sum_{0\leqslant \alpha_1,\dots,\alpha_n\leqslant N} \prod t_i^{\alpha_i} c_\alpha $
where $c_\alpha$ depends on $t_i$ iff $\alpha_i=N$. Then plugging under the integral yields that
$s\mapsto \int_{[0,1]^n} F^{s-1}t_n^{-m}\varphi d^nt$ has analytic continuation
as a meromorphic function on $\mathbb{C}$ with singular terms
of the form $\left(\prod_{i=1}^{n-1} \frac{1}{2s+\alpha_i-1}\right)\frac{1}{2s+\alpha_n-1-m}$ hence
poles are in $\{s\in \frac{1+m-\mathbb{N}}{2}\}$ with multiplicity at most $n$.
\end{proof}

\subsection{Every $\mathcal{R}\det$ solution of problem \ref{d:renormdet1} are obtained by local renormalization.}

By Lemma \ref{l:singulartermsfeynman}, $\forall k\in \mathbb{N}$,
$(V_1,\dots,V_{k+1})\mapsto FP|_{\varepsilon=0}Tr_{L^2}\left( e^{-2\varepsilon\Delta}\Delta^{-1}
V_1\dots e^{-2\varepsilon\Delta}\Delta^{-1}
V_{k+1} \right)$ is multilinear continuous hence it 
can be represented as a distributional pairing
$$FP|_{\varepsilon=0}Tr_{L^2}\left( e^{-2\varepsilon\Delta}\Delta^{-1}
V_1\dots e^{-2\varepsilon\Delta}\Delta^{-1}
V_{k+1} \right)=\left\langle \mathcal{R}t_{k+1},V_1\boxtimes \dots\boxtimes V_{k+1} \right\rangle $$
by the multilinear Schwartz 
kernel Theorem. Exactly as in the proof
of subsubsection~\ref{ss:disjsupport}, 
we find that
for $(V_1,\dots,V_{k+1})\in C^\infty(M,End(E))^{k+1}$
such that
$\text{supp}(V_1)\cap\dots\cap\text{supp}(V_{k+1})=\emptyset$,
$$FP|_{\varepsilon=0}Tr_{L^2}\left( e^{-2\varepsilon\Delta}\Delta^{-1}
V_1\dots e^{-2\varepsilon\Delta}\Delta^{-1}
V_{k+1} \right)=Tr_{L^2}\left( \Delta^{-1}
V_1\dots \Delta^{-1}
V_{k+1} \right)$$
where the $L^2$ trace on the r.h.s is well--defined
since
$WF(\Delta^{-1}
V_1)\cap\dots\cap WF(\Delta^{-1}
V_{k+1})=\emptyset$.
Therefore arguing as in
subsubsection~\ref{ss:disjsupport}
we find that for $n\leqslant \frac{d}{2}$,
$\mathcal{R}t_n$ is a distributional extension of 
$t_n=\mathbf{G}(x_1,x_2)\dots \mathbf{G}(x_n,x_1)$
and for $n>\frac{d}{2}$, the composition
$e^{-2\varepsilon\Delta}\Delta^{-1}
V_1\dots e^{-2\varepsilon\Delta}\Delta^{-1}
V_{k+1}\in \Psi^{-2k}(M,E)$ hence of trace class~\cite[Prop B 21]{DyZwscatt} 
uniformly in $\varepsilon\in (0,1]$ hence
$$FP|_{\varepsilon=0}Tr_{L^2}\left( e^{-2\varepsilon\Delta}\Delta^{-1}
V_1\dots e^{-2\varepsilon\Delta}\Delta^{-1}
V_{k+1} \right)=Tr_{L^2}\left(\Delta^{-1}
V_1\dots \Delta^{-1}
V_{k+1} \right) $$
where the r.h.s. is well--defined
as in the case with zeta regularization.

Now let $P_{n,\varepsilon}\in \mathcal{O}_{loc}\otimes \mathbb{C}[\varepsilon^{-\frac{1}{2}},\log(\varepsilon)]$ from
Lemma~\ref{l:singulartermsfeynman} s.t.
$$\lim_{\varepsilon\rightarrow 0} Tr_{L^2}\left( (e^{-2\varepsilon\Delta}\Delta^{-1}V)^n \right)-P_{n,\varepsilon}(V)=FP|_{\varepsilon=0}Tr_{L^2}\left( (e^{-2\varepsilon\Delta}\Delta^{-1}
V)^n\right).$$
One should think of $P_{n,\varepsilon}$ as being the \textbf{singular part}
of $Tr_{L^2}\left( e^{-2\varepsilon\Delta}\Delta^{-1}
V_1\dots e^{-2\varepsilon\Delta}\Delta^{-1}
V_{n} \right)$.
Then set $P_\varepsilon(V)=\sum_{n=1}^{\frac{d}{2}} P_{n,\varepsilon}(V)$, we have
\begin{eqnarray*}
\text{det}_F\left(Id+e^{-2\varepsilon\Delta}\Delta^{-1}V\right)e^{-P_\varepsilon(V)}&=&
\underbrace{\exp\left(\sum_{n=1}^{\frac{d}{2}} Tr_{L^2}\left( (e^{-2\varepsilon\Delta}\Delta^{-1}V)^n \right)-P_{n,\varepsilon}(V)  \right)}
\\
&\times &
\underbrace{\text{det}_{[\frac{d}{2}]+1 }\left(Id+e^{-2\varepsilon\Delta}\Delta^{-1}V \right)}
\end{eqnarray*}
by the factorization properties of Gohberg--Krein's determinant~\cite[d) Thm 9.2 p.~75]{Simon-traceideals}.
The individual factors underbraced converge as follows:
\begin{itemize}
\item $\lim_{\varepsilon\rightarrow 0^+} \det_{[\frac{d}{2}]+1 }\left(Id+e^{-2\varepsilon\Delta}\Delta^{-1}V \right)=\det_{[\frac{d}{2}]+1 }\left(Id+\Delta^{-1}V \right)$ because $e^{-2\varepsilon\Delta}\Delta^{-1}V\rightarrow \Delta^{-1}V \in \Psi^{-2}(M,E)$ hence in the Schatten ideal
$\mathcal{I}_{[\frac{d}{2}]+1}$ and Gohberg--Krein's determinant
$H\mapsto \det_{[\frac{d}{2}]+1}(Id+H)$ depends continuously on 
$H\in\mathcal{I}_{[\frac{d}{2}]+1}$.
\item $\lim_{\varepsilon\rightarrow 0^+} \exp\left(\sum_{n=1}^{\frac{d}{2}} Tr_{L^2}\left( (e^{-2\varepsilon\Delta}\Delta^{-1}V)^n \right)-P_{n,\varepsilon}(V)  \right)=\exp\left(\sum_{n=1}^{\frac{d}{2}} \left\langle \mathcal{R}t_n,V^{\boxtimes n}\right\rangle \right)$
where $\mathcal{R}t_n$ is a distributional extension of 
$t_n=\mathbf{G}(x_1,x_2)\dots \mathbf{G}(x_n,x_1)$ 
by construction.
\end{itemize}
Thus it is immediate
that
$$\mathcal{R}\det(\Delta+V)=\lim_{\varepsilon\rightarrow 0^+}\text{det}_F\left(Id+e^{-2\varepsilon\Delta}\Delta^{-1}V\right)e^{-P_\varepsilon(V)}=\exp\left(\sum_{n=1}^{\frac{d}{2}} \left\langle \mathcal{R}t_n,V^{\boxtimes n}\right\rangle \right)\text{det}_{[\frac{d}{2}]+1 }\left(Id+\Delta^{-1}V \right) $$
hence it satisfies the representation formula~\ref{e:repsformula} which makes it a solution of problem~\ref{d:renormdet1}.
If we are given any other solution $\mathcal{R}_2\det$ of problem~\ref{d:renormdet1},
then by the free transitive action of $\mathcal{O}_{loc,[\frac{d}{2}]}$, we know that there exists $Q\in \mathcal{O}_{loc,[\frac{d}{2}]}$
s.t.
$\mathcal{R}_2\det(\Delta+V)=e^{Q(V)}\mathcal{R}\det(\Delta+V)=\lim_{\varepsilon\rightarrow 0^+}\text{det}_F\left(Id+e^{-2\varepsilon\Delta}\Delta^{-1}V\right)e^{(Q-P_\varepsilon)(V)}  $ 
which shows that $\mathcal{R}_2\det$ is obtained by renormalization by subtraction of local counterterms.

\section{Relation with Gaussian Free Fields.}

In the bosonic case, there is a
nice interpretation of
the renormalized determinants from
Theorem \ref{t:quillenconjanal} in terms of the Gaussian Free Field. 

\subsubsection{Probabilistic representation.}

We next briefly recall 
some probabilistic 
definition of the Gaussian Free Field (GFF)
associated to our positive 
elliptic operator $\Delta$ which is
represented
as a random distribution on $M$.
 
\begin{defi}[Bundle valued Gaussian Free Field]
\label{d:GFF}
Under the geometric assumption from definition~\ref{d:bosoncase}, if $\Delta: C^\infty(E)\mapsto C^\infty(E)
$ is \textbf{positive, self-adjoint} then
the Gaussian free field $\phi$ associated to
$\Delta$ is defined as follows:
denote by $(e_\lambda)_{\lambda\in\sigma(\Delta)}$ 
the spectral resolution
associated to $\Delta$. 
Consider a sequence $(c_\lambda)_{\lambda\in\sigma(\Delta)}, c_\lambda\in \mathcal{N}(0,1)$ 
of independent, identically distributed Gaussian random variables.
Then we define
the quantum field $\phi$ as the random series
\begin{equation}
\phi=\sum_{\lambda\in \sigma(\Delta) } \frac{c_\lambda}{\sqrt{\lambda}}e_\lambda
\end{equation}
where the sum runs over the eigenvalues of
$\Delta$ and
the series converges almost surely as distributional section in 
$\mathcal{D}^\prime(M,E)$.

The covariance of the Gaussian free field defined above
is the Green function:
\begin{eqnarray*}
\boxed{\mathbf{G}(x,y)=\sum_{\lambda\in \sigma(\Delta) } \frac{1}{\lambda}e_\lambda(x)\boxtimes e_\lambda(y)} 
\end{eqnarray*}
where the above series converges in $\mathcal{D}^\prime(M\times M,E\boxtimes E)$.
\end{defi}

A classical result characterizes the support of the functional measure:
\begin{lemm}[Regularity of bundle GFF]
Using the notations of definition \ref{d:GFF},
the random 
section $\phi$ converges almost surely in
the Sobolev space 
$H^s(E)$ for every $s<1-\frac{d}{2}$.
\end{lemm}

In Euclidean quantum 
field theory, there is an analogy between
considering a discrete GFF on a lattice with spacing $\sqrt{\varepsilon}$, 
whose propagator is a discrete 
Green function which is the inverse of the discrete Laplacian
and considering
the heat regularized GFF $\phi_\varepsilon=e^{-\varepsilon\Delta}\phi$ whose covariance
reads $e^{-2\varepsilon\Delta}\Delta^{-1}$.
For discrete Laplacians $\Delta_\varepsilon$ on a regular lattice of mesh $\varepsilon$,  
there are beautiful results 
on the asymptotics of $\det(\Delta_\varepsilon)$~\cite{chaumard2003} (see~\cite{Karl} for related results):
\begin{thm}\label{t:DGFF}
On the flat torus $\mathbb{T}^2$, for discrete Laplacian $\Delta_\varepsilon$ with mesh $\varepsilon$ and denote by $\phi_\varepsilon$ the corresponding \textbf{discrete GFF},
if $V\in C^\infty(\mathbb{T}^2)$ s.t. $\int_{\mathbb{T}^2}V=0$ then:
\begin{equation}
\frac{\det_\zeta\left(\Delta+V \right)}{\det_\zeta\left(\Delta\right)}=\lim_{\varepsilon\rightarrow 0}\frac{\det(\Delta_\varepsilon+V)}{\det(\Delta_\varepsilon)} = \lim_{\varepsilon\rightarrow 0^+}\mathbb{E}\left(e^{-\frac{1}{2}\int_{\mathbb{T}^2} V\phi_\varepsilon^2 } \right)^{-2}.
\end{equation}
\end{thm}

In the bosonic case, replacing lattice regularization by 
the heat regularized GFF, 
we prove an analog of the above Theorem and describe all renormalized determinants from Theorem~\ref{t:quillenconjanal} as coming from
the local renormalization of Gaussian free fields partition function as follows:
\begin{thm}[GFF representation]
\label{t:mainthm6GFFreps}
Under the assumptions of definition~\ref{d:GFF}. 
Let $\phi$ be the Gaussian free field with covariance 
$\mathbf{G}$.
Denote by $\phi_\varepsilon=e^{-\varepsilon\Delta}\phi$ the heat regularized GFF.

Then a function $V\mapsto\mathcal{R}\det\left(\Delta+V \right)$ is a renormalized determinant in the sense of 
definition~\ref{d:renormdet1} if and only if there exists a
sequence $\left(\Lambda_\varepsilon:C^\infty(E)\mapsto C^\infty\left(E \right)\right)_{\varepsilon\in (0,1]}$
of \textbf{smooth local} polynomial functionals of minimal degree such that the following limit exists:
\begin{eqnarray}
\mathcal{R}\det\left(Id+\Delta^{-1}V \right)^{-\frac{1}{2}}=\lim_{\varepsilon\rightarrow 0^+}\mathbb{E}\left(\exp\left(-\frac{1}{2}\int_M \left\langle \phi_\varepsilon, V\phi_\varepsilon\right\rangle - \Lambda_\varepsilon\left(V \right)(x)dv(x) \right) \right).
\end{eqnarray}
Furthermore, if $V\in C^\infty(End(E))$ defines a positive operator on $L^2(E)$, 
we denote by $\mu$ the Gaussian measure of covariance $\Delta^{-1}$ then
the limit of measures 
\begin{eqnarray}
\nu=\lim_{\varepsilon\rightarrow 0^+}\frac{\exp\left(-\frac{1}{2}\int_M \left\langle \phi_\varepsilon, V\phi_\varepsilon\right\rangle - \Lambda_\varepsilon\left(V \right)(x)dv(x) \right)}{\mathbb{E}\left(\exp\left(-\frac{1}{2}\int_M \left\langle \phi_\varepsilon, V\phi_\varepsilon\right\rangle - \Lambda_\varepsilon\left(V \right)(x)dv(x) \right)\right)}\mu  
\end{eqnarray}
exists as a Gaussian measure on $\mathcal{D}^\prime\left(M\right)$ with covariance
$\left(\Delta+V\right)^{-1} $ and  $\nu$
is absolutely continuous w.r.t. $\mu$ iff $1\leqslant d \leqslant 3$ otherwise 
the measures $\left(\nu,\mu\right)$ are mutually singular.
%
\end{thm}

The intuitive idea is very simple. In QFT the renormalization problem 
arises from the fact that fields are irregular distributions then a natural idea is to study a 
regularized version of
the field and see if one can perform an explicit renormalization of the partition function
by subtracting \textbf{explicit local counterterms} in the action functional. The first part of 
Theorem \ref{t:mainthm6GFFreps} follows from Theorem \ref{t:quillenconjanal}
once we reformulate the partition function
$\mathbb{E}( e^{-\int_M \left\langle\varphi_\varepsilon, V\varphi_\varepsilon\right\rangle } )$, where 
$\phi_\varepsilon=e^{-\varepsilon\Delta}\phi$ is the smeared GFF, in terms of Fredholm determinants
$\det_F\left(Id+\Delta^{-1}e^{-\varepsilon\Delta}Ve^{-\varepsilon\Delta} \right)$ which is the goal of the next paragraph.
In a companion paper~\cite[Prop 1.4]{DangWick}, we give a simple derivation of the above Theorem~\ref{t:mainthm6GFFreps}
using elementary commutator arguments when $\dim(M)\leqslant 4$.

\subsubsection{Fredholm determinants and partition functions.}
The following Lemma relates partition functions and Fredholm determinants:
\begin{lemm}[Field regularization.]
Under the assumptions of definition~\ref{d:GFF}, let
$\phi_\varepsilon=e^{-\varepsilon\Delta}\phi$ be the mollified GFF.

Then for every $\varepsilon>0$, the following relation holds true:
\begin{eqnarray*}
\mathbb{E}\left( \exp\left(-\frac{1}{2}\int_M  \left\langle \phi_\varepsilon , V \phi_\varepsilon\right\rangle  dv(x)\right) \right)
=\text{det}_F\left(Id+ e^{-\varepsilon\Delta}\Delta^{-1}e^{-\varepsilon\Delta} V\right)^{-\frac{1}{2}}.
\end{eqnarray*}
\end{lemm}
\begin{proof}
This is an immediate consequence of~\cite[Remark 1 p.~211]{GlimmJaffe87} which allows to write
$\mathbb{E}\left( \exp\left(-\frac{1}{2}\int_M  \left\langle \phi_\varepsilon , V \phi_\varepsilon\right\rangle  dv(x)\right) \right)
=\exp\left(-\frac{1}{2}Tr_{L^2}\left(Id+\widehat{V}_\varepsilon \right) \right) $ for $\Vert V\Vert_\infty$ small enough, where $$\widehat{V}_\varepsilon = e^{-\varepsilon\Delta}\Delta^{-\frac{1}{2}}V\Delta^{-\frac{1}{2}} e^{-\varepsilon\Delta}$$  is positive, self-adjoint and smoothing hence trace class on $L^2\left(E\right)$.
We can expand the term $Tr_{L^2}\log\left(Id+ \widehat{V}_\varepsilon\right)$
in power series and use the cyclicity of the $L^2$ trace to identify $\exp\left(-\frac{1}{2}Tr_{L^2}\left(Id+\widehat{V}_\varepsilon \right) \right)$ with the power series defining the Fredholm determinant \\$\text{det}_F\left(Id+ e^{-\varepsilon\Delta}\Delta^{-1}e^{-\varepsilon\Delta} V\right)^{-\frac{1}{2}}$. 
A very similar proof can be found in~\cite[subsubsections 3.0.1 and 3.0.2]{DangWick} where we relate the Wick renormalized partition function 
with the Gohberg--Krein determinant $\det_2$.
\end{proof}

\subsection{The renormalized functional measure.}

In the previous part, we have constructed renormalized functional
determinants to rigorously define the partition function. 
The following Proposition proves the second part of Theorem \ref{t:mainthm6GFFreps} 
and answers some
natural questions
about the corresponding renormalized functional measure.
\begin{prop}
Under the assumptions of definition~\ref{d:GFF}, assume $V\in C^\infty(End(E))$
is Hermitian.
Let $\mu$ denote the GFF measure on $\mathcal{D}^\prime\left(M,E\right)$ with covariance $\mathbf{G}$
which is the Schwartz kernel of $\Delta^{-1}$. Then there exists 
$P_\varepsilon(.)=\int_M\Lambda_\varepsilon\left(. \right)\in \mathcal{O}_{loc,[\frac{d}{2}]}\left(J^{d-3}E \right)\otimes_{\mathbb{C}}
\mathbb{C}[\varepsilon^{-\frac{1}{2}},\log(\varepsilon)] $ s.t.
the limit 
\begin{eqnarray*}
\nu=\lim_{\varepsilon\rightarrow 0^+} \exp\left(-\frac{1}{2}\int_M \left( \left\langle  \phi_\varepsilon,V\phi_\varepsilon\right\rangle -\Lambda_\varepsilon(V) \right) dv(x) \right)\mu
\end{eqnarray*}
converges to a Gaussian measure  on $\mathcal{D}^\prime\left(M,E\right)$ 
which is
\textbf{absolutely continuous} w.r.t. $\mu$ if $d=(2,3)$ and 
the measure $\left(\mu,\nu\right)$ \textbf{are mutually singular} when $d\geqslant 4$.
\end{prop}

$\Lambda$ depends on  
the $(d-3)$-jet of $V$ in the above proposition.

\begin{proof}
Define 
$\nu_\varepsilon= \frac{\exp\left(-\frac{1}{2}\int_M \left( \left\langle  \phi_\varepsilon,V\phi_\varepsilon\right\rangle -\Lambda_\varepsilon(V) \right) dv(x) \right)}{\mathbb{E}\left( \exp\left(-\frac{1}{2}\int_M \left( \left\langle  \phi_\varepsilon,V\phi_\varepsilon\right\rangle -\Lambda_\varepsilon(V) \right) dv(x) \right) \right)}\mu $ for
$\varepsilon>0$. This is a Gaussian measure whose covariance is
$
\left(\Delta+ e^{-\varepsilon\Delta} Ve^{-\varepsilon\Delta}\right)^{-1}$
by~\cite[Prop 9.3.2 p.~213]{GlimmJaffe87}.
When $\varepsilon\rightarrow 0^+$, this covariance converges
to $\left(\Delta+V\right)^{-1}$ as \textbf{bilinear forms on} $C^\infty(M)\times C^\infty(M)$ for the weak 
topology~\cite[iv) p.~208]{GlimmJaffe87}
since $e^{-\varepsilon\Delta}\rightarrow Id$ in the 
\textbf{strong operator topology} when $\varepsilon\rightarrow 0^+$.
A necessary and sufficient condition for the renormalized measure to be 
absolutely continuous w.r.t. the initial measure is given by a Theorem 
of Shale~\cite[Thm I.23 p.~41]{SimonPhi2} is that 
$\Delta^{-\frac{1}{2}}V\Delta^{-\frac{1}{2}}\in \Psi^{-2}(M,E)$ is Hilbert--Schmidt which holds true
only if $\dim(M)=d\leqslant 3$.
\end{proof}

\section{Quillen's determinant line bundle.}

We recall the definition of Quillen's determinant line bundle which is an adaptation of the definition
of Segal~\cite[p.~137-138]{segal2004definition}, Furutani~\cite{furutani2004quillen} and Melrose--Rochon~\cite{melrose2007periodicity}
where holomorphicity properties 
are manifest. The reader can also look at~\cite[section 5.3 p.~642]{scott2010traces} for a very 
nice account of determinant line bundles
for families of $\Psi$dos.

\begin{defi}[Quillen's universal determinant line bundle]
Using the notations of subsubsection \ref{ss:quillenpict}.
Recall $\mathcal{I}_1(\mathcal{H})$ denotes the ideal of trace class operators on some Hilbert space $\mathcal{H}$.
Let $(T_b)_{b\in B}$ be a holomorphic family of Fredholm operators from $\mathcal{H}_0\mapsto \mathcal{H}_1$ of index $0$, parametrized by a complex Banach manifold $B$.
Consider the 
bundle $$\mathcal{G}=\bigcup_{b\in B} T_b(Id+\mathcal{I}_1(\mathcal{H}_1))\simeq  B\times (Id+\mathcal{I}_1(\mathcal{H}_1))$$ which fibers over 
the complex Banach manifold $B$.

Then we define the determinant line bundle $\textbf{Det}\mapsto B$
to
be the quotient $\mathcal{G}  \times  \mathbb{C}/\sim$
where $(A(Id+T),z)\sim (A,\det_F(Id+T)z)$. 
The canonical section $\underline{\det}(T)$ is defined
to be the equivalence class $T\mapsto [T,1]$.
\end{defi}
This definition is functorial since it works for any holomorphic family $(T_b)_{b\in B}$ and holomorphicity 
is checked as in the work of Furutani~\cite{furutani2004quillen}. Quillen's line bundle is recovered by letting
$B$ to be the space $\textbf{Fred}_0\left(\mathcal{H}_0,\mathcal{H}_1\right)$ 
of Fredholm operators of index $0$ as proved by Furutani~\cite[section 2 and prop 2.1]{furutani2004quillen}.
Let us recall that 
\begin{lemm}
The canonical section $T\mapsto \underline{\det}(T)=[T,1]$ vanishes if and only if
$T$ is non-invertible. 
\end{lemm}
\begin{proof}
$[T,1]\simeq [\tilde{T},0]$ means there exists $Id+A, A\in \mathcal{I}^1$ s.t.
$\tilde{T}(Id+A)=T$ and $\det_F(Id+A)=0$ hence $Id+A$ is non-invertible and
so is $T$. Conversely, even if $T$ is non-invertible, there is a finite rank operator $t$ such that $T+t$ invertible
since $T$ is Fredholm of index $0$. Therefore 
$T=(T+t)(Id-(T+t)^{-1}t) $ where $(Id-(T+t)^{-1}t)$ is in the determinant class and is non-invertible. Finally
$[T,1]\sim [T+t,\det_F((Id-(T+t)^{-1}t))]=0  $. 
\end{proof}

\section{Proof of Theorem~\ref{t:quillenholotriv}.}

We follow the notations 
from subsubsection \ref{ss:quillenpict}. 
The way Quillen trivializes the line bundle is by constructing a smooth
Hermitian metric on $\mathcal{L}$ named Quillen's metric and he 
calculates 
explicitely the curvature
of the corresponding Chern connection
which is exactly the K\"ahler form on $\mathcal{A}$.
Then he shows that by modifying the Hermitian metric, one 
can produce a modified 
Chern connection $\nabla$ which is flat. 
It follows from 
the contractibility
of $\mathcal{A}$ 
that flat sections 
for $\nabla$ trivialize 
$\mathcal{L}$ \textbf{holomorphically}.
Here the setting is slightly different. Our 
approach to holomorphic trivialization is more direct and does not use Quillen metrics. 
We already know that
the canonical section $\iota^*\underline{\det}$ has the same 
zeros
on $\mathcal{A}$
as any solution $\mathcal{R}\det$ of Theorem~\ref{t:quillenconjanal}. Hence, we expect that 
the ratio $\frac{\iota^*\underline{\det}(D)}{\mathcal{R}\det(D)}$ is holomorphic without zeros on $\mathcal{A}$. 
It remains to show that this is well--defined and locally bounded in order to conclude that the ratio is a holomorphic section
without zeros by proposition \ref{p:douadyweakstrongholo}, 
hence it yields a holomorphic trivialization of $\mathcal{L}$.

Following Segal and Furutani, 
we define
open sets $U_t\subset \mathbf{Fred}_0(\mathcal{H})$ indexed by 
finite rank 
operators $t$ such that
$U_t=\{ T\in \mathbf{Fred}_0(\mathcal{H}) \text{ s.t. } T+t\text{ invertible} \}$.
Since elements in $\mathbf{Fred}_0(\mathcal{H})$ have Fredholm index $0$, 
the collection $(U_t)_t$ forms
an open cover of $\mathbf{Fred}_0(\mathcal{H})$. 
Then we trivialize 
$\mathcal{L}$ over $U_t$ by the never vanishing section
$T\in U_t\mapsto [ T+t,1]$ which is holomorphic by the proof of Furutani.
In the local trivialization, the canonical section  
$$T\mapsto\underline{\det(T)}=[T,1]$$ is identified
with the holomorphic function
$\det_F(Id-(T+t)^{-1}t)$ since
$
[T,1]\sim [T+t,\det_F(Id-(T+t)^{-1}t) ]=\det_F(Id-(T+t)^{-1}t)[T+t,1].$

Now we shall prove a technical 
\begin{lemm}\label{l:holodivision}
Let $T_0$ be an invertible operator in $\iota\left(\mathcal{A}\right)$ such that
for all 
$T\in \iota\left(\mathcal{A}\right)$,
$T-T_0$ is in the Schatten ideal $\mathcal{I}_{[\frac{d}{k}]+1}$, $k=(1,2)$.
 
 It follows that for $p=[\frac{d}{k}]+1$, Gohberg--Krein's determinant
$\det_p\left(Id+T_0^{-1}\left(T-T_0 \right)\right)$
is holomorphic on $\iota\left(\mathcal{A}\right)$.
Then the
section
\begin{eqnarray}
T\in \iota\left(\mathcal{A}\right)\longmapsto \text{det}_p\left(Id+T_0^{-1}\left(T-T_0 \right)\right)^{-1}\underline{\det(T)}
\end{eqnarray}
defines a global holomorphic section of $\mathbf{Det}\mapsto \iota\left(\mathcal{A}\right)$ which never vanishes.
\end{lemm}
\begin{proof}
It suffices to prove the claim on each open subset $U_t\cap \iota\left(\mathcal{A} \right)$ where the canonical section 
$T\mapsto \underline{\det(T)}$ is identified with 
$T\in U_t\mapsto \det_F(Id-(T+t)^{-1}t)$ by the local trivialization.

Use the identity
$Id+T_0^{-1}(T-T_0)=T_0^{-1}T  $ and $Id-(T+t)^{-1}t=(T+t)^{-1}T $.
By the multiplicativity of Fredholm determinants, for every invertible $T\in U_t\cap \iota\left(\mathcal{A} \right)$, we find that
~\footnote{For every trace class $H$, 
we are using the fact that $(Id+H)^{-1}\in Id+\mathcal{I}_1 $ and $\det_F((Id+H)^{-1})=\det_F(Id+H)^{-1}$. 
}
\begin{eqnarray*}
&&\text{det}_F(Id-(T+t)^{-1}t)\text{det}_p\left(Id+T_0^{-1}\left(T-T_0 \right)\right)^{-1}\\
&=&\text{det}_F(Id-(T+t)^{-1}t)\text{det}_F\left(Id+R_p\left(T_0^{-1}\left(T-T_0 \right)\right)\right)^{-1}\\  
&=&\text{det}_F\left((T+t)^{-1}T_0(Id+T_0^{-1}(T-T_0))(Id+R_p(T_0^{-1}(T-T_0)))^{-1} \right).
\end{eqnarray*}
For every $T\in U_t\cap \iota\left(\mathcal{A} \right)$, the operator
$(T+t)^{-1}T_0$ is invertible. 
For such $T$, we observe by definition of $R_p$ that
\begin{eqnarray*}
&&(Id+T_0^{-1}(T-T_0))(Id+R_p(T_0^{-1}(T-T_0)))^{-1}\\&=&(Id+T_0^{-1}(T-T_0))
(e^{\sum_{k=1}^{p-1}\frac{(-1)^{k+1}}{k}(T_0^{-1}(T-T_0))^k }(Id+T_0^{-1}(T-T_0)))^{-1}\\
&=&e^{\sum_{k=1}^{p-1}\frac{(-1)^{k}}{k}(T_0^{-1}(T-T_0))^k } 
\end{eqnarray*}
where the term $e^{\sum_{k=1}^{p-1}\frac{(-1)^{k}}{k}(T_0^{-1}(T-T_0))^k }$ is well--defined thanks to the holomorphic
functional calculus for the compact operator $T_0^{-1}(T-T_0)$ and is easily seen to be invertible by the spectral mapping theorem for 
holomorphic functions of bounded operators. 
Indeed, for every bounded operator $A$, 
$f:\Omega\subset \mathbb{C}\mapsto \mathbb{C}$ holomorphic 
in some neighborhood $\Omega$ of $\sigma(A)$, $f(A)$ 
is well--defined with $\sigma(f(A))=f(\sigma(A))$ 
by the spectral mapping Theorem~\cite[Thm 2.3.6 p.~22]{McIntoshlectures}. In our case, this gives that $0$ is not in the spectrum of $e^{\sum_{k=1}^{p-1}\frac{(-1)^{k}}{k}(T_0^{-1}(T-T_0))^k }$.
Finally 
$$(T+t)^{-1}T_0(Id+T_0^{-1}(T-T_0))(Id+R_p(T_0^{-1}(T-T_0)))^{-1}=(T+t)^{-1}T_0e^{\sum_{k=1}^{p-1}\frac{(-1)^{k}}{k}(T_0^{-1}(T-T_0))^k }$$ is invertible for every 
$T\in U_t\cap \iota\left(\mathcal{A} \right)$ and
is the composition of two operators of the form $Id+\mathcal{I}_1$ and $(Id+\mathcal{I}_1)^{-1}$ 
hence it belongs to 
the determinant class. 
Therefore, its Fredholm determinant never vanishes.
It follows that
$T\in U_t\cap \iota\left(\mathcal{A} \right)\cap \text{invertible} \mapsto \text{det}_F(Id-(T+t)^{-1}t)\text{det}_p\left(Id+T_0^{-1}\left(T-T_0 \right)\right)^{-1}$
extends uniquely 
as a never vanishing 
holomorphic function
on $ U_t\cap \iota\left(\mathcal{A} \right) $. 
\end{proof}

Lemma \ref{l:holodivision} says the ratio
$ P+\mathcal{V}\in \mathcal{A}\mapsto  \det_{[\frac{d}{k}]+1}(Id+P^{-1}\mathcal{V})^{-1} \iota^*\underline{\det}(P+\mathcal{V})$
never vanishes over $\mathcal{A}$. Furthermore
Corollary \ref{c:renormvshadamard} states that $\mathcal{R}\det(P+\mathcal{V})=\exp(g(\mathcal{V})) \det_{[\frac{d}{k}]+1}(Id+P^{-1}\mathcal{V}) $
where $g$ is a polynomial function, therefore $\exp(g(\mathcal{V}))$ never vanishes and
the holomorphic section
$\sigma: P+\mathcal{V}\in \mathcal{A}\mapsto \mathcal{R}\det(P+\mathcal{V})^{-1}\iota^*\underline{\det}(P+\mathcal{V})$
never vanishes over $\mathcal{A}$ and defines a holomorphic trivialization of $\mathcal{L}$:
$ \tau: \mathcal{O}\left(\mathcal{L}\right)\mapsto Hol(\mathcal{A}) $
such that the canonical section $\iota^*\underline{\det}(T)$ is sent to
the entire function $T\in \mathcal{A}\mapsto\mathcal{R}\det(T)$.
The second claim follows
from the action of the renormalization group as in Theorem~\ref{t:quillenconjanal}.
Finally, every non vanishing section
$\sigma$ defines canonically a flat connection $\nabla$ whose flat section
is $\sigma$.

\section{Appendix.}

\subsection{Wave front set of Schwartz kernels of local polynomial functionals.}

We give the proof of the following
\begin{lemm}\label{l:localityappendix}
Let $P$ be a continuous polynomial function on $C^\infty(M)$ such that $P$ is local in the sense 
\begin{equation}\label{e:hammersteinweak}
D^2P(w; u,v)=0
\end{equation} 
when $(u,v)$ have disjoint supports and the linear term of $P$ is given by integration against a smooth
function.
If $WF\left([\mathbf{D^2P(V)}]\right)\subset N^*\left(d_2\subset M^2 \right)$ for all $V\in C^\infty(M)$
then $P\in \mathcal{O}_{loc}(C^\infty(M))$. 
\end{lemm}
\begin{proof}
Equation \ref{e:hammersteinweak}
implies that
all G\^ateaux differentials $D^nP(0)$ of $P$ at $0\in C^\infty(M)$
have their Schwartz kernels $[\mathbf{D^nP}(0)]\in \mathcal{D}^\prime(M^n)$ supported on the deepest diagonal
$d_n\subset M^n$ by~\cite[Proposition V.5]{brouder2018properties} and that $P$ is \textbf{additive} in the sense of~\cite{brouder2018properties}.
Since $P$ is a polynomial function, it equals its Taylor expansion
$P(V)=\sum_{n=1}^{\deg(P)} P_n(V)$ where $P_n$ 
homogeneous of degree $n$. 

The smoothness condition on the linear term in $P$ together with 
the microlocal condition on $[\mathbf{D^2P(V)}] \in \mathcal{D}^\prime(M\times M)  , \forall V\in C^\infty(M)$ imply that
$D P(0)$ is represented by integration against smooth function $[\mathbf{DP(0)}]\in C^\infty(M)$.

Therefore by uniqueness of the Taylor expansion
each $P_n$ satisfies equation \ref{e:hammersteinweak}. 
Let $\tilde{P}_n$ be the multilinear map corresponding to 
$P_n$ and its Schwartz kernel
$[\mathbf{P_n}]\in \mathcal{D}^\prime(M^n)$ whose existence is given by the kernel Theorem~\cite{brouder2018properties}.
The Schwartz kernel $[\mathbf{P_n}]\in\mathcal{D}^\prime(M^n)$ is a distribution 
carried by the deepest diagonal by locality of $P_n$. By a Theorem of Laurent Schwartz, 
$[\mathbf{P_n}]$
has an expression in local coordinates $(x_1,\dots,x_n)$ in $U^n$ 
as 
$$[\mathbf{P_n}](x_1,\dots,x_n)=\sum_{[\alpha]} f_{[\alpha]}(x_1)\partial^{\alpha_2}_{x_2}\dots\partial^{\alpha_n}_{x_n}\partial^{\alpha_1}_{x_1}\delta_{\{0\}}^{\mathbb{R}^{d(n-1)}}(x_1-x_2,\dots,x_1-x_n) $$
where the sum over the multiindices $[\alpha]=(\alpha_1,\dots,\alpha_n)\in \mathbb{N}^{dn}$ is finite and 
$f_{[\alpha]}$ is a distribution in the variable $x_1$. 
It follows that the Schwartz kernel of the second G\^ateaux differential
has the representation in local coordinates
$$[\mathbf{D^2P(V)}](x,y)=\sum_{[\alpha]} f_{[\alpha]}(x)\partial^{\alpha_2}_{x}V(x)\dots \partial^{\alpha_n}_{x}V(x)\partial^{\alpha_1}_{y}\delta^{\mathbb{R}^d}_{\{0\}}(x-y)$$
which implies $P$ satisfies condition 2 of~\cite[Lemma VI.9]{brouder2018properties}.
By~\cite[Lemma VI.9]{brouder2018properties}, this means
$V\in C^\infty(M)\mapsto [\mathbf{DP(V)}]\in C^\infty(M)$ is smooth. To summarize, $P$ is additive, 
its differential $DP(V)$ is represented by integration against a smooth function $[\mathbf{DP(V)}]\in C^\infty(M)$ and
$V\in C^\infty(M)\mapsto [\mathbf{DP(V)}]\in C^\infty(M)$ is smooth hence by~\cite[Theorem I.2]{brouder2018properties}, $P\in \mathcal{O}_{loc}(C^\infty(M))$.
\end{proof}

\subsection{Sharpness of the bound from the main Theorem }

We give an application of the Hadamard Theorem~\ref{t:hadamard}
by giving an example where the 
bound from Theorem~\ref{t:quillenconjanal} on the order of the entire
function $z\mapsto \mathcal{R}\det(P+z\mathcal{V})$ is sharp.

\begin{lemm}\label{c:optimalorder}
Let $\Delta$ be the Laplace--Beltrami operator of some Riemannian manifolds $(M,g)$
of dimension $d$. For any entire function $f:\mathbb{C}\mapsto \mathbb{C}$ 
s.t. $f(z)=0 \Leftrightarrow \ker\left(\Delta+z\right)\neq \{0\} $ with multiplicity
$\dim\left(\ker(\Delta+z)\right)$, we must have the order
$\rho(f)\geqslant [\frac{d}{2}]+1$.
\end{lemm} 
This proves the bound
from problem~\ref{d:renormdet1} is optimal.
\begin{proof}
Note that
$f(z)=0 \implies -z\in \sigma\left(\Delta\right)$.
By Weyl's law for spectral functions of
positive, elliptic pseudodifferential operators~\cite[Thm 2.1 p.~825]{GW14}, 
the number of eigenvalues $n_L(\Delta)$ of $\Delta$ less than $L$
grows like  a symplectic volume 
$\int_{\{\sigma(\Delta)(x;\xi)\leqslant L\}\subset T^*M} d^dxd^d\xi \sim_{L\rightarrow +\infty}CL^{\frac{d}{2}} $.
This implies for $p=[\frac{d}{2}]+1$ that
$ Tr_{L^2}\left( \Delta^{-p}\right) =\sum_{z\in \{f=0\} }\vert z \vert^{-p}<+\infty $ and
$ Tr_{L^2}\left( \Delta^{1-p}\right) = \sum_{z\in \{f=0\}} \vert z\vert^{-p+1}=+\infty$
hence $\rho(f)\geqslant [\frac{d}{2}]+1$ by Theorem~\ref{t:hadamard}.
\end{proof}
%
%
%
%
%

Both results show that the solution to the problem
of finding entire functions with prescribed zeros is 
\textbf{not unique}, the non unicity is due to
the critical exponents of zeros which forces 
the entire function to have non zero order.
So
there is an ambiguity relating all possible 
solutions of the problem which is of the form
$\exp(\text{Polynomial})$ by Hadamard's factorization Theorem.

\subsubsection{Proof of Lemma~\ref{l:gapspec1}.}
The composite operator
$\Delta^{-\frac{1}{4}}V\Delta^{-\frac{1}{4}}$ is a pseudodifferential of order $0$ in 
$\Psi^{0}(M,E) $ by the composition Theorem. Therefore
by the Calderon Vaillancourt Theorem, we can choose $V\in \text{Diff}^1(M,E)$ in some small enough neighborhood $\mathcal{U}$ of $0$ 
so that
$\max\left(\Vert\Delta^{-\frac{1}{4}}V\Delta^{-\frac{1}{4}}\Vert_{\mathcal{B}(H^{\frac{1}{2}},H^{\frac{1}{2}})},\Vert\Delta^{-\frac{1}{4}}V^*\Delta^{-\frac{1}{4}}\Vert_{\mathcal{B}(H^{\frac{1}{2}},H^{\frac{1}{2}})} \right) \leqslant \frac{\sqrt{\delta}}{4} .$
This yields for every $z$:
\begin{eqnarray*}
&& Re\left\langle u,\left(\Delta+V-z\right) u \right\rangle=Re\left\langle \Delta^{\frac{1}{4}}u, \left( \Delta^{\frac{1}{2}}+\Delta^{-\frac{1}{4}}V\Delta^{-\frac{1}{4}}-z\Delta^{-\frac{1}{2}}\right) \Delta^{\frac{1}{4}}u \right\rangle\\
&\geqslant & C\left(\sqrt{\delta}\Vert u\Vert_{H^{\frac{1}{2}}}^2-\Vert \Delta^{-\frac{1}{4}}V\Delta^{-\frac{1}{4}}\Vert_{\mathcal{B}(H^{\frac{1}{2}},H^{\frac{1}{2}})}\Vert u\Vert_{H^{\frac{1}{2}}}^2-Re(z)\delta^{-\frac{1}{2}}\Vert u\Vert_{H^{\frac{1}{2}}}^2\right)\\
&\geqslant &C\left(\sqrt{\delta}-\Vert \Delta^{-\frac{1}{4}}V\Delta^{-\frac{1}{4}}\Vert_{\mathcal{B}(H^{\frac{1}{2}},H^{\frac{1}{2}})}-\frac{Re(z)}{\sqrt{\delta}} \right) \Vert u\Vert_{H^{\frac{1}{2}}}^2,
\end{eqnarray*}
where $C$ is some constant such that $C\Vert u \Vert_{H^{\frac{1}{2}}} \leqslant \Vert \Delta^{\frac{1}{4}}
 u \Vert_{L^2}  \leqslant C^{-1}\Vert u \Vert_{H^{\frac{1}{2}}}$.
Hence when $Re(z)\leqslant \frac{\delta}{2} $:
$ Re\left\langle u,\left(\Delta+V-z\right) u \right\rangle\geqslant C\frac{\sqrt{\delta}}{4}\Vert u\Vert_{L^2}^2 .$
We have a similar estimate for the adjoint $V^*$ which implies that
$\{Re(z)\leqslant \frac{\delta}{2} \}$ lies in the resolvent set of $(\Delta+V+z)$.

\subsection{Proof of Proposition~\ref{p:heatcomplex}.}
Recall $(\Delta+B)_\pi^{-s}$ is the complex power defined in terms of the spectral cut at angle $\pi$.
The first formula we need to establish:
$$(\Delta+B)_\pi^{-s}= \frac{1}{\Gamma(s)} \int_0^\infty e^{-t(\Delta+B)} t^{s-1} dt$$
is widely used in the mathematical physics litterature to define complex powers of Schr\"odinger type operators. Since we work with non-self-adjoint operators, we need to justify it.
We first define 
$$(\Delta+B)^{-s}= \frac{1}{\Gamma(s)} \int_0^\infty e^{-t(\Delta+B)} t^{s-1} dt$$
where the integral on the r.h.s, which is valued in $\mathcal{B}(H^s,H^s)$, converges 
since on $\int_1^\infty$ we use the exponential decay of the semigroup and 
on $\int_0^1$, it is well defined for $Re(s)>0$. So we just defined a holomorphic 
family of operators $((\Delta+B)^{-s})_{Re(s)>0}: C^\infty(M) \mapsto \mathcal{D}^\prime(M) $. 
To extend it to the complex plane 
and to make the connection with actual complex powers, 
we shall identify it with 
the definition of complex powers using the contour integral and resolvent
instead of the Mellin transform of 
the heat kernel.
In Gilkey's book~\cite{Gilkey}, 
the heat operator $e^{-t(\Delta+B)}$ for non-self-adjoint operators is expressed 
in terms of the resolvent by the contour integral
$$ e^{-t(\Delta+B
)} = \frac{i}{2\pi} \int_\gamma e^{-t\lambda}\left(\Delta+B-\lambda \right)^{-1}d\lambda  $$
where the contour integral converges for $t>0$ by the exponential decay of $e^{-t\lambda}$
since the contour $\gamma$, oriented clockwise, is chosen to be some $V$ shaped curve which contains strictly the angular sector $\mathcal{R}$ from Lemma~\ref{l:resolventbound} and $\gamma$ is contained in the half--plane $Re(\lambda)\geqslant 0$. 
We saw that the complex power 
$(\Delta+B)^{-s}_\pi$ is defined using the spectral cut at angle $\pi$:
\begin{eqnarray*}
(\Delta+B)^{-s}_\pi=\frac{i}{2\pi}\int_{\tilde{\gamma}} \lambda^{-s}\left(\Delta+B-\lambda \right)^{-1}d\lambda
\end{eqnarray*}
where the operator valued integral converges absolutely for $Re(s)>1$ in $\mathcal{B}(L^2,L^2)$ and $\tilde{\gamma}=\{  re^{i\pi}, \infty>r\geqslant \rho \}\cup \{  \rho e^{i\theta}, \theta\in [\pi,-\pi] \}\cup \{re^{-i\pi}, \rho\leqslant r <\infty \}$.

Since the spectrum of $ \Delta+B$ is contained in some neighborhood of $[\frac{\delta}{2},+\infty) $,
we can deform the contour 
$\tilde{\gamma}$ to the contour $\gamma$ used in defining the heat 
operator without crossing $\sigma(\Delta+B)$.
Using the estimates on the resolvent of Lemma~\ref{l:resolventbound}, Cauchy's formula and contour deformation avoiding $\sigma(\Delta+B)$, it is simple to show that
$$ \int_{\tilde{\gamma}} \lambda^{-s}\left(\Delta+B-\lambda \right)^{-1}d\lambda=\int_{\gamma} \lambda^{-s}\left(\Delta+B-\lambda \right)^{-1}d\lambda $$
where both sides converge absolutely for $Re(s)>1$.
Now we use the formula 
$\lambda^{-s} = \frac{1}{\Gamma(s)}\int_0^\infty t^{s-1} e^{-t\lambda}dt $
which makes sense for the branch $\log(\lambda)=\log(\vert \lambda\vert)+i\arg(\lambda), -\pi\leqslant \arg(\lambda)\leqslant \pi$ since $Re(\lambda)>0$ and $Re(s)>1$.
Therefore
\begin{eqnarray*}
&&(\Delta+B)^{-s}_\pi=\frac{i}{2\pi}\int_{\gamma} \left(\frac{1}{\Gamma(s)}\int_0^\infty t^{s-1} e^{-t\lambda}dt  \right)\left(\Delta+B-\lambda \right)^{-1}d\lambda\\
&=&\frac{i}{2\pi} \frac{1}{\Gamma(s)}\int_0^\infty t^{s-1}\left( \int_{\gamma} e^{-t\lambda}\left(\Delta+B-\lambda \right)^{-1}d\lambda\right) dt
=\frac{i}{2\pi}\frac{1}{\Gamma(s)}\int_0^\infty t^{s-1}e^{-t(\Delta+B)} dt
\end{eqnarray*}
where we could invert the integrals since everything converges when $Re(s)>1$.
The above discussion also shows that for any differential operator $Q\in \text{Diff}^1(M,E)$, we have
$$ (\Delta+B)_\pi^{-s}= \frac{1}{\Gamma(s)} \int_0^\infty e^{-t(\Delta+B)} t^{s-1} dt\in \mathcal{B}(L^2,H^{-1}) .$$

This proves that one can define the complex powers with spectral cut 
using the heat kernel even in this non-self-adjoint setting.

\subsubsection{Taking the trace.}

We want to prove the relation
$$ Tr_{L^2}\left(Q(\Delta+B)^{-s}_\pi\right)=\frac{1}{\Gamma(s)}\int_0^\infty t^{s-1}Tr_{L^2}\left(Qe^{-t(\Delta+B)}\right) dt$$
which means that we want to take the functional trace 
on both sides of the previous identity.
Start again from the relation
$Q(\Delta+B)^{-s}_\pi=\frac{i}{2\pi}\frac{1}{\Gamma(s)}\int_0^\infty t^{s-1}Qe^{-t(\Delta+B)} dt.$
One key idea is that a sufficiently smoothing operator will be trace class and has continuous Schwartz kernel. For such operator, the $L^2$ trace coincides with the flat trace $Tr^\flat$ defined simply by integrating the 
Schwartz kernel of the operator restricted on the diagonal against a smooth density. 

By the work of Seeley, we know that
$ (\Delta+B)^{-s}_\pi\in \Psi^{-2s}(M)$, then by composition of pseudodifferential operators
$Q(\Delta+B)^{-s}_\pi\in \Psi^{-2s+1}(M) $. This implies that as soon as
$Re(s)>\frac{d+1}{2}$, $Q(\Delta+B)^{-s}_\pi$ is trace class and the left hand side 
$Tr_{L^2}\left(Q(\Delta+B)^{-s}_\pi \right)$ is well--defined and 
$Tr_{L^2}\left(Q(\Delta+B)^{-s}_\pi \right)=\frac{1}{\Gamma(s)}Tr_{L^2}\left(\int_0^\infty t^{s-1}Qe^{-t(\Delta+B)} dt\right) $. 
To exchange the trace and the integral on the r.h.s, 
note that 
$Tr_{L^2}\left(Q(\Delta+B)^{-s}_\pi \right)=Tr^\flat\left(Q(\Delta+B)^{-s}_\pi \right) $ since $Q(\Delta+B)^{-s}_\pi$ is trace class when $Re(s)>\frac{d+\deg(Q)}{2}$ and has continuous kernel arguing as in~\cite[p.~102--103]{Shubin}. Therefore
$$Tr_{L^2}\left(Q(\Delta+B)^{-s}_\pi \right)=  \frac{1}{\Gamma(s)} Tr^\flat\left(\lim_{\varepsilon\rightarrow 0^+}\lim_{\Lambda\rightarrow +\infty} \int_\varepsilon^\Lambda t^{s-1}Qe^{-t(\Delta+B)} \right) dt.$$ However note that for $t$ in $[\varepsilon,\Lambda]$, it is immediate to prove that
$t\mapsto Qe^{-t(\Delta+B)}$ is continuous and uniformly bounded in smoothing operators, therefore we can invert the flat traces and
the integral to get
$Tr^\flat\left(\int_\varepsilon^\Lambda t^{s-1}Qe^{-t(\Delta+B)} \right)=\int_\varepsilon^\Lambda t^{s-1}Tr^\flat\left( Qe^{-t(\Delta+B)} \right)=
\int_\varepsilon^\Lambda t^{s-1}Tr_{L^2}\left( Qe^{-t(\Delta+B)} \right)$ since $\left( Qe^{-t(\Delta+B)} \right)\in \Psi^{-\infty}$.
To conclude,
it suffices to show that under the assumption
that $Re(s)>\frac{d+1}{2}$, the integrand
$t^{s-1}Tr_{L^2}\left(Qe^{-t(\Delta+B)} \right)$ is Riemann integrable on $(0,+\infty)$.
But this follows almost immediately from the bound (using the fact that $Q$ is a differential operator of degree $1$)~\cite[Lemma 1.9.3 p.~77--78]{Gilkey},\cite[Thm 2.30 p.~87]{BGV}
\begin{eqnarray*}
\forall t\in (0,1], \vert Tr_{L^2}\left(Qe^{-t(\Delta+B)} \right)\vert\leqslant Ct^{-\frac{(d+1)}{2}}\implies \vert t^{s-1}Tr_{L^2}\left(Qe^{-t(\Delta+B)} \right)\vert\leqslant Ct^{Re(s)-1-\frac{(d+1)}{2}}
\end{eqnarray*}
where the r.h.s. is absolutely integrable near $0$ and
\begin{eqnarray*}
\forall t\in [1,+\infty),  \vert Tr_{L^2}\left(Qe^{-t(\Delta+B)} \right)\vert\leqslant \Vert e^{-(t-\frac{1}{2})(\Delta+B)}\Vert_{\mathcal{B}(L^2,L^2)} \Vert Qe^{-\frac{1}{2}(\Delta+B)} \Vert_{\mathcal{B}(L^2,H^{r})}  \leqslant  C e^{(t-\frac{1}{2})\frac{\delta}{2}}
\end{eqnarray*}
for any $r>d$,
which uses the exponential decay of the semigroup $e^{-t(\Delta+B)}$, the smoothing properties of
$  Qe^{-\frac{1}{2}(\Delta+B)} $ and allows us to control the integral for 
large times. Finally once the identity is proved in some domain $Re(s)>\frac{d+1}{2}$, the analytic continuation takes care of extending the relation on the whole complex plane. 

\subsubsection{Proof of Lemma~\ref{l:convheat}.}
\begin{proof}
For every real number $s$,
a symbol
$p\in S^{s}_{1,0}\left(\mathbb{R} \right)$ iff $p$ is in $C^\infty\left( \mathbb{R}\right)$ and 
$ \vert \partial_\xi^jp(\xi) \vert\leqslant 
C_j\left(1+\vert \xi\vert\right)^{s-j} $~\cite[Lemm 1.2 p.~295]{Taylorpsi} for every $j\in \mathbb{N}$.
Observe that the function $p_t:\xi\in \mathbb{R}\mapsto e^{-t\vert\xi\vert^2}$ 
defines a family $(p_t)_{t\in [0,+\infty)}$ of symbols in
$S^0_{1,0}\left(\mathbb{R} \right)$ such that $p_t\underset{t\rightarrow 0}{\rightarrow} 1$ 
in $S^{+0}_{1,0}\left(\mathbb{R} \right)$.
Indeed, for $k\in \mathbb{N}$ and
for $t$ in some compact interval $[0,a], a>0$, we find by direct computation that:
$(1+\vert\xi\vert)^k\vert\partial_\xi^ke^{-t\xi^2} \vert\leqslant C(1+\vert\xi\vert)^k\sum_{0\leqslant l\leqslant \frac{k}{2}} t^{k-l}\vert\xi\vert^{k-2l} e^{-t\xi^2}$ where the constant $C$ depends
only on $k$.\\ 
 When $\vert\xi\vert\geqslant a$, the function
$t\in [0,+\infty)\mapsto (t^{k-l}\xi^{k-2l})e^{-t\xi^2}$ goes to $0$ when $t=0,t\rightarrow +\infty$ and 
reaches its maximum when $\frac{d}{dt}\left( (t^{k-l}\xi^{k-2l})e^{-t\xi^2}\right)
=((k-l)t^{k-l-1}\xi^{k-2l} - t^{k-l} \xi^{k-2l+2}) e^{-t\xi^2}
=((k-l)-t\xi^2) t^{k-l-1} \xi^{k-2l}e^{-t\xi^2}=0 $ for $t=\frac{k-l}{\xi^2}$.
Hence when $\vert\xi\vert\geqslant a$, $$\sup_{t\in [0,a]} (1+\vert\xi\vert)^k\vert(t^{k-l}\xi^{k-2l})\vert e^{-t\xi^2} \leqslant  (k-l)^{k-l}(1+\vert\xi\vert)^k\vert\xi\vert^{-k}\leqslant (k-l)^{k-l}(1+a^{-k})^k.$$
On the other hand, 
if $\vert \xi\vert\leqslant a$, $t\in [0,a]$, we find that
$(1+\vert\xi\vert)^k\vert\partial_\xi^ke^{-t\xi^2} \vert\leqslant  C(1+a)^k\sum_{0\leqslant l\leqslant \frac{k}{2}} a^{2k-3l}$.\\
Therefore, we showed that
$(1+\vert\xi\vert)^k\vert \partial^k_\xi e^{-t\xi^2} \vert\leqslant C_k $
uniformly on $t\in [0,a]$, hence $p_t\in S^0_{1,0}$ uniformly on $t\in [0,a]$.
We also have
for all $\delta,u>0$,
$t\leqslant\delta^{1+2u}$
implies that
$\sup_{\xi} \vert(1+\vert \xi\vert)^{-u} (e^{-t\xi^2}-1)\vert\leqslant \delta $ which means that 
$\sup_{\xi} \vert(1+\vert \xi\vert)^{-u} (e^{-t\xi^2}-1)\vert\rightarrow 0 $
when $t\rightarrow 0^+$ which implies the convergence $p_t\rightarrow 1$ in $S^{+0}_{1,0}$.
By a result of Strichartz~\cite[Thm 1.3 p.~296]{Taylorpsi}, 
\begin{equation}\label{e:convheat}
p_t(\sqrt{\Delta})=e^{-t\Delta}\underset{t\rightarrow 0^+}{\rightarrow} Id\text{ in }\Psi^{+0}_{1,0}(M).
\end{equation}
\end{proof}

\subsection{Proof of Lemma~\ref{l:raytofrechet}. }
\label{ss:appendixholo}

Without loss of generality we assume $0\in \Omega$ and we try to prove the Lemma in some neighborhood 
of $0\in \Omega$.
For every fixed $V\in \Omega$, and for every complex 
$z\in \mathbb{C}$ small enough, $\partial_z^nF_1(zV)=\partial_z^nF_2(zV)$ by assumption, therefore the uniqueness of the Taylor series and its convergence for analytic functions of one variable yields the identity
$$F_1(zV)=P(z,V)+F_2(zV)$$
where both sides are holomorphic germs in $z$ near $0\in \mathbb{C}$
and $P(z,V)$ is a polynomial in $z$ of degree $k-1$.  
The subtlety is that here we have an identity which holds true along 
every complex ray $\{zV, z\in \mathbb{C} \}$ in the open subset $\Omega\subset E$ 
and both sides are holomorphic functions
of one variable $z\in \mathbb{C}$.
We would like to deduce a similar identity without $z\in \mathbb{C}$ and were both sides are viewed as 
holomorphic functions on $C^\infty(End(E))$ in the sense of definition~\ref{d:analfunfrechet}. 

To finish the proof of the Lemma, we  
recall the definition of finitely holomorphic (also called G\^ateaux--holomorphic) functions which is the weakest
notion of holomorphicity in $\infty$--dimension~\cite[p.~54 def 2.2]{Dineen}: 
\begin{defi}[Finitely holomorphic functions]
Let $\Omega$ open in some Fr\'echet space
$E$ over $\mathbb{C}$.
A function $f:E\mapsto \mathbb{C}$ is said to be finitely holomorphic on $\Omega$ 
if for all $A\in \Omega$, every $B\in E$, $z\in \mathbb{C}\mapsto f(A+zB)$
is a \textbf{holomorphic germ} at $z=0$.
\end{defi}
Beware that finitely holomorphic maps 
are not necessarily continuous since any $\mathbb{C}$-linear map
$F:E\mapsto \mathbb{C}$
which is not even continuous is always finitely holomorphic.
The notion of holomorphicity from definition~\ref{d:analfunfrechet} is the strongest possible and Fr\'echet holomorphic functions
are automatically smooth hence $C^0$ unlike finitely holomorphic functions.

Our goal in this 
part is to recall the proof that
finitely analytic maps near $A$
which are \textbf{locally bounded} are analytic near $A$.
\begin{defi}[Local boundedness]
A map $f$ is locally bounded near $A$ if there is an open neighborhood
$U\subset E$ of $A$ and $0\leqslant M<+\infty$ such that $\vert f|_U\vert\leqslant M$.
\end{defi}
The proof is inspired from 
the thesis of Douady~\cite[Prop 2 p.~9]{douady1966probleme} and also~\cite[p.~57--58]{Dineen}.
\begin{prop}\label{p:douadyweakstrongholo}
Let $E$ be a Fr\'echet space and 
$F:E\mapsto \mathbb{C}$ finitely analytic on $\Omega\subset E$.
If $F$ is locally bounded at $A$, in particular
if 
on a ball $B(A,r)=\{B\text{ s.t. } \Vert B-A\Vert<r  \}$ for a continuous norm $\Vert .\Vert$ on $E$,
$\sup_{B(A,r)} \vert F \vert\leqslant M<+\infty$, 
then $F$ is Fr\'echet differentiable at $A$ at any order and can be identified with its Taylor
series near $a$:
$$ F(A+H)=\sum_{n=0}^\infty P_n(H) $$
where each $P_n$ is a continuous polynomial map homogeneous 
of degree $n$, the $P_n$ are uniquely determined 
by 
$P_n(h)=\frac{n!}{2i\pi} \int_\gamma F(A+\lambda h)\frac{d\lambda}{\lambda^{n=1}} $
and $\sum \Vert P_n\Vert \tilde{r}^n<+\infty$ for every $0<\tilde{r}<r$.

In particular $F$ smooth in some neighborhood of $A$.
\end{prop}
\begin{proof}
This proposition is well--known when $B$ has 
finite dimension and the expansion
$F(a+h)=\sum_n P_n(h)$ is given by the formula
$P_n(h)=\frac{1}{2\pi} \int_0^{2\pi} F(a+e^{i\theta}h)e^{-in\theta}d\theta $
where $\vert h\vert\leqslant r$ and then we keep this formula in the infinite
dimensional case. 
The integral is that 
of a continuous function (by finite analyticity)
hence is well--defined. 
If $F$ is bounded by $M$ on a ball of radius $r>0$ for the
continuous norm $\Vert .\Vert$ then
so is
$P_n $. To show that
$P_n$ is a homogeneous monomial, 
we follow Douady's approach by
setting $\tilde{P}_n(h_1,\dots,h_n)=\frac{1}{n!}\Delta_{h_1}\dots\Delta_{h_n}P_n$
where $\Delta_h$ is the finite difference operator $\Delta_hP(x)=\frac{1}{2}\left(P(x+h)-P(x-h)\right)$. In the finite dimensional case
$\tilde{P}_n$ is multilinear and it is the same in the infinite dimensional case since it only depends on the restriction
of $P_n$ to some finite dimensional subspace of $E$.
Hence $P_n(h)=\tilde{P}_n(h,\dots,h)$ for some symmetric multilinear map $\tilde{P}_n$.
From Cauchy's integral formula, we know that every 
$P_n$ is bounded by $M$ when $\Vert h \Vert\leqslant r$ which implies that
$\vert\tilde{P}_n(h_1,\dots,h_n)\vert \leqslant \frac{n^nM}{n! } \Vert h_1\Vert \dots \Vert h_n\Vert $
hence $\tilde{P}_n$ is continuous.
 From this it results that the series
$\sum_n P_n$ has normal convergence and the proposition is proved.
\end{proof}

Since both $V\in \Omega\mapsto F_1(V)$ and $V\in \Omega\mapsto F_2(V)$ are locally bounded near $V=0$ by holomorphicity of $F_1,F_2$,
the above
Proposition~\ref{p:douadyweakstrongholo} 
applied to the holomorphic function $F_1-F_2$ implies that we have the equality
\begin{equation}
F_1(V)=P(V)+F_2(V)
\end{equation}
for $V$ close enough to $0$ where $P$ is a \textbf{uniquely determined continuous polynomial function} of $V\in E$.

\end{document}